\documentclass[a4paper]{easychair}

\usepackage{multido}
\usepackage{subfig}
\usepackage{pstricks}
\usepackage{pst-node}
\usepackage{hyperref}
\usepackage[nompar,nosign,final]{commenting}

\newcommand\acmeasychair[2]{#2}

\usepackage{pifont}
\usepackage{xspace}
\usepackage{amsfonts}
\usepackage{mathtools}
\usepackage{stmaryrd}

\providecommand\easyhevea[2]{#1}

\easyhevea{
    \DeclarePairedDelimiter{\ceil}{\lceil}{\rceil}
    \DeclarePairedDelimiter{\sizeof}{|}{|}
}{
    \newcommand\ceil[1]{\left\lceil{#1}\right\rceil}
    \newcommand\sizeof[1]{\left|{#1}\right|}
    \renewcommand\xrightarrow[1]{---^{#1}\rightarrow}
}

\newcommand\etal{~\textit{et al.}\xspace}

\newcommand\hopda[1]{${#1}$-PDA}
\newcommand\bhopda[2]{$\tup{{#1}, {#2}}$-PDA}

\newcommand\lang{\mathcal{L}}
\newcommand\numof{m}
\newcommand\altnumof{l}
\newcommand\idxi{i}
\newcommand\idxj{j}
\newcommand\nats{\mathbb{N}}

\newcommand\brac[1]{\left({#1}\right)}
\newcommand\tup[1]{\brac{#1}}
\newcommand\ap[2]{{#1}\mathord{\brac{#2}}}
\newcommand\setcomp[2]{\left\{{#1}\ \left|\ {#2}\right.\right\}}
\newcommand\set[1]{\left\{{#1}\right\}}
\newcommand\sequence[2]{\brac{#1}_{#2}}

\newcommand\unbounded[2]{\ap{\text{Diagonal}_{#1}}{#2}}

\newcommand\treedom{D}
\newcommand\tree{T}
\newcommand\node{\eta}
\newcommand\treedir{\delta}
\newcommand\treelabels{\oalphabet}
\newcommand\treelabel{\och}
\newcommand\rootnode{\varepsilon}
\newcommand\treelabelling{\lambda}

\newcommand\treeap[2]{{#1}\mathord{\left[{#2}\right]}}
\newcommand\treesingle[1]{\tree\mathord{\left[{#1}\right]}}

\newcommand\pdaord{n}
\newcommand\opord{k}
\newcommand\pda{P}
\newcommand\och{a}
\newcommand\ochb{b}
\newcommand\ech{\varepsilon}
\newcommand\oalphabet{\Sigma}
\newcommand\cha{\gamma}
\newcommand\chb{\sigma}
\newcommand\salphabet{\Gamma}
\newcommand\controls{\mathcal{P}}
\newcommand\control{p}
\newcommand\rules{\mathcal{R}}
\newcommand\finals{\mathcal{F}}
\newcommand\op{o}
\newcommand\stackw{s}
\newcommand\configc{c}
\newcommand\configs{C}
\newcommand\pdrun{\rho}
\newcommand\pdruler{r}
\newcommand\word{w}
\newcommand{\controlinit}{\control_{\mathrm{in}}}
\newcommand{\chainit}{\cha_{\mathrm{in}}}
\newcommand\uniquefinalcontrol{\control_f}

\newcommand\allstacks[2]{\mathcal{S}^{#2}_{#1}}
\newcommand\config[2]{\langle {#1}, {#2} \rangle}
\newcommand\kstack[2]{\left[{#2}\right]_{#1}}
\newcommand\pdatran[1]{\xrightarrow{#1}}
\newcommand\ops[1]{\text{Ops}_{#1}}
\newcommand\topop[1]{\text{top}_{#1}}
\newcommand\push[1]{\text{push}_{#1}}
\newcommand\pop[1]{\text{pop}_{#1}}
\newcommand\rew[1]{\text{rew}_{#1}}
\newcommand\pdrule[5]{\tup{{#1}, {#2}} \xrightarrow{#3} \tup{{#5}, {#4}}}
\newcommand\apdrule[4]{\tup{{#1}, {#2}} \rightarrow \tup{{#4}, {#3}}}
\newcommand\pdrulet[6]{\tup{{#1}, {#2}, {#3}} \xrightarrow{#4} \tup{{#6}, {#5}}}

\acmeasychair{

    \newcommand\bigpdrulet[6]{%
        \begin{array}{c} %
            \tup{{#1}, {#2}, {#3}} \\ %
            \xdownarrow{
                \begin{array}{c} %
                    \\ %
                    {#4} \\ %
                    \\ %
                \end{array} %
            } \\ %
            \tup{{#6}, {#5}} %
        \end{array} %
    }

    \newcommand\bigpdrule[5]{%
        \begin{array}{c} %
            \tup{{#1}, {#2}} \\ %
            \xdownarrow{
                \begin{array}{c} %
                    \\ %
                    {#3} \\ %
                    \\ %
                \end{array} %
            } \\ %
            \tup{{#5}, {#4}} %
        \end{array} %
    }

}{
    \newcommand\bigpdrulet[6]{\tup{{#1}, {#2}, {#3}} \xrightarrow{#4} \tup{{#6}, {#5}}}
    \newcommand\bigpdrule[5]{\tup{{#1}, {#2}} \xrightarrow{#3} \tup{{#5}, {#4}}}
}

\newcommand\toptest[1]{{#1}?}
\easyhevea{
    \newcommand\stackinit[2]{\left\llbracket {#2} \right\rrbracket_{#1}}
}{
    \newcommand\stackinit[2]{[[ {#2} ]]}
}

\newcommand\auta{A}
\newcommand\autb{B}
\newcommand\sastates{\mathcal{Q}}
\newcommand\sastate{q}
\newcommand\safinals{\sastates_F}
\newcommand\sadelta{\Delta}
\newcommand\saalphabet{\salphabet_{[]}}

\newcommand\sopen[1]{[_{#1}}
\newcommand\sclose[1]{]_{#1}}
\newcommand\satran[1]{\xrightarrow{#1}}

\newcommand\auttrue{\auta_{\mathtt{tt}}}

\newcommand\augmented[1]{#1_\och}

\newcommand\canpopaut[2]{\auta_{{#1}, {#2}}}
\newcommand\canpopautout[2]{\auta^\och_{{#1}, {#2}}}
\newcommand\prestar[2]{\ap{\text{Pre}^\ast_{#1}}{#2}}

\newcommand\finalcontrol{f}
\newcommand\siminitrules{\rules_{\text{init}}}
\newcommand\simsimrules{\rules_{\text{sim}}}
\newcommand\simfinrules{\rules_{\text{fin}}}
\newcommand\simpushrules{\rules_{\text{push}}}
\newcommand\simpoprules{\rules_{\text{pop}}}
\newcommand\simcontrolinit{\control^{-1}_{\text{in}}}

\newcommand\simulator[1]{#1_{-1}}
\newcommand\simcontrol[2]{\tup{{#1}, {#2}}}

\newcommand\treedecomp[1]{\ap{\text{Tree}}{#1}}
\newcommand\treescore[1]{\ap{\text{Score}}{#1}}
\newcommand\isout[1]{\overline{#1}}

\newcommand\downclosure[1]{\ap{\downarrow}{#1}}

\newcommand\numchs{\alpha}
\newcommand\csrank{\theta}

\newcommand\chcount[2]{\sizeof{#2}_{#1}}
\newcommand\canpopautbr[3]{\auta^{#3}_{{#1},{#2}}}
\newcommand\alternating[1]{#1_\diamond}
\newcommand\altrule[3]{\tup{{#1}, {#2}} \rightarrow {#3}}

\newcommand\treescorech[2]{\ap{\text{Score}_{#1}}{#2}}
\newcommand\simcontrolbr[4]{\tup{{#1}, {#2}, {#3}, {#4}}}
\newcommand\linearised[1]{\overline{#1}}
\newcommand\lincontrol[2]{\tup{{#1}, {#2}}}

\newcommand\outputs{O}
\newcommand\simbrrules{\rules_{\text{br}}}
\newcommand\branchchs{B}
\newcommand\controlx{x}
\newcommand\controly{y}
\newcommand\outputsx{X}
\newcommand\outputsy{Y}
\newcommand\fronlen{h}

\newcommand\pselipsenode[1]{\rnode{#1}{$\vdots$}}

\newcommand\optestsim{\text{Update}}
\newcommand\testaccepts{\text{Accepts}}

\newcommand\autfun[1]{\tau_{#1}}
\newcommand\autfuns[1]{\overline{\tau}_{#1}}
\newcommand\funcha[2]{\tup{{#1}, {#2}}}
\easyhevea{
    \newcommand\notest[1]{\hat{#1}}
}{
    \newcommand\notest[1]{#1^T}
}

\newtheorem{theorem}{Theorem}[section]
\newtheorem{definition}{Definition}[section]

\newtheorem{lemma}{Lemma}[section]
\newtheorem{property}{Property}[section]
\newtheorem{corollary}{Corollary}[section]

\newenvironment{namedtheorem}[2]{%
    \expandafter\gdef\csname reftheorem#1\endcsname{%
        Theorem~\ref{#1} (#2)%
    }%
    \begin{theorem}[{#2}] \label{#1}%
}{%
    \end{theorem}%
}
\newcommand\reftheorem[1]{\expandafter\csname reftheorem#1\endcsname}

\newenvironment{namedlemma}[2]{%
    \expandafter\gdef\csname reflemma#1\endcsname{%
        Lemma~\ref{#1} (#2)%
    }%
    \begin{lemma}[{#2}] \label{#1}%
}{%
    \end{lemma}%
}
\newcommand\reflemma[1]{\expandafter\csname reflemma#1\endcsname}

\newcommand\refproperty[1]{\expandafter\csname refproperty#1\endcsname}

\newenvironment{nameddefinition}[2]{%
    \expandafter\gdef\csname refdefinition#1\endcsname{%
        Definition~\ref{#1} (#2)%
    }%
    \begin{definition}[{#2}] \label{#1}%
}{%
    \end{definition}%
}
\newcommand\refdefinition[1]{\expandafter\csname refdefinition#1\endcsname}

\colorlet{defaultcol}{black!90!yellow}
\declareauthor{lo}{Luke}{red}
\authorcommand{lo}{comment}
\declareauthor{jk}{Jonathan}{blue!90!black}
\authorcommand{jk}{comment}
\declareauthor{mh}{Matt}{green!90!black}
\authorcommand{mh}{comment}

\title{Unboundedness and Downward Closures of Higher-Order Pushdown Automata}
\titlerunning{Unboundedness of HOPDA}

\author{Matthew Hague\inst{1}
        \and
        Jonathan Kochems\inst{2}
        \and
        C.-H. Luke Ong\inst{2}}
\authorrunning{M. Hague \and J. Kochems \and C.-H. L. Ong}

\institute{
    Royal Holloway, Univsersity of London
    \and
    Department of Computer Science, Oxford University
}

\begin{document}

    \maketitle

    \begin{abstract}
        
We show the diagonal problem for higher-order pushdown automata (HOPDA), and hence the simultaneous unboundedness problem, is decidable.
From recent work by Zetzsche this means that we can construct the downward closure of the set of words accepted by a given HOPDA.
This also means we can construct the downward closure of the Parikh image of a HOPDA.
Both of these consequences play an important r\^{o}le in verifying concurrent higher-order programs expressed as HOPDA or safe higher-order recursion schemes.

    \end{abstract}

\section{Introduction}

Recent work by Zetzsche~\cite{Zetzsche:2015} has given a new technique for computing the downward closure of classes of languages.
The downward closure $\downclosure{\lang}$ of a language $\lang$ is the set of all subwords of words in $\lang$
(e.g. $\och  \och$ is a subword of $\ochb \och  \ochb  \och \ochb$).
It is well known that the downward closure is regular for any language~\cite{Haines:1969}.
However, there are only a few classes of languages for which it is known how to compute this closure.
In general it is not possible to compute the downward closure since it would easily lead to a solution to the halting problem for Turing machines.

However, once a regular representation of the downward closure has been obtained, it can be used in all kinds of analysis, since regular languages are well behaved under all kinds of transformations.
For example, consider a system that waits for messages from a complex environment.
This complex environment can be abstracted by the downward closure of the messages it sends or processes it spawns.
This corresponds to a lossy system where some messages may be ignored (or go missing), or some processes may simply not contribute to the remainder of the execution.
In many settings -- e.g. the analysis of safety properties of certain kinds of systems -- unread messages or unscheduled processes do not effect the precision of the analysis.
Since many types of system permit synchronisation with a regular language, this environment abstraction can often be built into the system being analysed.

Many popular languages such as JavaScript, Python, Ruby, and even C++, include higher-order features -- which are increasingly important given the popularity of event-based programs and asynchronous programs based on a continuation or callback
style of programming.
Hence, the modelling of higher-order function calls is becoming key to analysing modern day programs.

A popular approach to verifying higher-order programs is that of \emph{recursion schemes} and several tools and practical techniques have been developed~\cite{Kobayashi:2009,Unno:2010,Kobayashi:2011,Kobayashi:2012,Neatherway:2012,Broadbent:2013,Broadbent:2013b,Ramsay:2014}.
Recursion schemes have an automaton model in the form of collapsible pushdown automata (CPDA)~\cite{Hague:2008} which generalises an order-$2$ model called 2-PDA with links~\cite{Aehlig:2005} or, equivalently,
panic automata~\cite{Knapik:2005}.
When these recursion schemes satisfy a syntactical condition called \emph{safety}, a restriction of CPDA called \emph{higher-order pushdown automata (HOPDA or \hopda{\pdaord} for order-$\pdaord$ HOPDA)} is sufficient~\cite{Maslov:1976,Knapik:2002}.
HOPDA can be considered an extension of pushdown automata to a ``stack of stacks'' structure.
It remains open as to whether CPDA are strictly more expressive than nondeterministic HOPDA when generating languages of words.
It is known that, at order 2, nondeterministic HOPDA and CPDA generate the same word languages~\cite{Aehlig:2005}.
However, there exists a 
language generated by a deterministic order-$2$ CPDA that cannot be generated by a deterministic HOPDA of any order~\cite{Parys:2012}.

It is well known that concurrency and first-order recursion very quickly leads to undecidability (e.g.~\cite{Ramalingam:2000}).
Hence, much recent research has focussed on decidable abstractions and restrictions (e.g.~\cite{Esparza:2000,Bouajjani:2005,Kahlon:2009,Lal:2009,Esparza:2011,Torre:2011,Madhusudan:2011,Cyriac:2012,Hague:2014}).
Recently, these results have been extended to concurrent versions of
CPDA and recursion schemes (e.g.~\cite{Seth:2009,Kobayashi:2013,Hague:2013,Penelle:2015}).
Many approaches rely on combining representations of the Parikh image of individual automata (e.g.~\cite{Esparza:2011,Hague:2012,Hague:2014}).
However, combining Parikh images of HOPDA quickly leads to undecidability (e.g.~\cite{Hague:2012}).
In many cases, the downward closure of the Parikh image is an adequate abstraction.

Computing downward closures appears to be a hard problem.
Recently Zetzsche introduced a new general technique for classes of automata effectively closed under rational transductions -- also referred to as a \emph{full trio}.
For these automata the downward closure is computable iff the \emph{simultaneous unboundedness problem (SUP)} is decidable.

\begin{definition}[SUP~\cite{Zetzsche:2015}]
    Given a language
    $\lang \subseteq \och_1^\ast \ldots \och_\numchs^\ast$
    does
    $\downclosure{\lang} = \och_1^\ast \ldots \och_\numchs^\ast$?
\end{definition}

\begin{theorem}
    \label{thm:zetzsche}
    \cite[Theorem 1]{Zetzsche:2015}
    Let $\mathcal{C}$ be class of languages that is a full trio.
    Then downward closures are computable for $\mathcal{C}$ if and only if the SUP is decidable for $\mathcal{C}$.
\end{theorem}

Zetzsche used this result to obtain the downward closure of languages definable by \hopda{2},
or equivalently, languages definable by \emph{indexed grammars}~\cite{Aho:1968}.
Moreover, for classes of languages closed under rational transductions, Zetzsche shows that the simultaneous unboundedness problem is decidable iff the \emph{diagonal problem} is decidable.
The diagonal problem was introduced by Czerwi\'nski and Martens~\cite{Czerwinski:2014}.
Intuitively, it is a relaxation of the SUP that is insensitive to the order the characters are output.
For a word $\word$, let
$\chcount{\och}{\word}$
be the number of occurrences of $\och$ in $\word$.

\begin{definition}[Diagonal Problem~\cite{Czerwinski:2014}]
    Given language $\lang$ we define
    \[
        \unbounded{\och_1, \ldots, \och_\numchs}{\lang} =
        \forall \numof .
        \exists \word \in \lang .
        \forall 1 \leq \idxi \leq \numchs .
        \chcount{\och_\idxi}{\word} \geq \numof \ .
    \]
    The diagonal problem asks if
    $\unbounded{\och_1, \ldots, \och_\numchs}{\lang}$
    holds of $\lang$.
\end{definition}

\changed[mh]{
    \begin{corollary}[Diagonal Problem and Downward Closures]
        \label{cor:diagdown}
        Let $\mathcal{C}$ be class of languages that is a full trio.
        Then downward closures are computable for $\mathcal{C}$ if and only if the diagonal problem is decidable for $\mathcal{C}$.
    \end{corollary}
    \begin{proof}
        The only-if direction follows from Theorem~\ref{thm:zetzsche} since given a language
        $\lang \subseteq \och_1^\ast \ldots \och_\numchs^\ast$
        the diagonal problem is immediately equivalent to the SUP.
        In the if direction, the result follows since $\lang$ satisfies the diagonal problem iff
        $\downclosure{\lang}$
        also satisfies the diagonal problem.
        Since the diagonal problem is decidable for regular languages and
        $\downclosure{\lang}$
        is regular, we have the result.
    \end{proof}
}

In this work, we generalise Zetzsche's result for \hopda{2} to the general case of \hopda{\pdaord}.
We show that the diagonal problem is decidable.
Since HOPDA are closed under rational transductions, we obtain decidability of the simultaneous unboundedness problem, and hence a method for constructing the downward closure of a language defined by a HOPDA.

\begin{corollary}[Downward Closures]
    Let $\pda$ be an \hopda{\pdaord}.
    The downward closure
    $\downclosure{\ap{\lang}{\pda}}$
    is computable.
\end{corollary}
\begin{proof}
    From Theorem~\ref{thm:diagonal-sim} (proved in the sequel), we know that the diagonal problem for HOPDA is decidable.
    Thus, using Corollary~\ref{cor:diagdown}, we can construct the downward closure of $\pda$.
\end{proof}

This result provides an abstraction upon which new results may be based.
It also has several immediate consequences:
\begin{enumerate}

\item
    decidability of separability by piecewise testable languages, \acmeasychair{}{which follows from} from Czerwi\'nski and Martens~\cite{Czerwinski:2014}, 

\item
    decidability of reachability for parameterised concurrent systems of HOPDA communicating asynchronously via a shared global register, from La Torre\etal~\cite{Torre:2015}, 

\item
    decidability of finiteness of a language defined by a HOPDA, and
\item
    computability of the downward closure of the Parikh image of a HOPDA.
\end{enumerate}

We present our decidability proof in two stages.  First we show how to decide
$\unbounded{\och}{\pda}$
for a single character and HOPDA $\pda$ in Sections~\ref{sec:technique} and \ref{sec:correctness}.
In Sections~\ref{sec:simultaneous}, \ref{sec:reduction-sim}, and~\ref{sec:correctness-sim} we generalise our techniques to the full diagonal problem.

In Section~\ref{sec:outline} we give an outline of the proof techniques for deciding
$\unbounded{\och}{\pda}$.
In short, the outermost stacks of an \hopda{\pdaord} are created and destroyed using $\push{\pdaord}$ and $\pop{\pdaord}$ operations.
These $\push{\pdaord}$ and $\pop{\pdaord}$ operations along a run of an \hopda{\pdaord} are ``well-bracketed'' (each $\push{\pdaord}$ has a matching $\pop{\pdaord}$ and these matchings don't overlap).
The essence of the idea is to take a standard tree decomposition of these well-bracketed runs and observe that each branch of such a tree can be executed by an \hopda{(\pdaord-1)}.
We augment this \hopda{(\pdaord-1)} with ``regular tests'' that allow it to know if, each time a branch is chosen, the alternative branch could have output some $\och$ characters.
If this is true, then the \hopda{(\pdaord-1)} outputs a single $\och$ to account for these missed characters.
We prove that, although the \hopda{(\pdaord-1)} outputs far fewer characters, it can still output an unbounded number iff the \hopda{\pdaord} could.
Hence, by repeating this reduction, we obtain a \hopda{1}, for which the diagonal problem is decidable since it is known how to compute their downward closures~\cite{Leeuwen:1978,Courcelle:1991}.

In Section~\ref{sec:outline-sim} we outline the generalisation of the proof to the full problem
$\unbounded{\och_1, \ldots, \och_\numchs}{\pda}$.
The key difficulty is that it is no longer enough for the \hopda{(\pdaord-1)} to follow only a single branch of the tree decomposition:
it may need up to one branch for each of the
$\och_1, \ldots, \och_\numchs$.
Hence, we define HOPDA that can output trees with a bounded number ($\numchs$) of branches.
We then show that our reduction can generalise to HOPDA outputting trees (relying essentially on the fact that the number of branches is bounded).


\section{Preliminaries}

\subsection{Downward Closures}

Given two words
$\word = \cha_1 \ldots \cha_\numof \in \oalphabet^\ast$
and
$\word' = \chb_1 \ldots \chb_\altnumof \in \oalphabet^\ast$
for some alphabet $\oalphabet$, we write
$\word \leq \word'$
iff there exist
$\idxi_1 < \ldots < \idxi_\numof$
such that for all
$1 \leq \idxj \leq \numof$
we have
$\cha_\idxj = \chb_{\idxi_\idxj}$.
Given a set of words
$\lang \subseteq \oalphabet^\ast$,
we denote its downward closure
$\downclosure{\lang} = \setcomp{\word}{\word \leq \word' \in \lang}$.

\subsection{Trees}

A $\treelabels$-labelled finite tree is a tuple
$\tree = \tup{\treedom, \treelabelling}$
where
    $\treelabels$ is a set of node labels, and
    $\treedom \subset \nats^\ast$
    is a finite set of nodes that is prefix-closed,
    that is, $\node \, \treedir \in \treedom$ implies $\node \in \treedom$,
    and
    $\treelabelling : \treedom \rightarrow \treelabels$
    is a function labelling the nodes of the tree.

We write $\rootnode$ to denote the root of a tree (the empty sequence).
We also write
\[
    \treeap{\treelabel}
           {\tree_1, \ldots, \tree_\numof}
\]
to denote the tree whose root node is labelled $\treelabel$ and has children
$\tree_1, \ldots, \tree_\numof$.
That is, we define
$\treeap{\treelabel}
        {\tree_1, \ldots, \tree_\numof}
 =
 \tup{\treedom', \treelabelling'}$
when for each $\treedir$ we have
$\tree_\treedir = \tup{\treedom_\treedir,
                       \treelabelling_\treedir}$
and
$\treedom' = \setcomp{\treedir\node}{\node \in \treedom_\treedir}
             \cup
             \set{\rootnode}$
and
\[
    \ap{\treelabelling'}{\node} =
    \begin{cases}
        \treelabel
        &
        \node = \rootnode
        \\

        \ap{\treelabelling_\treedir}{\node'}
        &
        \node = \treedir\node'
    \end{cases} \ .
\]
Also, let $\treesingle{\treelabel}$ denote the tree
$\tup{\set{\rootnode}, \treelabelling}$
where
$\ap{\treelabelling}{\rootnode} = \treelabel$.
A \emph{branch} in $T = \tup{\treedom, \treelabelling}$ is a sequence of nodes of $T$, $\node_1 \cdots \node_n$, such that $\node_1 = \epsilon$, $\node_n = \treedir_1 \, \treedir_2 \cdots \treedir_{n-1}$ is maximal in $\treedom$, and $\node_{j+1} = \node_j \, \treedir_{j}$ for each $1 \leq j \leq n-1$.

\subsection{HOPDA}

HOPDA are a generalisation of pushdown systems to a stack-of-stacks structure.
An order-$\pdaord$ stack is a stack of order-$(\pdaord-1)$ stacks.
An order-$\pdaord$ push operation pushes a new order-$(\pdaord-1)$ stack onto the stack that is a copy of the existing topmost order-$(\pdaord-1)$ stack.
Rewrite operations update the character that is at the top of the topmost stacks.

\begin{definition}[Order-$\pdaord$ Stacks]
    The set of order-$\pdaord$ stacks
    $\allstacks{\pdaord}{\salphabet}$
    over a given stack alphabet $\salphabet$ is defined inductively as follows.
    \[
        \begin{array}{rcl} %
            \allstacks{0}{\salphabet} %
            &=& %
            \salphabet \\ %
            \allstacks{\opord+1}{\salphabet} %
            &=& %
            \setcomp{ %
                \kstack{\opord+1}{\stackw_1 \ldots \stackw_\numof} %
            }{ %
                \forall \idxi . %
                    \stackw_\idxi \in \allstacks{\opord}{\salphabet} %
            } \ . %
        \end{array} %
    \]
\end{definition}
Stacks are written with the top part of the stack to the left.
We define several operations.
\[
    \begin{array}{rcll} %
        \ap{\topop{\opord}}{\kstack{\opord}{\stackw_1 \ldots \stackw_\numof}} %
        &=& %
        \stackw_1 %
        \\ %
        \ap{\topop{\opord}}{\kstack{\pdaord}{\stackw_1 \ldots \stackw_\numof}} %
        &=& %
        \ap{\topop{\opord}}{\stackw_1} %
        & %
        \pdaord > \opord %
        \\ %
        \\ %
        \ap{\rew{\cha}}{\kstack{1}{\cha_1 \ldots \cha_\numof}} %
        &=& %
        \kstack{1}{\cha \; \cha_2 \ldots \cha_\numof} %
        \\ %
        \ap{\rew{\cha} \; }{\kstack{\pdaord}{\stackw_1 \ldots \stackw_\numof}} %
        &=& %
        \kstack{\pdaord}{\ap{\rew{\cha}}{\stackw_1} \; %
                         \stackw_2 \ldots \stackw_\numof} %
        & %
        \pdaord > 1 %
        \\ %
        \\ %
        \ap{\push{\opord}}{\kstack{\opord}{\stackw_1 \ldots \stackw_\numof}} %
        &=& %
        \kstack{\opord}{\stackw_1 \; \stackw_1 \ldots \stackw_\numof} %
        \\ %
        \ap{\push{\opord}}{\kstack{\pdaord}{\stackw_1 \ldots \stackw_\numof}} %
        &=& %
        \kstack{\pdaord}{\ap{\push{\opord}}{\stackw_1} %
                       \;   \stackw_2, \ldots, \stackw_\numof} %
        & %
        \pdaord > \opord %
        \\ %
        \\ %
        \ap{\pop{\opord}}{\kstack{\opord}{\stackw_1 \ldots \stackw_\numof}} %
        &=& %
        \kstack{\opord}{\stackw_2 \ldots \stackw_\numof} %
        \\ %
        \ap{\pop{\opord}}{\kstack{\pdaord}{\stackw_1 \ldots \stackw_\numof}} %
        &=& %
        \kstack{\pdaord}{\ap{\pop{\opord}}{\; \stackw_1} %
                        \; \stackw_2, \ldots, \stackw_\numof} %
        & %
        \pdaord > \opord %
        \\ %
    \end{array} %
\]
and set
\[
    \ops{\pdaord} =
    \setcomp{\rew{\cha}}{\cha \in \salphabet}
    \cup
    \setcomp{\push{\opord}, \pop{\opord}}
            {1 \leq \opord \leq \pdaord}
\]
to be the set of order-$\pdaord$ stack operations.

For example
\[
    \begin{array}{rcl} %
        \ap{\push{2}}{ %
            \kstack{2}{ %
                \kstack{1}{\cha \; \chb} %
            } %
        } %
        &=& %
        \kstack{2}{ %
            \kstack{1}{\cha \; \chb} %
            \kstack{1}{\cha \; \chb} %
        } %
        \\ %
        \ap{\rew{\chb}}{ %
            \kstack{2}{ %
                \kstack{1}{\cha \; \chb} %
                \kstack{1}{\cha \; \chb} %
            } %
        } %
        &=& %
        \kstack{2}{ %
            \kstack{1}{\chb \; \chb} %
            \kstack{1}{\cha \; \chb} %
        } \ . %
    \end{array} %
\]

\begin{definition}[HOPDA or \hopda{\pdaord}]
    An order-$\pdaord$ \emph{higher order pushdown automaton (HOPDA or \hopda{\pdaord})} is given by a tuple
    $\tup{\controls,
          \oalphabet,
          \salphabet,
          \rules,
          \finals,
          \controlinit,
          \chainit}$
    where
        $\controls$ is a finite set of control states,
        $\oalphabet$ is a finite output alphabet (that contains the empty word character $\epsilon$),
        $\salphabet$ is a finite stack alphabet,
        $\rules \subseteq \controls \times
                          \salphabet \times
                          \oalphabet \times
                          \ops{\pdaord} \times
                          \controls$ is a set of transition rules,
        $\finals$ is a set of accepting control states,
        $\controlinit \in \controls$ is the initial control state, and
        $\chainit \in \salphabet$ is the initial stack character.
\end{definition}

We write
$\pdrule{\control}
        {\cha}
        {\och}
        {\op}
        {\control'}$
for a rule
$\tup{\control, \cha, \och, \op, \control'} \in \rules$.

A configuration of an \hopda{\pdaord} is a tuple
$\config{\control}{\stackw}$
where
$\control \in \controls$
and $\stackw$ is an order-$\pdaord$ stack over $\salphabet$.
We have a transition
$\config{\control}{\stackw} \pdatran{\och} \config{\control'}{\stackw'}$
whenever we have
$\pdrule{\control}
        {\cha}
        {\och}
        {\op}
        {\control'}$,
        $\ap{\topop{1}}{\stackw} = \cha$,
and
$\stackw' = \ap{\op}{\stackw}$.

A run over a word
$\word \in \oalphabet^\ast$
is a sequence of configurations
$\configc_0 \pdatran{\och_1} \cdots \pdatran{\och_\numof} \configc_\numof$
such that the word
$\och_1\ldots\och_\numof$
is $\word$.
It is an accepting run if
$\configc_0 = \config{\controlinit}{\stackinit{\pdaord}{\chainit}}$
---
where we write
$\stackinit{\pdaord}{\cha}$
for
$\kstack{\pdaord}{\cdots\kstack{1}{\cha}\cdots}$
---
and where
$\configc_\numof = \config{\control}{\stackw}$
with
$\control \in \finals$.
Furthermore, for a set of configurations $\configs$, we define
\[
    \prestar{\pda}{\configs}
\]
to be the set of configurations $\configc$ such that there is a run over some word from $\configc$ to
$\configc' \in \configs$.
When $\configs$ is defined as the language of some automaton $\auta$ accepting configurations, we abuse notation and write
$\prestar{\pda}{\auta}$
instead of
$\prestar{\pda}{\ap{\lang}{\auta}}$.

For convenience, we sometimes allow a set of characters to be output instead of only one.
This is to be interpreted as outputing each of the characters in the set once (in some arbitrary order).
We also allow sequences of operations
$\op_1; \ldots; \op_\numof$
in the rules instead of single operations.
When using sequences we allow a test operation
$\toptest{\cha}$
that only allows the sequence to proceed if the $\topop{1}$ character of the stack is $\cha$.
All of these extensions can be encoded by introducing intermediate control states.

\subsubsection{Regular Sets of Stacks}

We will need to represent sets of stacks.
To do this we will use automata to recognise stacks.
We define the stack automaton model of Broadbent\etal~\cite{Broadbent:2010} restricted to HOPDA rather than CPDA.
We will sometimes call these \emph{bottom-up stack automata} or simply \emph{automata}.
The automata operate over stacks interpreted as words, hence the opening and closing braces of the stacks appear as part of the input.
We annotate these braces with the order of the stack the braces belong to.
Let
$\saalphabet = \set{\sopen{\pdaord-1},\ldots,\sopen{1},
                    \sclose{1},\ldots,\sclose{\pdaord-1}}
               \uplus
               \salphabet$.
Note, we don't include
$\sopen{\pdaord}, \sclose{\pdaord}$
since these appear exclusively at the start and end of the stack.

\begin{definition}[Bottom-up Stack Automata]
    A tuple $\auta$ is a \emph{bottom-up stack automaton} when $\auta$ is
    $\tup{\sastates, \salphabet, \sastate_{\mathrm{in}}, \safinals, \sadelta}$
    where
        $\sastates$ is a finite set of states,
        $\salphabet$ is a finite input alphabet,
        $\sastate_{\mathrm{in}} \in \sastates$ is the initial state and
        $\sadelta : \brac{\sastates \times \salphabet}
                    \rightarrow
                    \sastates$
        is a deterministic transition function.
\end{definition}

Representing higher order stacks as a linear word graph, where the start of an order-$\opord$ stack is an edge labelled $\sopen{\opord}$ and the end of an order-$\opord$ stack is an edge labelled $\sclose{\opord}$, a run of a bottom-up stack automaton is a labelling of the nodes of the graph with states in $\sastates$ such that
\begin{enumerate}
\item
    the rightmost (final) node is labelled by $\sastate_{\mathrm{in}}$, and

\item
    whenever we have for any
    $\cha \in \saalphabet$,
    and pair of labelled nodes with an edge
    $\sastate \satran{\cha} \sastate'$
    then
    $\sastate = \ap{\sadelta}{\sastate', \cha}$.
\end{enumerate}
The run is accepting if the leftmost (initial) node is labelled by
$\sastate \in \safinals$.
An example run over the word graph representation of
$\kstack{3}{
    \kstack{2}{
        \kstack{1}{\cha \; \chb}\kstack{1}{\chb}
    }
    \kstack{2}{
        \kstack{1}{\chb}
    }
}$
is given in Figure~\ref{fig:word-graph-run}.

Let $\ap{\lang}{\auta}$ be the set of stacks with accepting runs of $\auta$.
Sometimes, for convenience, if we have a configuration
$\configc = \config{\control}{\stackw}$
of a HOPDA, we will write
$\configc \in \ap{\lang}{\auta}$
when
$\stackw \in \ap{\lang}{\auta}$.

\begin{figure*}
    \centering
    \psset{nodesep=1ex}
    \vspace{3ex}
    \begin{psmatrix}[colsep=4ex]
        \rnode{N1}{$\sastate_f$} &
        \rnode{N2}{$\sastate_{13}$} &
        \rnode{N3}{$\sastate_{12}$} &
        \rnode{N4}{$\sastate_{11}$} &
        \rnode{N5}{$\sastate_{10}$} &
        \rnode{N6}{$\sastate_9$} &
        \rnode{N7}{$\sastate_8$} &
        \rnode{N8}{$\sastate_7$} &
        \rnode{N9}{$\sastate_6$} &
        \rnode{N10}{$\sastate_5$} &
        \rnode{N11}{$\sastate_4$} &
        \rnode{N12}{$\sastate_3$} &
        \rnode{N13}{$\sastate_2$} &
        \rnode{N14}{$\sastate_1$} &
        \rnode{N15}{$\sastate_{\mathrm{in}}$}
        
        \ncline{N1}{N2}\naput{$\sopen{2}$}
        \ncline{N2}{N3}\naput{$\sopen{1}$}
        \ncline{N3}{N4}\naput{$\cha$}
        \ncline{N4}{N5}\naput{$\chb$}
        \ncline{N5}{N6}\naput{$\sclose{1}$}
        \ncline{N6}{N7}\naput{$\sopen{1}$}
        \ncline{N7}{N8}\naput{$\chb$}
        \ncline{N8}{N9}\naput{$\sclose{1}$}
        \ncline{N9}{N10}\naput{$\sclose{2}$}
        \ncline{N10}{N11}\naput{$\sopen{2}$}
        \ncline{N11}{N12}\naput{$\sopen{1}$}
        \ncline{N12}{N13}\naput{$\chb$}
        \ncline{N13}{N14}\naput{$\sclose{1}$}
        \ncline{N14}{N15}\naput{$\sclose{2}$}
    \end{psmatrix}
    \caption{
        \label{fig:word-graph-run}
        A run over
        $\kstack{3}{
            \kstack{2}{
                \kstack{1}{\cha \; \chb}\kstack{1}{\chb}
            }
            \kstack{2}{
                \kstack{1}{\chb}
            }
        }$
    }
\end{figure*}

\section{The Single Character Case}
\label{sec:technique}

We assume $\oalphabet = \set{\och, \ech}$ and use $\ochb$ to range over $\oalphabet$.
This can be obtained by simply replacing all other characters with $\ech$.
We also assume that all rules of the form
$\pdrule{\control}{\cha}{\ochb}{\op}{\control'}$
with $\op = \push{\pdaord}$ or $\op = \pop{\pdaord}$ have $\ochb = \ech$.
We can enforce this using intermediate control states to first apply $\op$ in one step, and then in another output $\ochb$
(the stack operation on the second step will be $\rew{\cha}$ where $\cha$ is the current top character).
We start with an outline of the proof, and then explain each step in detail.

For convenience, we assume acceptance is by reaching a unique control state in $\finals$ with an empty stack (i.e. the lowermost stack was removed with a $\pop{\pdaord}$ and
$\finals = \set{\uniquefinalcontrol}$).
This can easily be obtained by adding a rule to a new accepting state whenever we have a rule leading to a control state in $\finals$.
From this new state we can loop and perform $\pop{\pdaord}$ operations until the stack is empty.

\subsection{Outline of Proof}
\label{sec:outline}

The 
approach is to take an \hopda{\pdaord} $\pda$ and produce an \hopda{(\pdaord-1)} $\simulator{\pda}$ that satisfies the diagonal problem iff $\pda$ does.
The idea behind this reduction is that an (accepting) run
of $\pda$ can be decomposed into a tree with out-degree at most 2:
each $\push{\pdaord}$ has a matching $\pop{\pdaord}$ that brings the stack back to be the same as it was before the $\push{\pdaord}$;
we cut the run at the $\pop{\pdaord}$ and hang the tail
next to the $\push{\pdaord}$ and repeat this to form a tree from a run.
This is illustrated in Figure~\ref{fig:tree-decomp} where nodes are labelled by their configurations, and the $\push{\pdaord}$ and $\pop{\pdaord}$ points are marked.
The dotted arcs connect nodes matched by their pushes and pops -- these nodes have the same stacks.
Notice that at each branching point, the left and right subtrees start with the same order-$(\pdaord-1)$ stacks on top.
Notice also that for each branch, none of its transitions remove the topmost order-$(\pdaord-1)$ stack.
Hence, we can produce an \hopda{(\pdaord-1)} that picks a branch of this tree decomposition to execute and only needs to keep track of the topmost order-$(\pdaord-1)$ stack of the \hopda{\pdaord}.
When picking a branch to execute, the \hopda{(\pdaord-1)} outputs a single $\och$ if the branch not chosen could have output some $\och$ characters.
We prove that this is enough to maintain unboundedness.

\begin{figure*}
\easyhevea{}{\begin{figure*}}
    \newcommand\ci{$\config{\control_1}{\kstack{\pdaord}{\stackw_1}}$}
    \newcommand\cii{$\config{\control_2}{\kstack{\pdaord}{\stackw_2}}$}
    \newcommand\ciii{$\config{\control_3}{\kstack{\pdaord}{\stackw_2 \; \stackw_2}}$}
    \newcommand\civ{$\config{\control_4}{\kstack{\pdaord}{\stackw_3 \; \stackw_2}}$}
    \newcommand\cv{$\config{\control_5}{\kstack{\pdaord}{\stackw_3  \; \stackw_3  \; \stackw_2}}$}
    \newcommand\cvi{$\config{\control_6}{\kstack{\pdaord}{\stackw_4  \; \stackw_3  \; \stackw_2}}$}
    \newcommand\cvii{$\config{\control_7}{\kstack{\pdaord}{\stackw_3  \; \stackw_2}}$}
    \newcommand\cviii{$\config{\control_8}{\kstack{\pdaord}{\stackw_2}}$}
    \newcommand\cix{$\config{\control_9}{\kstack{\pdaord}{\stackw_5}}$}
    \centering
    \acmeasychair{
        \psset{nodesep=1ex,rowsep=5ex,colsep=0ex}
    }{
        \psset{nodesep=1ex,rowsep=6ex,colsep=0ex}
    }
    \subfloat[][a run of $\pda$ with $\push{\pdaord}$s and $\pop{\pdaord}$s marked.]{
        \parbox[t]{37ex}{
            \centering
            \begin{psmatrix}
                \rnode{N1}{\ci} \\
                \rnode{N2}{\cii} \\
                \rnode{N3}{\ciii} \\
                \rnode{N4}{\civ} \\
                \rnode{N5}{\cv} \\
                \rnode{N6}{\cvi} \\
                \rnode{N7}{\cvii} \\
                \rnode{N8}{\cviii} \\
                \rnode{N9}{\cix}

                \ncline{N1}{N2}
                \ncline{N2}{N3}\naput{$\push{\pdaord}$}
                \ncline{N3}{N4}
                \ncline{N4}{N5}\naput{$\push{\pdaord}$}
                \ncline{N5}{N6}
                \ncline{N6}{N7}\naput{$\pop{\pdaord}$}
                \ncline{N7}{N8}\naput{$\pop{\pdaord}$}
                \ncline{N8}{N9}
                \ncarc[arcangle=-75,linestyle=dotted]{N4}{N7}
                \ncarc[arcangle=-75,linestyle=dotted]{N2}{N8}
            \end{psmatrix}
        }
    }
    \subfloat[][The tree decomposition of the run]{
        \parbox[t]{50ex}{
            \centering
            \begin{psmatrix}
                                 &                   & \rnode{N1}{\ci} \\
                                 &                   & \rnode{N2}{\cii} \\
                                 & \rnode{N3}{\ciii} &                   & \rnode{N8}{\cviii} \\
                                 & \rnode{N4}{\civ}  &                   & \rnode{N9}{\cix} \\
                \rnode{N5}{\cv}  &                   & \rnode{N7}{\cvii} & \\
                \rnode{N6}{\cvi} & \\
                \\
                \\
                \ncline{N1}{N2}
                \ncline{N2}{N3}
                \ncline{N3}{N4}
                \ncline{N4}{N5}
                \ncline{N5}{N6}
                \ncline{N4}{N7}
                \ncline{N2}{N8}
                \ncline{N8}{N9}
            \end{psmatrix}
        }
    }
\easyhevea{}{\end{figure*}}
    \caption{\label{fig:tree-decomp}Tree decompositions of runs.}
\end{figure*}

In more detail, we perform the following steps.
\begin{enumerate}
\item
    Instrument $\pda$ to record whether an $\och$ character has been output.
    Then, using known reachability results,
    obtain regular sets of configurations from which the current $\topop{\pdaord}$ stack can be popped, and
        moreover, we can know
    whether an $\och$ is output on the way.
    These tests can be seen as a generalisation of pushdown systems with regular tests introduced by Esparza\etal~\cite{Esparza:2003}.

\item
    From an \hopda{\pdaord} $\pda$, we define an \hopda{(\pdaord-1)} with tests
        $\simulator{\pda}$ and then an \hopda{(\pdaord-1)} $\pda'$ such that
        \[
            \unbounded{\och}{\pda} \iff \unbounded{\och}{\pda'}
            \ .
        \]
    The tests will be used to check the branches of the tree decomposition not explored by $\simulator{\pda}$.

\item
    By repeated applications of the above reduction, we obtain an \hopda{1} $\pda$ for which $\unbounded{\och}{\pda}$ is decidable since the downward closure of a context-free grammar (equivalent to \hopda{1}) is computable~\cite{Leeuwen:1978,Courcelle:1991} and this is equivalent to the diagonal problem.
\end{enumerate}
The \hopda{(\pdaord-1)} with tests $\simulator{\pda}$ will simulate the \hopda{\pdaord} $\pda$ in the following way.
\begin{itemize}
\item
    All operations except for $\push{\pdaord}$ and $\pop{\pdaord}$ will be simulated directly.

\item
    In lieu of performing a $\push{\pdaord}$, $\simulator{\pda}$ will choose to simulate the run of $\pda$ between the push and its corresponding
    $\pop{\pdaord}$, or the run of $\pda$ after the corresponding
    $\pop{\pdaord}$ has taken place.
    \begin{itemize}
    \item
        Tests will be used to determine which control state could appear after the corresponding $\pop{\pdaord}$.

    \item
        If the part of the run not being simulated output some $\och$s, then $\pda$ will output a single $\och$ in place of the omitted $\och$s.
    \end{itemize}
\end{itemize}
Although $\simulator{\pda}$ will output far fewer $\och$ characters than $\pda$ (since it does not execute the full run), we show that it still outputs enough $\och$s for the language to remain unbounded.

We thus have the following theorem.

\begin{namedtheorem}{thm:diagonal}{Decidability of the Diagonal Problem}
    Given an \hopda{\pdaord} $\pda$ and output character $\och$, whether
    $\unbounded{\och}{\pda}$
    holds is decidable.
\end{namedtheorem}
\begin{proof}
    We construct via Lemma~\ref{lem:reduction} an \hopda{(\pdaord-1)} $\pda'$ such that
    $\unbounded{\och}{\pda}$
    iff
    $\unbounded{\och}{\pda'}$.
    We repeat this step until we have a \hopda{1}.
    It is known that
    $\unbounded{\och}{\pda}$
    for an \hopda{1} is decidable
        since it is possible to compute the downward closure~\cite{Leeuwen:1978,Courcelle:1991}.
\end{proof}

\subsection{HOPDA with Tests}
\label{sec:tests}

When executing a branch of the tree decomposition, to be able to ensure the branch is correct and whether we should output an extra $\och$ we need to know how the system could have behaved on the skipped branch.
To do this we add tests to the HOPDA that allow it to know if the current stack belongs to a given regular set.
We show in the following sections that the properties required for our reduction can be represented as regular sets of stacks.
Although we take Broadbent\etal's logical reflection as the basis of our proof, HOPDA with tests can be seen as a generalisation of pushdown systems with regular valuations due to Esparza\etal~\cite{Esparza:2003}.

\begin{definition}[\hopda{\pdaord} with Tests]
\label{def:hopda-tests}
    Given a sequence of automata
    $\auta_1, \ldots, \auta_\numof$
    recognising regular sets of stacks, an \emph{\hopda{\pdaord} with tests} is a tuple
    $\pda = \tup{\controls,
                 \oalphabet,
                 \salphabet,
                 \rules,
                 \finals,
                 \controlinit,
                 \chainit}$
    where
    $\controls, \oalphabet, \salphabet, \finals, \controlinit$,
    and $\chainit$ are as in HOPDA, and
    \[
        \rules \subseteq \controls \times
                         \salphabet \times
                         \set{\auta_1, \ldots, \auta_\numof} \times
                         \oalphabet \times
                         \ops{\pdaord} \times
                         \controls
    \]
    is a set of transition rules.
\end{definition}
We write
$\pdrulet{\control}{\cha}{\auta_\idxi}{\ochb}{\op}{\control'}$
for
$\tup{\control, \cha, \auta_\idxi, \ochb, \op, \control'} \in \rules$.
We have a transition
$\config{\control}{\stackw} \pdatran{\ochb} \config{\control'}{\stackw'}$
whenever
$\pdrulet{\control}{\cha}{\auta_\idxi}{\ochb}{\op}{\control'} \in \rules$
and
$\ap{\topop{1}}{\stackw} = \cha$,
$\stackw \in \ap{\lang}{\auta_\idxi}$,
and
$\stackw' = \ap{\op}{\stackw}$.

We know from Broadbent\etal that these tests do not add any extra power to HOPDA.
Intuitively, we can embed runs of the automata into the stack during runs of the HOPDA.

\begin{namedtheorem}{thm:no-tests}{Removing Tests}
    \cite[Theorem 3 (adapted)]{Broadbent:2010}
    For every \hopda{\pdaord} with tests $\pda$, we can compute an \hopda{\pdaord} $\pda'$ with
    $\ap{\lang}{\pda} = \ap{\lang}{\pda'}$.
\end{namedtheorem}
\begin{proof}
    This is a straightforward adaptation of Broadbent\etal~\cite{Broadbent:2010}.
    A more general theorem is proved in Theorem~\ref{thm:no-tests-sim}.
\end{proof}

\subsubsection{Marking Outputs}

When the HOPDA is in a configuration of the form $\config{\control}{\kstack{\pdaord}{\stackw}}$
-- i.e. the outermost stack contains only a single order-$(\pdaord-1)$ stack --
we require the HOPDA to be able to know whether,
\begin{itemize}
\item
    for a given $\control_1$ and $\control_2$, there is a run from $\config{\control_1}{\kstack{\pdaord}{\stackw}}$ to $\config{\control_2}{\kstack{\pdaord}{}}$
    (that is, the HOPDA empties the stack), and
\item
    whether, during the run, an $\och$ is output.
\end{itemize}

Given $\pda$, we first augment $\pda$ to record whether an $\och$ has been produced.
This can be done simply by recording in the control state whether $\och$ has been output.

\begin{definition}[$\augmented{\pda}$]
    Given
    $\pda = \tup{\controls,
                 \oalphabet,
                 \salphabet,
                 \rules,
                 \finals,
                 \controlinit,
                 \chainit}$ we define
    \[
        \augmented{\pda} =
        \tup{
            \controls \cup \augmented{\controls},
            \oalphabet,
            \salphabet,
            \rules \cup \augmented{\rules},
            \finals \cup \augmented{\finals},
            \controlinit,
            \chainit
        }
    \]
    where
    \[
        \begin{array}{rcl}
            \augmented{\controls}
            &=&
            \setcomp{\augmented{\control}}{\control \in \controls}
            \\

            \augmented{\rules}
            &=&
            \setcomp{\pdrule{\augmented{\control}}
                            {\cha}
                            {\ochb}
                            {\op}
                            {\augmented{\control'}}}
                    {\pdrule{\control}
                            {\cha}
                            {\ochb}
                            {\op}
                            {\control'} \in \rules}\ \cup
            \\
            &&
            \setcomp{\pdrule{\control}
                            {\cha}
                            {\och}
                            {\op}
                            {\augmented{\control'}}}
                    {\pdrule{\control}
                            {\cha}
                            {\och}
                            {\op}
                            {\control'} \in \rules}
            \\

            \augmented{\finals}
            &=&
            \setcomp{\augmented{\control}}{\control \in \finals}
        \end{array}
    \]
\end{definition}

It is easy to see that $\pda$ and $\augmented{\pda}$ accept the same languages, and that $\augmented{\pda}$ is only in a control state $\augmented{\control}$ if an $\och$ has been output.

\subsubsection{Building the Automata}

Fix some
$\pda = \tup{\controls,
             \oalphabet,
             \salphabet,
             \rules,
             \finals}$
and
$\augmented{\pda} = \tup{\augmented{\controls},
                         \oalphabet,
                         \salphabet,
                         \augmented{\rules},
                         \augmented{\finals}}$.
To obtain a HOPDA with tests, we need, for each $\control_1, \control_2 \in \controls$ the following automata.
Note, we define these automata to accept order-$(\pdaord-1)$ stacks since they will be used in an \hopda{(\pdaord-1)} with tests.
\begin{enumerate}
\item
    $\canpopaut{\control_1}{\control_2}$
    accepting all stacks $\stackw$ such that there is a run of $\pda$ from
    $\config{\control_1}{\kstack{\pdaord}{\stackw}}$
    to
    $\config{\control_2}{\kstack{\pdaord}{}}$,

\item
    $\canpopautout{\control_1}{\control_2}$
    accepting all stacks $\stackw$ such that there is a run of $\pda$ from
    $\config{\control_1}{\kstack{\pdaord}{\stackw}}$
    to
    $\config{\control_2}{\kstack{\pdaord}{}}$
    that outputs at least one $\och$.
\end{enumerate}
To do this we will use a reachability result due to Broadbent\etal that appeared in ICALP 2012~\cite{Broadbent:2012}.
This result uses an automata representation of sets of configurations.
However, these automata are slightly different in that they read full configurations ``top down'', whereas the automata of \reftheorem{thm:no-tests} read only stacks ``bottom up''.

It is known that these two representations are effectively equivalent, and that both form an effective boolean algebra~\cite{Broadbent:2010,Broadbent:2012}.
In particular, for a top-down automaton $\auta$ and a control state $\control$ we can build a bottom-up stack automaton $\autb$ such that
$\config{\control}{\stackw} \in \ap{\lang}{\auta}$
iff
$\stackw \in \ap{\lang}{\autb}$ and vice versa.
We recall the reachability result.

\begin{theorem}
\label{thm:regreach}
    \cite[Theorem~1 (specialised)]{Broadbent:2012}
    Given an HOPDA $\pda$ and a top-down automaton $\auta$, we can construct an
    automaton $\auta'$ accepting
    $\prestar{\pda}{\auta}$.
\end{theorem}

Let $\auta_{\control, \cha}$ be a top-down automaton accepting configurations of the form $\config{\control}{\kstack{\pdaord}{\stackw}}$ where $\ap{\topop{1}}{\stackw} = \cha$.
Next, let
\[
    \auta_{\control} =
    \bigcup\limits_{
        \pdrule{\control'}{\cha}{\ech}{\pop{\pdaord}}{\control} \in \rules
    }
    \auta_{\control', \cha}
\]
and
\[
    \auta^\och_\control =
    \bigcup\limits_{
        \pdrule{\control'}{\cha}{\ech}{\pop{\pdaord}}{\control} \in \rules
    }
    \auta_{\augmented{\control'}, \cha}
\]
I.e. $\auta_\control$ and $\auta^\och_\control$ accept configurations of $\augmented{\pda}$ from which it is possible to perform a $\pop{\pdaord}$ operation to $\control$ and reach the empty stack.

\begin{definition}[$\canpopaut{\control_1}{\control_2}$
                   and
                   $\canpopautout{\control_1}{\control_2}$]
    Using the preceding notation, given $\control_1$ and $\control_2$ we define bottom-up automata
    \begin{itemize}
    \item
        $\canpopaut{\control_1}{\control_2}$ where
        $
            \ap{\lang}{\canpopaut{\control_1}{\control_2}} =
            \setcomp{
                \stackw
            }{
                \config{\control_1}{\kstack{\pdaord}{\stackw}} \in
                \prestar{\pda}{\auta_{\control_2}}
            } \ .
        $

    \item
        $\canpopautout{\control_1}{\control_2}$ where
        $
            \ap{\lang}{\canpopautout{\control_1}{\control_2}} =
            \setcomp{
                \stackw
            }{
                \config{\control_1}{\kstack{\pdaord}{\stackw}} \in
                \prestar{\augmented{\pda}}{\auta^\och_{\control_2}}
            } \ .
        $
    \end{itemize}
\end{definition}

It is easy to see both $\canpopaut{\control_1}{\control_2}$ and $\canpopautout{\control_1}{\control_2}$ are regular and representable by bottom-up automata since both
\[
    \prestar{\pda}{\auta_{\control_2}}
    \quad \text{ and } \quad
    \prestar{\augmented{\pda}}{\auta^\och_{\control_2}}
\]
are regular from Theorem~\ref{thm:regreach}, and bottom-up and top-down automata are effectively equivalent.
To enforce only stacks of the form
$\kstack{\pdaord}{\stackw}$
we intersect with an automaton $\auta_1$ accepting all stacks containing a single order-$(\pdaord-1)$ stack (this is clearly regular).

\subsection{Reduction to Lower Orders}
\label{sec:reduction}

We are now ready to complete the reduction.
Correctness is shown in Section~\ref{sec:correctness}.
Let $\auttrue$ be the automaton accepting all stacks.
In the following definition, a control state
$\simcontrol{\control_1}{\control_2}$
means that we are currently in control state $\control_1$ and are aiming to empty the stack on reaching $\control_2$, and the rules
    $\simsimrules$ simulate all operations apart from $\push{\pdaord}$ and $\pop{\pdaord}$ directly,
    $\simfinrules$ detect when the run is accepting,
    $\simpushrules$ follow the push branch of the tree decomposition, using tests to ensure the existence of the pop branch, and
    $\simpoprules$ follow the pop branch of the tree decomposition, also using tests to check the existence of the push branch.

\begin{nameddefinition}{def:sim-pda}{$\simulator{\pda}$}
    Given an \hopda{\pdaord}
    $\pda$
    described by the tuple
    $\tup{\controls,
         \oalphabet,
         \salphabet,
         \rules,
         \set{\uniquefinalcontrol},
         \controlinit,
         \chainit}$
    as well as families of automata
    $\sequence{\canpopaut{\control_1}{\control_2}}
              {\control_1, \control_2 \in \controls}$
    and
    $\sequence{\canpopautout{\control_1}{\control_2}}
              {\control_1, \control_2 \in \controls}$
    we define an \hopda{(\pdaord-1)} with tests
    \[
        \simulator{\pda} = \tup{
            \simulator{\controls},
            \oalphabet,
            \salphabet,
            \simulator{\rules},
            \simulator{\finals},
            \simcontrol{\controlinit}{\uniquefinalcontrol},
            \chainit
        }
    \]
    where
    \[
        \begin{array}{rcl}
            \simulator{\controls}
            &=&
            \setcomp{
                \simcontrol{\control_1}{\control_2}
            }{
                \control_1, \control_2 \in \controls
            }
            \uplus
            \set{\finalcontrol}
            \\

            \simulator{\rules}
            &=&
            \simsimrules \cup \simfinrules \cup \simpushrules \cup \simpoprules
            \\

            \simulator{\finals}
            &=&
            \set{\finalcontrol}
       \end{array}
    \]
    and we define
    \begin{itemize}
    \item
        $\simsimrules$ is the set containing all rules of the form
        \[
            \pdrulet{\simcontrol{\control_1}{\control_2}}
                    {\cha}
                    {\auttrue}
                    {\ochb}
                    {\op}
                    {\simcontrol{\control'_1}{\control_2}}
        \]
        for all
        $\pdrule{\control_1}
                {\cha}
                {\ochb}
                {\op}
                {\control'_1} \in \rules$
        with
        $\op \notin \set{\push{\pdaord}, \pop{\pdaord}}$
        and
        $\control_2 \in \controls$, and

    \item
        $\simfinrules$ is the set containing all rules of the form
        \[
            \pdrulet{\simcontrol{\control_1}{\control_2}}
                    {\cha}
                    {\auttrue}
                    {\ech}
                    {\rew{\cha}}
                    {\finalcontrol}
        \]
        for all
        $\pdrule{\control_1}
                {\cha}
                {\ech}
                {\pop{\pdaord}}
                {\control_2} \in \rules$,
        and

    \item
        $\simpushrules$ is the smallest set of rules containing all rules of the form
        \[
            \pdrulet{\simcontrol{\control_1}{\control_2}}
                    {\cha}
                    {\canpopaut{\control}{\control_2}}
                    {\ech}
                    {\rew{\cha}}
                    {\simcontrol{\control'_1}{\control}}
        \]
        for all
        $\pdrule{\control_1}
                 {\cha}
                 {\ech}
                 {\push{\pdaord}}
                 {\control'_1} \in \rules$
        and
        $\control, \control_2 \in \controls$,
        and all rules of the form
        \[
            \pdrulet{\simcontrol{\control_1}{\control_2}}
                    {\cha}
                    {\canpopautout{\control}{\control_2}}
                    {\och}
                    {\rew{\cha}}
                    {\simcontrol{\control'_1}{\control}}
        \]
        for all
        $\pdrule{\control_1}
                {\cha}
                {\ech}
                {\push{\pdaord}}
                {\control'_1} \in \rules$
        and
        $\control, \control_2 \in \controls$,
        and

    \item
        $\simpoprules$ is the set containing all rules of the form
        \[
            \pdrulet{\simcontrol{\control_1}{\control_2}}
                    {\cha}
                    {\canpopaut{\control'_1}{\control}}
                    {\ech}
                    {\rew{\cha}}
                    {\simcontrol{\control}{\control_2}}
        \]
        for all
        $\pdrule{\control_1}
                {\cha}
                {\ech}
                {\push{\pdaord}}
                {\control'_1} \in \rules$
        and
        $\control, \control_2 \in \controls$
        and all rules of the form
        \[
            \pdrulet{\simcontrol{\control_1}{\control_2}}
                    {\cha}
                    {\canpopautout{\control'_1}{\control}}
                    {\och}
                    {\rew{\cha}}
                    {\simcontrol{\control}{\control_2}}
        \]
        for all
        $\pdrule{\control_1}
                {\cha}
                {\ech}
                {\push{\pdaord}}
                {\control'_1} \in \rules$
        and
        $\control, \control_2 \in \controls$.
    \end{itemize}
\end{nameddefinition}

\changed[mh]{
    In the next section, we show the reduction is correct.

    \begin{namedlemma}{lem:correct-sim}{Correctness of $\simulator{\pda}$}
        \[
            \unbounded{\och}{\pda} \iff \unbounded{\och}{\simulator{\pda}}
        \]
    \end{namedlemma}
}

To complete the reduction, we convert the HOPDA with tests into a HOPDA without tests.

\begin{namedlemma}{lem:reduction}{Reduction to Lower Orders}
    For every \hopda{\pdaord} $\pda$ we can construct an \hopda{(\pdaord-1)} $\pda'$ such that
    \[
        \unbounded{\och}{\pda} \iff \unbounded{\och}{\pda'} \ .
    \]
\end{namedlemma}
\begin{proof}
    From \refdefinition{def:sim-pda} and \reflemma{lem:correct-sim}, we obtain from $\pda$ an \hopda{(\pdaord-1)} with tests $\simulator{\pda}$ satisfying the conditions of the lemma.  To complete the proof, we invoke \reftheorem{thm:no-tests} to find $\pda'$ as required.
\end{proof}

\section{Correctness of Reduction}
\label{sec:correctness}

This section is dedicated to the proof of
\changed[mh]{
    \reflemma{lem:correct-sim}.
}

The idea of the proof is that each run of $\pda$ can be decomposed into a tree:
each $\push{\pdaord}$ operation creates a node whose left child is the run up to the matching $\pop{\pdaord}$,
and whose right child is the run after the matching $\pop{\pdaord}$.
All other operations create a node with a single child which is the successor configuration.

Each branch of such a tree corresponds to a run of $\simulator{\pda}$.
To prove that $\simulator{\pda}$ can output an unbounded number of $\och$s we prove that any tree containing $\numof$ edges outputting $\och$ must have a branch along which $\simulator{\pda}$ would output $\ap{\log}{\numof}$ $\och$ characters.
Thus, if $\pda$ can output an unbounded number of $\och$ characters, so can $\simulator{\pda}$.

\subsection{Tree Decomposition of Runs}

Given a run
\[
    \pdrun =
    \configc_0 \pdatran{\ochb_1}
    \configc_1 \pdatran{\ochb_2}
    \cdots
    \pdatran{\ochb_\numof} \configc_\numof
\]
of $\pda$ where each $\push{\pdaord}$ operation has a matching $\pop{\pdaord}$, we can construct a tree representation of $\pdrun$ inductively.
That is, we define
$\treedecomp{\configc} = \treesingle{\ech}$ for the single-configuration run $\configc$, and, when
\[
    \pdrun = \configc \pdatran{\ochb} \pdrun'
\]
where the first rule applied does not contain a $\push{\pdaord}$ operation, we have
\[
    \treedecomp{\pdrun} = \treeap{\ochb}{\treedecomp{\pdrun'}}
\]
and, when
\[
    \pdrun = \configc_0 \pdatran{\ech}
             \pdrun_1 \pdatran{\ech}
             \pdrun_2
\]
with $\configc_1$ being the first configuration of $\pdrun_2$ and where the first rule applied in $\pdrun$ contains a $\push{\pdaord}$ operation,
$\configc_0 = \config{\control}{\stackw}$
and
$\configc_1 = \config{\control'}{\stackw}$
for some
$\control, \control', \stackw$
and there is no configuration in $\pdrun_1$ of the form
$\config{\control''}{\stackw}$, then
\[
    \treedecomp{\pdrun} =
    \treeap{\ech}{
        \treedecomp{\pdrun_1},
        \treedecomp{\pdrun_2}
    }
    \ .
\]

An accepting run of $\pda$ has the form $\pdrun \pdatran{\ech} \configc$ where $\pdrun$ has the property that all $\push{\pdaord}$ operations have a matching $\pop{\pdaord}$ and the final transition is a $\pop{\pdaord}$ operation to $\configc = \config{\control}{\kstack{\pdaord}{}}$ for some $\control \in \finals$.
Hence, we define the tree decomposition of an accepting run to be
\[
    \treedecomp{\pdrun \pdatran{\ech} \configc} =
    \treeap{\ech}{
        \treedecomp{\pdrun},
        \treesingle{\ech}
    } \ .
\]

\subsection{Scoring Trees}

In the above tree decomposition of runs, the tree branches at each instance of a $\push{\pdaord}$ operation.
This mimics the behaviour of $\simulator{\pda}$, which performs such branching non-deterministically.
Hence, given a run $\pdrun$ of $\pda$, each branch of $\treedecomp{\pdrun}$ corresponds to a run of $\simulator{\pda}$.

We formalise this intuition in the following section.
    In this section, we assign scores to each subtree $\tree$ of $\treedecomp{\pdrun}$.
These scores correspond directly to the largest number of $\och$ characters that $\simulator{\pda}$ can output while simulating a branch of $\tree$.

Note, in the following definition, we exploit the fact that only nodes with exactly one child may have a label other than $\ech$.
We also give a general definition applicable to trees with out-degree larger than 2.
This is needed in the simultaneous unboundedness section.
For the moment, we only have trees with out-degree at most 2.

Let
\[
    \isout{\ochb} =
    \begin{cases} %
        0 & \ochb = \ech \\ %
        1 & \ochb = \och %
    \end{cases} %
    \quad \text{ and } \quad
    \isout{\numof} =
    \begin{cases} %
        0 & \numof = 0 \\ %
        1 & \numof > 0 %
    \end{cases} %
    \ .
\]
Then,
\acmeasychair{
    $\treescore{\tree} =$
}{}
\[
    \acmeasychair{}{
        \treescore{\tree} =
    }
    \begin{cases} %
        0 %
        & %
        \tree = \treesingle{\ech} %
        \\ %
        \treescore{\tree_1} + \isout{\ochb} %
        & %
        \tree = \treeap{\ochb}{\tree_1} %
        \\ %
        \ap{\max\limits_{1 \leq \idxi \leq \numof}}{ %
            \treescore{\tree_\idxi} + %
            \isout{ %
                \sum\limits_{\idxj \neq \idxi} %
                \treescore{\tree_\idxj} %
            } %
        } %
        & %
        \tree = \treeap{\ech}{\tree_1, \ldots, \tree_\numof} %
    \end{cases} %
\]

We then have the following lemma for trees with out-degree 2.

\begin{namedlemma}{lem:tree-scores}{Minimum Scores}
    Given a tree $\tree$ containing $\numof$ nodes labelled $\och$, we have
    \[
        \treescore{\tree} \geq \ap{\log}{\numof}
    \]
\end{namedlemma}
\begin{proof}
    The proof is by induction over $\numof$.
    In the base case $\numof = 1$ and there is a single node $\node$ in $\tree$ labelled $\och$.
    By definition, the subtree $\tree'$ rooted at $\node$ has
    $\treescore{\tree'} = 1$.
    Since the score of a tree is bounded from below by the score of any of its subtrees, we have
    $\treescore{\tree} \geq \ap{\log}{1}$ as required.

    Now, assume $\numof > 1$.
    Find the smallest subtree $\tree'$ of $\tree$ containing $\numof$ nodes labelled $\och$.
    We necessarily have either
    \begin{enumerate}
    \item
    \label{item:score-case-one-child}
        $\tree' = \treeap{\och}{\tree_1}$,
        or

    \item
    \label{item:score-case-two-children}
        $\tree' = \treeap{\ech}{\tree_1, \tree_2}$
        where $\tree_1$ and $\tree_2$ each have at least one node each labelled $\och$.
    \end{enumerate}

    In case~(\ref{item:score-case-one-child}) we have by induction
    \[
        \treescore{\tree'} = 1 + \ap{\log}{\numof - 1} \geq \ap{\log}{\numof}
    \]

    In case~(\ref{item:score-case-two-children}) we have
    \[
        \treescore{\tree'} =
        \ap{\max}{
            \begin{array}{c} %
                \treescore{\tree_1} + \isout{\treescore{\tree_2}}, \\ %
                \treescore{\tree_2} + \isout{\treescore{\tree_1}} %
            \end{array} %
        } \ .
    \]
    We pick whichever of $\tree_1$ and $\tree_2$ has the most nodes labelled $\och$.
    This tree has at least $\ceil{\numof / 2}$ nodes labelled $\och$.
    Note, since both trees contain nodes labelled $\och$, the right-hand side of the addition is always $1$.
    Hence, we need to show
    \[
        \ap{\log}{\ceil{\numof / 2}} + 1 \geq \ap{\log}{\numof}
    \]
    which follows from
    \[
        \begin{array}{c} %
            \ap{\log}{\numof} - \ap{\log}{\ceil{\numof / 2}} %
            = %
            \ap{\log}{\frac{\numof}{\ceil{\numof / 2}}} %
            \\ %
            \leq %
            \\ %
            \ap{\log}{\frac{\numof}{\numof / 2}} %
            = %
            \ap{\log}{2} = 1 \ . %
        \end{array} %
    \]
    By our choice of $\tree'$ we thus have
    $\treescore{\tree} = \treescore{\tree'} \geq \ap{\log}{\numof}$
    as required.
\end{proof}

\acmeasychair{
    \subsection{Completing the Proof}

    To complete the proof we show the following lemmas, whose proofs are given in the full version~\cite{Hague:2015}.
        These lemmas simply formalise the connection between runs of $\pda$ and runs of $\simulator{\pda}$.

    \begin{namedlemma}{lem:branch-runs}{Scores to Runs}
        Given an accepting run $\pdrun$ of $\pda$, if
        $\treescore{\treedecomp{\pdrun}} = \numof$
        then
        $\och^\numof \in \ap{\lang}{\simulator{\pda}}$.
    \end{namedlemma}

    \begin{lemma}[$\simulator{\pda}$ to $\pda$]
    \label{lem:sim-to-pda}
        If
        $\unbounded{\och}{\simulator{\pda}}$
        then
        $\unbounded{\och}{\pda}$.
    \end{lemma}
}{
    \subsection{From Branches to Runs}

    \begin{namedlemma}{lem:branch-runs}{Scores to Runs}
        Given an accepting run $\pdrun$ of $\pda$, if
        $\treescore{\treedecomp{\pdrun}} = \numof$
        then
        $\och^\numof \in \ap{\lang}{\simulator{\pda}}$.
    \end{namedlemma}
    \begin{proof}
        Let
$\uniquefinalcontrol$
be the final (accepting) control state of $\pda$ and let
$\tree = \treedecomp{\pdrun}$.
We begin at the root node of $\tree$, which corresponds to the initial configuration of $\pdrun$.
Let
$\config{\control}{\kstack{\pdaord}{\stackw}}$
be this initial configuration and let
$\config{\simcontrol{\control}{\uniquefinalcontrol}}
        {\stackw}$
be the initial configuration of $\simulator{\pda}$.

Thus, assume we have a node $\node$ of $\tree$, with a corresponding configuration
$\configc = \config{\control}{\stackw}$
of $\pda$ and configuration
$\simulator{\configc} =
 \config{\simcontrol{\control}{\control_{\text{pop}}}}
        {\ap{\topop{\pdaord}}{\stackw}}$
of $\simulator{\pda}$ and a run $\simulator{\pdrun}$ of $\simulator{\pda}$ ending in $\simulator{\configc}$ and outputting $\brac{\numof - \treescore{\tree'}}$ $\och$ characters where $\tree'$ is the subtree of $\tree$ rooted at $\node$.
The subtree $\tree'$ corresponds to a sub-run $\pdrun'$ of $\pdrun$ where the transition immediately following $\pdrun'$ is a $\pop{\pdaord}$ transition to a control state $\control_{\text{pop}}$.

There are two cases when we are dealing with internal nodes.
\begin{itemize}
\item
    $\tree' = \treeap{\ochb}{\tree_1}$.

    In this case there is a transition
    $\configc \pdatran{\ochb} \configc'$
    via a rule
    $\pdrule{\control}{\cha}{\ochb}{\op}{\control'}$
    where
    $\op \notin \set{\push{\pdaord}, \pop{\pdaord}}$.
    Hence, we have the rule
    $\pdrulet{\simcontrol{\control}{\control_{\text{pop}}}}
             {\cha}
             {\auttrue}
             {\ochb}
             {\op}
             {\simcontrol{\control'}{\control_{\text{pop}}}}$
    in $\simulator{\pda}$ and thus we can extend $\simulator{\pdrun}$ with a transition
    $\simulator{\configc} \pdatran{\ochb} \simulator{\configc'}$
    via this rule where $\simulator{\pdrun}$, $\configc'$ and $\simulator{\configc'}$ maintain the assumptions above.

\item
    $\tree' = \treeap{\ech}{\tree_1, \tree_2}$.

    In this case we have that $\tree'$ corresponds to a sub-run
    \[
        \configc \pdatran{\ech} \pdrun_1 \pdatran{\ech} \pdrun_2
    \]
    of $\pdrun$.  The transition from $\configc$ to the beginning of $\pdrun_1$ is via a rule
    $\pdruler_1 = \pdrule{\control}{\cha}{\ech}{\push{\pdaord}}{\control_1}$
    and the transition from the end of $\pdrun_1$ to the start of $\pdrun_2$ is via a rule
    $\pdruler_2 = \pdrule{\control_2}{\cha_1}{\ech}{\pop{\pdaord}}{\control_3}$.
    Moreover, from the definition of the decomposition, the final configuration in $\pdrun_2$ is followed in $\pdrun$ by a pop rule
    $\pdruler_3 = \pdrule{\control_4}{\cha_2}{\ech}{\pop{\pdaord}}{\control_{\text{pop}}}$.

    There are two further cases depending on whether the score of $\tree'$ is derived from the score of $\tree_1$ or $\tree_2$.
    \begin{itemize}
    \item
        In the case of $\tree_1$, then, first observe that $\pdrun_2$ followed by an application of $\pdruler_3$ is a run from
        $\config{\control_3}{\stackw}$
        to
        $\config{\control_{\text{pop}}}{\ap{\pop{\pdaord}}{\stackw}}$
        where the stack
        $\ap{\pop{\pdaord}}{\stackw}$
        does not appear in $\pdrun_2$.
        Thus, there is a run of $\pda$ from
        $\config{\control_3}{\kstack{\pdaord}{\ap{\topop{\pdaord}}{\stackw}}}$
        to
        $\config{\control_{\text{pop}}}{\kstack{\pdaord}{}}$
        and moreover, this run outputs an $\och$ whenever the original run does.
        Hence, there is also a corresponding run of $\pda$ from which outputs an $\och$ whenever the original run does.

        If an $\och$ is output, we have
        $\simulator{\configc} \in
         \ap{\lang}{\canpopautout{\control_3}{\control_{\text{pop}}}}$
        and
        $\treescore{\tree'} - \treescore{\tree_1} = 1$.
        We can extend $\pdrun$ via an application of the rule
        $\pdrulet{\simcontrol{\control}{\control_{\text{pop}}}}
                 {\cha}
                 {\canpopautout{\control_3}{\control_{\text{pop}}}}
                 {\och}
                 {\rew{\cha}}
                 {\simcontrol{\control_1}{\control_3}}$
        that exists in $\simulator{\pda}$ since
        $\simulator{\configc} \in
         \ap{\lang}{\canpopautout{\control_3}{\control_{\text{pop}}}}$.
        This transition maintains the property on the stacks since the $\push{\pdaord}$ copies the topmost stack, hence $\simulator{\pda}$ does not need to change its stack.
        It maintains the property on the scores since it outputs $\och$, accounting for the part of the score contributed by $\tree_2$.
        Finally, the condition on control states is satisfied since the second component is set to $\control_2$.

        If an $\och$ is not output, then the case is similar to the above, except $\tree_2$ does not contribute to the score,
        we have
        $\simulator{\configc} \in
         \ap{\lang}{\canpopaut{\control_3}{\control_{\text{pop}}}}$,
        and the transition of $\simulator{\pda}$ is labelled $\ech$ instead of $\och$.

    \item
        The case of $\tree_2$ is almost symmetric to $\tree_1$.
        Observe that $\pdrun_1$ followed by an application of $\pdruler_2$ is a run from
        $\config{\control_1}{\ap{\push{\pdaord}}{\stackw}}$
        to
        $\config{\control_3}{\stackw}$
        where the stack
        $\stackw$
        does not appear in $\pdrun_1$.
        Thus, there is a run of $\pda$ from
        $\config{\control_1}{\kstack{\pdaord}{\ap{\topop{\pdaord}}{\stackw}}}$
        to
        $\config{\control_3}{\kstack{\pdaord}{}}$
        and moreover, this run outputs an $\och$ whenever the original run does.
        Hence, there is also a corresponding run of $\pda$ from which outputs an $\och$ whenever the original run does.

        If an $\och$ is output, we have
        $\simulator{\configc} \in
         \ap{\lang}{\canpopautout{\control_1}{\control_3}}$
        and
        $\treescore{\tree'} - \treescore{\tree_2} = 1$.
        We can extend $\pdrun$ via an application of the rule
        $\pdrulet{\simcontrol{\control}{\control_{\text{pop}}}}
                 {\cha}
                 {\canpopautout{\control_1}{\control_3}}
                 {\och}
                 {\rew{\cha}}
                 {\simcontrol{\control_3}
                             {\control_{\text{pop}}}}$
        that exists in $\simulator{\pda}$ since
        $\simulator{\configc} \in
         \ap{\lang}{\canpopautout{\control_1}{\control_3}}$
        This transition maintains the property on the stacks since the stack after the $\pop{\pdaord}$ is identical to the stack before the $\push{\pdaord}$, hence $\simulator{\pda}$ does not need to change its stack.
        It maintains the property on the scores since it outputs $\och$, accounting for the part of the score contributed by $\tree_1$.
        Finally, the condition on control states is satisfied since the second component is unchanged.

        If an $\och$ is not output, then the case is similar to the above, except $\tree_1$ does not contribute to the score,
        we have
        $\simulator{\configc} \in
         \ap{\lang}{\canpopaut{\control_1}{\control_3}}$
        and the transition of $\simulator{\pda}$ is labelled $\ech$ instead of $\och$.
   \end{itemize}
\end{itemize}

Finally, we reach a leaf node $\node$ with a run outputting the required number of $\och$s.
We need to show that the run constructed is accepting.
Let $\node'$ be the first ancestor of $\node$ that contains $\node$ in its leftmost subtree.
Let $\tree'$ be the subtree rooted at $\node'$.
This tree corresponds to a sub-run $\pdrun'$ of $\pdrun$ that is followed immediately by a $\pop{\pdaord}$ rule
$\pdrule{\control}{\cha}{\ech}{\pop{\pdaord}}{\control_{\text{pop}}}$.
Moreover, we have
$\pdrulet{\simcontrol{\control}{\control_{\text{pop}}}}
         {\cha}
         {\auttrue}
         {\ech}
         {\rew{\cha}}
         {\finalcontrol}$
with which we can complete the run of $\simulator{\pda}$ as required.

    \end{proof}

    \subsection{The Other Direction}

    Finally, we need to show that each accepting run of $\simulator{\pda}$ gives rise to an accepting run of $\pda$ containing at least as many $\och$s.

    \begin{lemma}[$\simulator{\pda}$ to $\pda$]
    \label{lem:sim-to-pda}
        We have
        $\unbounded{\och}{\simulator{\pda}}$
        implies
        $\unbounded{\och}{\pda}$.
    \end{lemma}
    \begin{proof}
        Let $\uniquefinalcontrol$ be the unique accepting conrol state of $\pda$.
Take an accepting run $\simulator{\pdrun}$ of $\simulator{\pda}$.
We show that there exists a corresponding run $\pdrun$ of $\pda$ outputting at least as many $\och$s.

Let
\[
    \configc_0 \pdatran{\ochb}
    \cdots
    \pdatran{\ochb} \configc_\numof
    \pdatran{\ech} \config{\finalcontrol}{\stackw}
\]
for some $\stackw$ be the accepting run of $\simulator{\pda}$.
We define inductively for each
$0 \leq \idxi \leq \numof$
a pair of runs $\pdrun^\idxi_1, \pdrun^\idxi_2$ of $\pda$ such that
\begin{enumerate}
\item
    $\pdrun^\idxi_2$ ends in a configuration
    $\config{\uniquefinalcontrol}{\kstack{\pdaord}{}}$
    (i.e. is accepting), and

\item
    if
    $\configc_\idxi =
     \config{\simcontrol{\control}{\control_{\text{pop}}}}
            {\stackw}$
    then
    \begin{enumerate}
    \item
        the final configuration of $\pdrun^\idxi_1$ is
        $\config{\control}
                {\kstack{\pdaord}
                        {\stackw \stackw_1 \ldots \stackw_\altnumof}}$,
        for some
        $\stackw_1, \ldots, \stackw_\altnumof$,
        and

    \item
        the first configuration of $\pdrun^\idxi_2$ is
        $\config{\control_{\text{pop}}}
                {\kstack{\pdaord}
                        {\stackw_1 \ldots \stackw_\altnumof}}$, and
    \end{enumerate}

    \item
        the sum of the number of $\och$ characters output by $\pdrun^\idxi_1$ and $\pdrun^\idxi_2$ is at least the number of $\och$ characters output by
        $\configc_0 \pdatran{\ochb_1}
         \cdots
         \pdatran{\ochb_\idxi} \configc_\idxi$.
\end{enumerate}

Initially we have
$\configc_0 = \config{\simcontrol{\controlinit}{\uniquefinalcontrol}}{\stackw}$
and
$\stackw = \stackinit{\pdaord-1}{\chainit}$.
We define
$\pdrun^0_1 = \config{\controlinit}{\kstack{\pdaord}{\stackw}}$
and
$\pdrun^0_2 = \config{\uniquefinalcontrol}{\kstack{\pdaord}{}}$
which immediately satisfy the required conditions.

Assume we have $\pdrun^\idxi_1$ and $\pdrun^\idxi_2$ as required.
We show how to obtain $\pdrun^{\idxi+1}_1$ and $\pdrun^{\idxi+1}_2$.
There are several cases depending on the rule used on the transition
$\configc_\idxi \pdatran{\ochb_{\idxi+1}} \configc_{\idxi+1}$.
Let
$\configc_\idxi =
 \config{\simcontrol{\control}{\control_{\text{pop}}}}
        {\stackw}$,
the final configuration of $\pdrun^\idxi_1$ be
$\config{\control}
        {\kstack{\pdaord}{\stackw \stackw_1 \ldots \stackw_\altnumof}}$
and the first configuration of $\pdrun^\idxi_2$ be
$\config{\control_{\text{pop}}}
        {\kstack{\pdaord}{\stackw_1 \ldots \stackw_\altnumof}}$.
\begin{itemize}
\item
    If the rule was
    $\pdrulet{\simcontrol{\control}{\control_{\text{pop}}}}
             {\cha}
             {\auttrue}
             {\ochb}
             {\op}
             {\simcontrol{\control'}{\control_{\text{pop}}}}$
    with
    $\op \notin{\push{\pdaord}, \pop{\pdaord}}$
    then we have
    $\pdrule{\control}{\cha}{\ochb}{\op}{\control'} \in \rules$
    and we define $\pdrun^{\idxi+1}_1$ to be $\pdrun^\idxi_1$ extended by an application of this rule.
    We also define $\pdrun^{\idxi+1}_2 = \pdrun^\idxi_2$.

    The required conditions are inherited from $\pdrun^\idxi_1$ and $\pdrun^\idxi_2$ since $\op$ only changes the $\topop{\pdaord}$ stack, the final configuration of $\pdrun^{\idxi+1}_2$ is the same as $\pdrun^\idxi_2$, $\control_{\text{pop}}$ is not changed, and the rule of $\pda$ outputs an $\och$ iff the rule of $\simulator{\pda}$ does.

\item
    If the rule was
    $\pdrulet{\simcontrol{\control}{\control_{\text{pop}}}}
             {\cha}
             {\canpopaut{\control'_{\text{pop}}}
                        {\control_{\text{pop}}}}
             {\ech}
             {\rew{\cha}}
             {\simcontrol{\control'}{\control'_{\text{pop}}}}$
    then we have a rule
    $\pdruler =
     \pdrule{\control}
            {\cha}
            {\ech}
            {\push{\pdaord}}
            {\control'} \in \rules$.
    Moreover, from the test
    $\canpopaut{\control'_{\text{pop}}}{\control_{\text{pop}}}$
    we know there is a run of $\pda$ from
    $\config{\control'_{\text{pop}}}{\kstack{\pdaord}{\stackw}}$
    to
    $\config{\control_{\text{pop}}}{\kstack{\pdaord}{}}$
    and hence there is also a run $\pdrun$ from
    $\config{\control'_{\text{pop}}}
            {\kstack{\pdaord}{\stackw \stackw_1 \ldots \stackw_\altnumof}}$
    to
    $\config{\control_{\text{pop}}}
            {\kstack{\pdaord}{\stackw_1 \ldots \stackw_\altnumof}}$.
    We set
    $\pdrun^{\idxi+1}_2 = \pdrun \pdrun^\idxi_2$
    and $\pdrun^{\idxi+1}_1$ to be $\pdrun^\idxi_1$ extended by an application of $\pdruler$.

    Since the final configuration of $\pdrun^{\idxi+1}_1$ is
    $\config{\control'}
            {\kstack{\pdaord}
                    {\stackw \stackw \stackw_1 \ldots \stackw_\altnumof}}$
    it is easy to check the required correspondence with the first configuration
    $\config{\control'_{\text{pop}}}
            {\kstack{\pdaord}{\stackw \stackw_1 \ldots \stackw_\altnumof}}$
    of $\pdrun^{\idxi+1}_2$.

    The remaining conditions are immediate since no $\och$ is output and the final configuration of $\pdrun^{\idxi+1}_2$ is the same as $\pdrun^\idxi_2$.

\item
    The case of
    $\pdrulet{\simcontrol{\control}{\control_{\text{pop}}}}
             {\cha}
             {\canpopautout{\control'_{\text{pop}}}
                           {\control_{\text{pop}}}}
             {\och}
             {\rew{\cha}}
             {\simcontrol{\control'}{\control'_{\text{pop}}}}$
    is almost identical to the previous case.
    To adapt the proof, one needs only observe that since
    $\configc_\idxi \in
     \ap{\lang}{\canpopautout{\control'_{\text{pop}}}
                             {\control_{\text{pop}}}}$
    the run $\pdrun$ used to extend $\pdrun^\idxi_2$ also outputs at least one $\och$ character.

\item
    If the rule was
    $\pdrulet{\simcontrol{\control}{\control_{\text{pop}}}}
             {\cha}
             {\canpopaut{\control_1}{\control_2}}
             {\ech}
             {\rew{\cha}}
             {\simcontrol{\control_2}{\control_{\text{pop}}}}$
    then there is also a rule
    $\pdruler =
     \pdrule{\control}
            {\cha}
            {\ech}
            {\push{\pdaord}}
            {\control_1} \in \rules$
    and from the test
    $\canpopaut{\control_1}{\control_2}$
    we know there is a run of $\pda$ from
    $\config{\control_1}{\kstack{\pdaord}{\stackw}}$
    to
    $\config{\control_2}{\kstack{\pdaord}{}}$
    and therefore there is also a run $\pdrun$ that goes from
    $\config{\control_1}
            {\kstack{\pdaord}
                    {\stackw \stackw \stackw_1 \ldots \stackw_\altnumof}}$
    to
    $\config{\control_2}
            {\kstack{\pdaord}
                    {\stackw \stackw_1 \ldots \stackw_\altnumof}}$.
    We set $\pdrun^{\idxi+1}_1$ to be $\pdrun^\idxi_1$ extended with an application of $\pdruler$ and then the run $\pdrun$.
    We also set
    $\pdrun^{\idxi+1}_2 = \pdrun^\idxi_2$.

    To verify that the properties hold, we observe that
    $\configc_{\idxi+1} =
     \config{\simcontrol{\control_2}{\control_{pop}}}{\stackw}$,
    and $\pdrun^{\idxi+1}_1$ ends with
    $\config{\control_2}
            {\kstack{\pdaord}{\stackw \stackw_1 \ldots \stackw_\altnumof}}$
    and $\pdrun^{\idxi+1}_2$ still begins with
    $\config{\control_{\text{pop}}}
            {\kstack{\pdaord}{\stackw_1 \ldots \stackw_\altnumof}}$
    and has the required final configuration.
    The property on the number of $\och$s holds since the rule of $\simulator{\pda}$ did not output an $\och$.

\item
    The case of
    $\pdrulet{\simcontrol{\control}{\control_{\text{pop}}}}
             {\cha}
             {\canpopautout{\control_1}{\control_2}}
             {\och}
             {\rew{\cha}}
             {\simcontrol{\control_2}{\control_{\text{pop}}}}$
    is almost identical to the previous case.
    To adapt the proof, one needs only observe that since
    $\configc_\idxi \in
     \ap{\lang}{\canpopautout{\control_1}{\control_2}}$
    the run $\pdrun$ used to extend $\pdrun^\idxi_1$ also outputs at least one $\och$ character.
\end{itemize}
Finally,  when we reach $\idxi = \numof$ we have from the final transition of the run of $\simulator{\pda}$ that there is a rule
$\pdrule{\control}{\cha}{\ech}{\pop{\pdaord}}{\control_{\text{pop}}}$.
We combine $\pdrun^\numof_1$ and $\pdrun^\numof_2$ with this pop transition, resulting in an accepting run of $\pda$ that outputs at least as many $\och$ characters as the run of $\simulator{\pda}$.

    \end{proof}
}


\section{Multiple Characters}
\label{sec:simultaneous}

We generalise the previous result to the full diagonal problem.
Na\"{i}vely, the previous approach cannot work.
Consider the HOPDA executing
\[
    \push{1}^\numof; \push{\pdaord}; \pop{1}^\numof; \pop{\pdaord}; \pop{1}^\numof
\]
where the first sequence of $\pop{1}$ operations output $\och_1$ and the second sequence output $\och_2$.

The corresponding run trees are of the form given in Figure~\ref{fig:multi-branch}.
In particular, $\simulator{\pda}$ can only choose one branch, hence all runs of $\simulator{\pda}$ produce a bounded number of $\och_1$s or a bounded number of $\och_2$s.
They cannot be simultaneously unbounded.

\begin{figure}
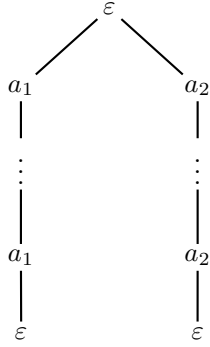

\begin{center}
    \acmeasychair{
        \psset{nodesep=.5ex,rowsep=2ex,colsep=6ex}
    }{
        \psset{nodesep=.5ex,rowsep=4ex,colsep=6ex}
    }
    \begin{psmatrix}
                               & \rnode{N}{$\ech$} & \\
        \rnode{N0}{$\och_1$}   &                   & \rnode{N1}{$\och_2$} \\
        \pselipsenode{N00}     &                   & \pselipsenode{N10} \\
        \rnode{N000}{$\och_1$} &                   & \rnode{N100}{$\och_2$} \\
        \rnode{N0000}{$\ech$}  &                   & \rnode{N1000}{$\ech$}

        \ncline{N}{N0}
        \ncline{N}{N1}
        \ncline{N0}{N00}
        \ncline{N00}{N000}
        \ncline{N000}{N0000}
        \ncline{N1}{N10}
        \ncline{N10}{N100}
        \ncline{N100}{N1000}
    \end{psmatrix}
    \caption{\label{fig:multi-branch}An example showing that following a single branch does not work for simultaneous unboundedness.}
\end{center}
\end{figure}

For $\simulator{\pda}$ to be able to output both an unbounded number of $\och_1$ and $\och_2$ characters, it must be able to output two branches of the tree.
To this end, we define a notion of $\numchs$-branch HOPDA, which output trees with up to $\numchs$ branches.
We then show that the reduction from \hopda{\pdaord} to \hopda{(\pdaord-1)} can be generalised to $\numchs$-branch HOPDA.

\subsection{Branching HOPDA}

We define \hopda{\pdaord} outputting trees with at most $\numchs$ branches, denoted \bhopda{\pdaord}{\numchs}.
Note, an \hopda{\pdaord} that outputs a word is an \bhopda{\pdaord}{1}.
Indeed, any \bhopda{\pdaord}{\numchs} is also an \bhopda{\pdaord}{\numchs'} whenever $\numchs \leq \numchs'$.

\begin{definition}[\bhopda{\pdaord}{\numchs}]
    We define an \emph{order-$\pdaord$ $\numchs$-branch pushdown automaton (\bhopda{\pdaord}{\numchs})} to be given by a tuple
    $\pda = \tup{\controls,
                 \oalphabet,
                 \salphabet,
                 \rules,
                 \finals,
                 \controlinit,
                 \chainit,
                 \csrank}$
    where $\controls$, $\oalphabet$, $\salphabet$, $\finals$, $\controlinit$, and $\chainit$ are as in HOPDA.
    The set of rules
    $\rules \subseteq \bigcup\limits_{1 \leq \numof \leq \numchs}
                      \controls \times
                      \salphabet \times
                      \oalphabet \times
                      \ops{\pdaord} \times
                      \controls^\numof$
    together with a mapping
    $\csrank : \controls \rightarrow \set{1, \ldots, \numchs}$
    such that for all
    $\tup{\control,
          \cha,
          \ochb,
          \op,
          \control_1, \ldots, \control_\numof} \in \rules$
    we have
    $\ap{\csrank}{\control} \geq
     \ap{\csrank}{\control_1} +
     \cdots +
     \ap{\csrank}{\control_\numof}$.
\end{definition}

We use the notation
$\pdrule{\control}{\cha}{\ochb}{\op}{\control_1, \ldots, \control_\numof}$
to denote a rule
$\tup{\control,
      \cha,
      \ochb,
      \op,
      \control_1, \ldots, \control_\numof} \in \rules$.
Intuitively, such a rule generates a node of a tree with $\numof$ children.
The purpose of the mapping $\csrank$ is to bound the number of branches that this tree may have.
Hence, at each branching rule, the quota of branches is split between the different subtrees.
The existence of such a mapping implies this information is implicit in the control states and an \bhopda{\pdaord}{\numchs} can only output trees with at most $\numchs$ branches.

From the initial configuration
$\configc_0 = \config{\controlinit}{\stackinit{\pdaord}{\chainit}}$
a run of an \bhopda{\pdaord}{\numchs} is a tree
$\tree = \tup{\treedom, \treelabelling}$
whose nodes are labelled with \hopda{\pdaord} configurations,
and generates an output tree
$\tree' = \tup{\treedom, \treelabelling'}$
whose nodes are labelled with symbols from the output alphabet.
Precisely
\begin{itemize}
\item
    $\ap{\treelabelling}{\rootnode} = \configc_0$, and

\item
    for a node $\node$ with children
    $\node_1, \ldots, \node_\numof$
    and
    $\ap{\treelabelling}{\node} = \config{\control}{\stackw}$
    there is a rule
    $\pdrule{\control}{\cha}{\ochb}{\op}{\control_1, \ldots, \control_\numof}$
    such that for all
    $1 \leq \idxi \leq \numof$
    we have
    $\ap{\treelabelling}{\node_\idxi} = \config{\control_\idxi}{\stackw'}$
    where
    $\ap{\topop{1}}{\stackw} = \cha$,
    $\stackw' = \ap{\op}{\stackw}$.
    Moreover we have
    $\ap{\treelabelling'}{\node} = \ochb$.

\item
    For all leaf nodes $\node$ we have
    $\ap{\treelabelling'}{\node} = \ech$.
\end{itemize}
The run is accepting if for all leaf nodes $\node$ we have
$\ap{\treelabelling}{\node} = \config{\control}{\kstack{\pdaord}{}}$
and
$\control \in \finals$.
Let $\ap{\lang}{\pda}$ be the set of output trees of $\pda$.

Given an output tree $\tree$ we write
$\chcount{\och}{\tree}$
to denote the number of nodes labelled $\och$ in $\tree$.
For an \bhopda{\pdaord}{\numchs} $\pda$, we define
\begin{multline*}
    \unbounded{\och_1, \ldots, \och_\numchs}{\pda} =
    \\
    \forall \numof .
    \exists \tree \in \ap{\lang}{\pda} .
    \forall 1 \leq \idxi \leq \numchs .
    \chcount{\och_\idxi}{\tree} \geq \numof \ .
\end{multline*}


\section{Reduction For Simultaneous Unboundedness}
\label{sec:reduction-sim}

Given an \bhopda{\pdaord}{\numchs} $\pda$ we construct an \bhopda{\pdaord-1}{\numchs} $\simulator{\pda}$ such that
\[
    \unbounded{\och_1, \ldots, \och_\numchs}{\pda}
    \iff
    \unbounded{\och_1, \ldots, \och_\numchs}{\simulator{\pda}} \ .
\]
Moreover, we show
$\unbounded{\och_1, \ldots, \och_\numchs}{\pda}$
is decidable for a \bhopda{0}{\numchs} (i.e. a regular automaton outputting an $\numchs$-branch tree) $\pda$.

For simplicity, we assume for all rules
$\pdrule{\control}{\cha}{\ochb}{\op}{\control_1, \ldots, \control_\numof}$
if $\numof > 1$ then $\op = \rew{\cha}$ (i.e. the stack is unchanged).
Additionally we have $\ochb = \ech$.

We also make analogous assumptions to the single character case.
That is, we assume
$\oalphabet = \set{\och_1, \ldots, \och_\numchs, \ech}$
and use $\ochb$ to range over $\oalphabet$.
Moreover, all rules of the form
$\pdrule{\control}{\cha}{\ochb}{\op}{\control'}$
with $\op = \push{\pdaord}$ or $\op = \pop{\pdaord}$ have $\ochb = \ech$.
Finally, we assume acceptance is by reaching a unique control state in $\finals$ with an empty stack.

\subsection{Some Intuition}
\label{sec:outline-sim}

We briefly sketch the intuition behind the algorithm.
We illustrate the reduction from \bhopda{\pdaord}{\numchs} to \bhopda{\pdaord-1}{\numchs} in Figure~\ref{fig:reduction-sim}.
\begin{itemize}
\item
    We begin with an \hopda{\pdaord} which we first interpret as an \bhopda{\pdaord}{\numchs}.
    This is possible because an \bhopda{\pdaord}{\numchs} can produce \emph{at most} $\numchs$ branches.
    Thus, an \hopda{\pdaord} --- which produces a single branch --- is also a \bhopda{\pdaord}{\numchs}.
    We work with HOPDA producing $\numchs$ branches because, after each reduction step, we will need to output one branch for each character in $\och_1, \ldots, \och_\numchs$.
\item
    We have an \bhopda{\pdaord}{\numchs} $\pda$ that outputs a tree with at most $\numchs$ branches.
    In Figure~\ref{fig:reduction-sim} we show part of a run tree with $2$ branches.
    The $\push{\pdaord}$ and $\pop{\pdaord}$ operations are shown on the edges of the tree.
    Nodes are numbered to help identify them during the different transformations.

\item
    We ``decompose'' this tree into another tree where the branches appearing after the $\pop{\pdaord}$ operations are hung from the same parent as their matching $\push{\pdaord}$.
    This is shown in the middle of Figure~\ref{fig:reduction-sim}.
    Notice that this tree has an unbounded number of branches (it branches at each $\push{\pdaord}$).
    However, we know that the maximum out-degree of any of its nodes is $(\numchs + 1)$ since 
the source of a  $\push{\pdaord}$-labelled edge has one child, and we add at most $\numchs$ extra children corresponding to the $\pop{\pdaord}$ on each of its at most $\numchs$ branches.

\item
    We prove a generalisation of \reflemma{lem:tree-scores} that shows a run tree with at least $\numof$ instances of a character $\och$ has a branch with a score of at least
    $\ap{\log_{(\numchs+1)}}{\numof}$.
    Thus, we need to select one branch for each $\och$ we wish to output.

\item
    We build an \bhopda{\pdaord-1}{\numchs} $\simulator{\pda}$ that non-deterministically picks out the highest scoring branches for each $\och$.  This is shown on the right of Figure~\ref{fig:reduction-sim}.
\end{itemize}

\begin{figure*}
\easyhevea{}{\begin{figure*}}
    \centering
    \psset{nodesep=1ex}
    \acmeasychair{
        \psset{rowsep=3ex}
    }{}
    \subfloat[][An $\numchs$-branch run tree of $\pda$]{
        \begin{psmatrix}
                                     & \rnode{N}{1}  &                         \\
                                     & \rnode{N0}{2} &                         \\
            \rnode{N00}{3}           &               & \rnode{N01}{5}          \\
            \rnode{N000}{4}          &               & \rnode{N010}{6}         \\
            \rnode{N0000}{$\vdots$}  &               & \rnode{N0100}{$\vdots$}

            \ncline{N}{N0}\naput{$\push{\pdaord}$}
            \ncline{N0}{N00}
            \ncline{N0}{N01}
            \ncline{N00}{N000}\nbput{$\pop{\pdaord}$}
            \ncline{N01}{N010}\naput{$\pop{\pdaord}$}
            \ncline{N000}{N0000}
            \ncline{N010}{N0100}
        \end{psmatrix}
    }
    \acmeasychair{
        \qquad\qquad
    }{
        \qquad
    }
    \subfloat[][\label{subfig:decomp}The decomposition of the run tree]{
        \begin{psmatrix}
                           &                & \rnode{N}{1}          &               \\
                           & \rnode{N0}{2}  & \rnode{N1}{4}         & \rnode{N2}{6} \\
            \rnode{N00}{3} & \rnode{N01}{5} & \rnode{N10}{$\vdots$} & \rnode{N20}{$\vdots$} \\
            \\

            \ncline{N}{N0}
            \ncline{N}{N1}
            \ncline{N}{N2}
            \ncline{N0}{N00}
            \ncline{N0}{N01}
            \ncline{N1}{N10}
            \ncline{N2}{N20}
        \end{psmatrix}
    }
    \acmeasychair{
        \qquad\qquad
    }{
        \qquad
    }
    \subfloat[][An $\numchs$-branch subtree of (\ref{subfig:decomp})]{
        \begin{psmatrix}
                           & \rnode{N}{1}         \\
            \rnode{N0}{2}  & \rnode{N1}{4}         \\
            \rnode{N01}{5} & \rnode{N10}{$\vdots$} \\
            \\

            \ncline{N}{N0}
            \ncline{N}{N1}
            \ncline{N0}{N01}
            \ncline{N1}{N10}
        \end{psmatrix}
    }
\easyhevea{}{\end{figure*}}
    \caption{\label{fig:reduction-sim}Illustrating the reduction steps.}
\end{figure*}

\subsection{Branching HOPDA with Regular Tests}

As before, we instrument our HOPDA with tests.
Removing these tests requires a simple adaptation of Broadbent\etal~\cite{Broadbent:2010}.

\begin{definition}[\bhopda{\pdaord}{\numchs} with Tests]
    Given a sequence of automata
    $\auta_1, \ldots, \auta_\numof$,
    an \emph{\bhopda{\pdaord}{\numchs} with tests} is given by a tuple
    $\pda = \tup{\controls,
                 \oalphabet,
                 \salphabet,
                 \rules,
                 \finals,
                 \controlinit,
                 \chainit,
                 \csrank}$
    where $\controls$, $\oalphabet$, $\salphabet$, $\finals$, $\controlinit$, $\chainit$ are as in HOPDA.
    The set of rules
    $\rules \subseteq \bigcup\limits_{1 \leq \numof \leq \numchs}
                      \controls \times
                      \salphabet \times
                      \set{\auta_1, \ldots, \auta_\numof} \times
                      \oalphabet \times
                      \ops{\pdaord} \times
                      \controls^\numof$
    together with a mapping
    $\csrank : \controls \rightarrow \set{1, \ldots, \numchs}$
    such that for all
    $\tup{\control,
          \cha,
          \auta,
          \ochb,
          \op,
          \control_1, \ldots, \control_\numof} \in \rules$
    we have
    $\ap{\csrank}{\control} \geq
     \ap{\csrank}{\control_1} +
     \cdots +
     \ap{\csrank}{\control_\numof}$.
\end{definition}

We use the notation
$\pdrulet{\control}
         {\cha}
         {\auta}
         {\ochb}
         {\op}
         {\control_1, \ldots, \control_\numof}$
to denote a rule
$\tup{\control,
      \cha,
      \auta,
      \ochb,
      \op,
      \control_1, \ldots, \control_\numof} \in \rules$.

From the initial configuration
$\configc_0 = \config{\controlinit}{\stackinit{\pdaord}{\chainit}}$
a run of an \bhopda{\pdaord}{\numchs} with tests is a tree
$\tree = \tup{\treedom, \treelabelling}$
and generates an output tree
$\pdrun = \tup{\treedom, \treelabelling'}$
where
\begin{itemize}
\item
    $\ap{\treelabelling}{\rootnode} = \configc_0$, and

\item
    for a node $\node$ with children
    $\node_1, \ldots, \node_\numof$
    and
    $\ap{\treelabelling}{\node} = \config{\control}{\stackw}$
    there is a rule
    $\pdrulet{\control}
             {\cha}
             {\auta}
             {\ochb}
             {\op}
             {\control_1, \ldots, \control_\numof}$
    such that
    $\stackw \in \ap{\lang}{\auta}$
    and for all
    $1 \leq \idxi \leq \numof$
    we have
    $\ap{\treelabelling}{\node_\idxi} = \config{\control_\idxi}{\stackw'}$
    where
    $\ap{\topop{1}}{\stackw} = \cha$,
    and
    $\stackw' = \ap{\op}{\stackw}$.
    Moreover we have
    $\ap{\treelabelling'}{\node} = \ochb$.

\item
    For all leaf nodes $\node$ we have
    $\ap{\treelabelling'}{\node} = \ech$.
\end{itemize}
The run is accepting if for all leaf nodes $\node$ we have
$\ap{\treelabelling}{\node} = \config{\control}{\kstack{\pdaord}{}}$
and
$\control \in \finals$.
Let $\ap{\lang}{\pda}$ be the set of output trees of $\pda$.

\begin{namedtheorem}{thm:no-tests-sim}{Removing Tests}
    \cite[Theorem 3 (adapted)]{Broadbent:2010}
    For every \bhopda{\pdaord}{\numchs} with tests $\pda$, we can compute an \bhopda{\pdaord}{\numchs} $\pda'$ with
    $\ap{\lang}{\pda} = \ap{\lang}{\pda'}$.
\end{namedtheorem}
\acmeasychair{
    The adapted proof of the above theorem is given in the full version~\cite{Hague:2015}.
}{
    \begin{proof}
        This is a straightforward adaptation of Broadbent\etal~\cite{Broadbent:2010}.
Let the \bhopda{\pdaord}{\numchs} with tests be
$\pda = \tup{\controls,
             \oalphabet,
             \salphabet,
             \rules,
             \finals,
             \controlinit,
             \chainit,
             \csrank}$
with test automata
$\auta_1, \ldots, \auta_\numof$.
We build an \bhopda{\pdaord}{\numchs} that mimics $\pda$ almost directly.
The only difference is that each character $\cha$ appearing in the stack is replaced by
\[
    \funcha{\cha}{\autfuns{1}, \ldots, \autfuns{\numof}} \ .
\]
For each test $\auta$ we have a vector of functions
\[
    \autfuns = \tup{\autfun{1}, \ldots, \autfun{\pdaord}} \ .
\]
The function
$\autfun{\opord} : \sastates \rightarrow \sastates$
intuitively describes runs of $\auta$ from the bottom of
$\ap{\topop{\opord+1}}{\stackw}$
to the top of
$\ap{\pop{\opord}}{\ap{\topop{\opord+1}}{\stackw}}$.
Thus, we can reconstruct an entire run over
$\ap{\pop{1}}{\stackw}$
from initial state $\sastate$ as
\[
    \sastate' = \ap{\autfun{1}}{\cdots\ap{\autfun{\pdaord}}{\sastate}}
\]
and then we can consult $\sadelta$ to complete the run by adding the effect of reading
$\ap{\topop{1}}{\stackw}$.

Thus, let
$\auta_\idxi = \tup{\sastates_\idxi,
                    \saalphabet,
                    \sastate^\idxi_{\text{in}},
                    \sadelta_\idxi,
                    \safinals^\idxi}$.
We define
\[
    \notest{\pda} =
    \tup{\controls,
         \oalphabet,
         \notest{\salphabet},
         \notest{\rules},
         \finals,
         \controlinit,
         \notest{\chainit},
         \csrank}
\]
where
\[
    \notest{\salphabet} =
    \setcomp{
        \funcha{\cha}{\autfuns{1}, \ldots, \autfuns{\numof}}
    }{
        \cha \in \salphabet
        \land
        \forall \idxi .
            \autfuns{\idxi} \in
            \brac{\sastates_\idxi
                  \rightarrow
                  \sastates_\idxi}^\pdaord
    }
\]
and $\notest{\rules}$ is the smallest set of rules of the form
\[
    \pdrule{\control}
           {\notest{\cha}}
           {\ochb}
           {\ap{\optestsim}{\op, \notest{\cha}}}
           {\control_1, \ldots, \control_\altnumof}
\]
where
$\notest{\cha} = \funcha{\cha}
                        {\autfuns{1},
                         \ldots,
                         \autfuns{\numof}}$
and
$\pdrulet{\control}
         {\cha}
         {\auta_\idxi}
         {\ochb}
         {\op}
         {\control_1, \ldots, \control_\altnumof} \in \rules$
and
$\ap{\testaccepts}{\cha,
                   \autfuns{\idxi},
                   \sadelta_\idxi,
                   \sastate^\idxi_{\text{in}},
                   \safinals^\idxi}$
and we define
\[
    \begin{array}{c}
        \ap{\testaccepts}{\cha,
                          \autfun{1}, \ldots, \autfun{\pdaord},
                          \sadelta,
                          \sastate_{\text{in}},
                          \safinals}
        \\
        \iff
        \\
        \sastate =
        \ap{\autfun{1}}{
            \cdots
            \ap{\autfun{\pdaord}}{
                \sastate_{\text{in}}
            }
        }
        \land
        \ap{\sadelta}{\sastate, \sopen{\pdaord} \cdots \sopen{1} \cha} \in \safinals
    \end{array}
\]
where
$\ap{\sadelta}{\sastate, \sopen{\pdaord} \cdots \sopen{1} \cha}$
is shorthand for the repeated application of $\sadelta$ on $\cha$ then $\sopen{1}$, back to $\sopen{\pdaord}$,
and we define
$\ap{\optestsim}{\op, \notest{\cha}} = \notest{\op}$
following the cases below.
Let
$\notest{\cha} = \funcha{\cha}{\autfuns{1}, \ldots, \autfuns{\numof}}$.
\begin{itemize}
\item
    When
    $\op = \rew{\chb}$
    then
    $\notest{\op} = \funcha{\chb}{\autfuns{1}, \ldots, \autfuns{\numof}}$.

\item
    When
    $\op = \push{\opord}$
    then
    $\notest{\op} =
     \push{\op};
     \rew{\funcha{\cha}{\autfuns{1}', \ldots, \autfuns{\numof}'}}$
    where for all $\idxi$ we have
    \[
        \autfuns{\idxi} =
        \tup{\autfun{1}, \ldots, \autfun{\opord-1},
             \autfun{\opord}',
             \autfun{\opord+1},
             \ldots
             \autfun{\pdaord}}
    \]
    and
    \[
        \ap{\autfun{\opord}'}{\sastate} =
        \ap{\autfun{\opord}}{
            \ap{\sadelta_\idxi}{
                \ap{\autfun{1}}{
                    \cdots
                    \ap{\autfun{\opord}}{
                        \sastate
                    }
                },
                \sclose{\opord-1}
                \sopen{\opord-1} \cdots \sopen{1}
                \cha
            }
        } .
    \]
    I.e., we apply the functions to read the whole stack once, and then the correct part of the copy created by the $\push{\opord}$.

\item
    When
    $\op = \pop{\opord}$
    then
    \[
        \notest{\op} =
        \pop{\op};
        \toptest{\funcha{\chb}{\autfuns{1}', \ldots, \autfuns{\numof}'}};
        \rew{\funcha{\chb}{\autfuns{1}'', \ldots, \autfuns{\numof}''}}
    \]
    where for all $\idxi$ we have
    $\autfuns{\idxi} = \tup{\autfun{1}, \ldots, \autfun{\pdaord}}$
    and
    $\autfuns{\idxi}' = \tup{\autfun{1}', \ldots, \autfun{\pdaord}'}$
    and
    \[
        \autfuns{\idxi}'' =
        \tup{\autfun{1}', \ldots, \autfun{\opord-1}',
             \autfun{\opord},
             \ldots
             \autfun{\pdaord}} \ .
    \]
    We can see that this is correct since we do not update the functions that read parts of the stack unchanged (i.e., stacks outside of those changed by the $\pop{\opord}$), and we take the functions that are correct for the newly exposed top parts of the stack for the remaining functions.
\end{itemize}
Finally, we set
$\notest{\chainit} =
 \funcha{\chainit}{\autfuns{1}, \ldots, \autfuns{\numof}}$
where for each $\idxi$ we have
$\autfuns{\idxi} = \tup{\autfun{1}, \ldots, \autfun{\pdaord}}$
such that for each $\opord$ we have
$\ap{\autfun{\opord}}{\sastate} =
 \ap{\sadelta}{\sastate,
               \sclose{\opord} \cdots \sclose{\pdaord}}$.

    \end{proof}
}

\subsection{Building The Automata}

Previously we built automata
$\canpopaut{\control_1}{\control_2}$
to indicate that from $\control_1$, the current top stack could be removed, arriving at $\control_2$.
This is fine for words, however, we now have $\numchs$-branch trees.
It is no longer enough to specify a single control state:
the top stack may be popped once on each branch of the tree, hence for a control state $\control$ we need
\changed[mh]{
    to recognise configurations with control state $\control$ from which there is a run tree where the leaves of the trees are labelled with configurations with control states
    $\control_1, \ldots, \control_\numof$
    and empty stacks.
    Moreover we need to recognise the set $\outputs$ of characters output by the run tree.
    More precisely, %
}
for these automata we write
\[
    \canpopautbr{\control}{\control_1,\ldots,\control_\numof}{\outputs}
\]
where
$\ap{\csrank}{\control} \geq
 \ap{\csrank}{\control_1} +
 \cdots +
 \ap{\csrank}{\control_\numof}$
and
$\outputs \subseteq \set{\och_1, \ldots, \och_\numchs}$.
We have
$\stackw \in
 \ap{\lang}{\canpopautbr{\control}
                        {\control_1,\ldots,\control_\numof}
                        {\outputs}}$
iff there is a run tree $\tree$ with the root labelled
$\config{\control}{\kstack{\pdaord}{\stackw}}$
and $\numof$ leaf nodes labelled
$\config{\control_1}{\kstack{\pdaord}{}},
 \ldots,
 \config{\control_\numof}{\kstack{\pdaord}{}}$
respectively.
Moreover, we have $\och \in \outputs$ iff the corresponding output tree $\tree'$ has
$\chcount{\och}{\tree'} > 0$.

\subsubsection{Alternating HOPDA}

To construct the required stack automata, we need to do reachability analysis of \bhopda{\pdaord}{\numchs}.
We show that such analyses can be rephrased in terms of alternating higher-order pushdown systems (HOPDS),
for which the required algorithms are already known~\cite{Broadbent:2012}.
Note, we refer to these machines as ``systems'' rather than ``automata'' because they do not output a language.

\begin{definition}[Alternating HOPDS]
    An \emph{alternating order-$\pdaord$ pushdown system} is a tuple
    $\pda = \tup{\controls, \salphabet, \rules}$
    where
        $\controls$ is a finite set of control states,
        $\salphabet$ is a finite stack alphabet, and
        \[
            \rules \subseteq
            \brac{\controls \times
                  \salphabet \times
                  \ops{\pdaord} \times
                  \controls}
            \cup
            \brac{\controls \times
                  \salphabet \times
                  2^{\controls}}
        \]
        is a set of transition rules.
\end{definition}

We write
$\apdrule{\control}{\cha}{\op}{\control}$
to denote
$\tup{\control, \cha, \op, \control} \in \rules$
and
$\altrule{\control}{\cha}{\control_1, \ldots, \control_\numof}$
to denote
$\tup{\control, \cha, \set{\control_1, \ldots, \control_\numof}} \in \rules$.

An run of an alternating HOPDS may split into several configurations, each of which must reach a target state.
Hence, the branching of the alternating HOPDS mimics the branching of the \bhopda{\pdaord}{\numchs}.
Given a set $\configs$ of configurations, we define
$\prestar{\pda}{\configs}$
to be the smallest set $\configs'$ such that
\[
    \begin{array}{rcl} %
        \configs' %
        &=& %
        \configs %
        \ \cup \\ %
        && %
        \setcomp{ %
            \config{\control}{\stackw} %
        }{ %
            \begin{array}{c} %
                \apdrule{\control}{\cha}{\op}{\control'} \in \rules %
                \ \land \\ %
                \ap{\topop{1}}{\stackw} = \cha %
                \ \land \\ %
                \config{\control'}{\ap{\op}{\stackw}} \in \configs' %
            \end{array} %
        } %
        \ \cup \\ %
        && %
        \setcomp{ %
            \config{\control}{\stackw} %
        }{ %
            \begin{array}{c} %
                \altrule{\control} %
                        {\cha} %
                        {\control_1, \ldots, \control_\numof} \in \rules %
                \ \land \\ %
                \ap{\topop{1}}{\stackw} = \cha %
                \ \land \\ %
                \forall \idxi . %
                    \config{\control_\idxi}{\stackw} \in \configs' %
            \end{array} %
        } \ . %
    \end{array} %
\]

\subsubsection{Constructing the Tests}

In order to use standard results to obtain
$\canpopautbr{\control}{\control_1,\ldots,\control_\numof}{\outputs}$
we construct an alternating HOPDS $\alternating{\pda}$ and automaton $\auta$ such that
\changed[mh]{
    checking
    $\configc \in \prestar{\alternating{\pda}}{\auta}$
    for a suitably constructed $\configc$ allows us to check whether
    $\stackw \in
     \ap{\lang}{\canpopautbr{\control}
                            {\control_1,\ldots,\control_\numof}
                            {\outputs}}$.
}

The alternating HOPDS $\alternating{\pda}$ will mimic the branching of $\pda$ with alternating transitions\footnote{
    We slightly alter the alternation rule from ICALP 2012~\cite{Broadbent:2012} by matching the top stack character as well as the control state.
    This is a benign alteration since it one can track the top of stack character in the control state.
}
$\altrule{\control}{\cha}{\control_1, \ldots, \control_\numof}$
of $\alternating{\pda}$.
It will maintain in its control states information about which characters have been output, as well as which control states should appear on the leaves of the branches.
This final piece of information prevents all copies of the alternating HOPDS from verifying the same branch of $\pda$.

\begin{definition}[$\alternating{\pda}$]
    Given an \bhopda{\pdaord}{\numchs}
    $\pda$
    described by the tuple
    $\tup{\controls,
          \oalphabet,
          \salphabet,
          \rules,
          \finals,
          \controlinit,
          \chainit}$,
    of $\pda$, we define
    \[
        \alternating{\pda} =
        \tup{\alternating{\controls},
             \salphabet,
             \alternating{\rules}}
    \]
    where
    \[
        \alternating{\controls} =
        \setcomp{\tup{\control,
                      \outputs,
                      \control_1, \ldots, \control_\numof}}
                {\begin{array}{c} %
                     1 \leq \numof \leq \numchs\ \land \\ %
                     \outputs \subseteq \set{\och_1, \ldots, \och_\numchs}\ \land \\ %
                     \control_1, \ldots, \control_\numof \in \controls %
                 \end{array}} %
    \] %
    and $\alternating{\rules}$ is the set of rules containing, for each
    \[
        \pdrule{\control}
               {\cha}
               {\ochb}
               {\op}
               {\control'} \in \rules
    \]
    all rules
    \[
        \apdrule{\tup{\control,
                      \outputs,
                      \control_1, \ldots, \control_{\idxi}}}
                {\cha}
                {\op}
                {\tup{\control_1,
                      \outputs \setminus \set{\ochb},
                      \control_1, \ldots, \control_\idxi}}
    \]
    and for each
    \[
        \pdrule{\control}
               {\cha}
               {\ech}
               {\rew{\cha}}
               {\control_1, \ldots, \control_\numof} \in \rules
    \]
    with $\numof > 1$ all alternating rules
    \acmeasychair{
        \[
            \altrule{\tup{\control,
                          \outputs,
                          \control'_1, \ldots, \control'_{\idxi}}}
                    {\cha}
                    {\begin{array}{c} %
                        \tup{\control_1, %
                             \outputs_1, %
                             \control^1_1, \ldots, \control^1_{\idxi_1}}, \\ %
                        \ldots \\ %
                        \tup{\control_\numof, %
                             \outputs_\numof, %
                             \control^\numof_1, \ldots, \control^\numof_{\idxi_\numof}} %
                    \end{array}} %
        \]
    }{
        \[
            \altrule{\tup{\control,
                          \outputs,
                          \control'_1, \ldots, \control'_{\idxi}}}
                    {\cha}
                    {\tup{\control_1,
                          \outputs_1,
                          \control^1_1, \ldots, \control^1_{\idxi_1}},
                     \ldots
                     \tup{\control_\numof,
                          \outputs_\numof,
                          \control^\numof_1, \ldots, \control^\numof_{\idxi_\numof}}}
        \]
    }
    where
    $\control'_1, \ldots, \control'_\idxi$
    is a permutation of
    $\control^1_1, \ldots, \control^1_{\idxi_1},
     \ldots
     \control^\numof_1, \ldots, \control^\numof_{\idxi_\numof}$
    and
    $\outputs = \outputs_1 \cup \cdots \cup \outputs_\numof$.
\end{definition}

In the above definition, the permutation condition ensures that the
\changed[mh]{
    target control states %
}%
are properly distributed amongst the newly created branches.

\begin{lemma}
\label{lem:alt-aut-pre}
    We have
    $\stackw \in
     \ap{\lang}{\canpopautbr{\control}
                            {\control_1,\ldots,\control_\numof}
                            {\outputs}}$
    iff
    \[
        \config{\tup{\control,
                     \outputs,
                     \control_1, \ldots, \control_\numof}}
               {\kstack{\pdaord}{\stackw}} \in
        \prestar{\alternating{\pda}}{\auta}
    \]
    where $\auta$ is such that
    \[
        \ap{\lang}{\auta} =
        \setcomp{
            \config{\tup{\control, \emptyset, \control}}
                   {\kstack{\pdaord}{}}
        }{
            \control \in \set{\control_1, \ldots, \control_\numof}
        } \ .
    \]
\end{lemma}
\acmeasychair{
    The proof of the above lemma is given in the full version~\cite{Hague:2015}.
}{
    \begin{proof}
        First take
$\stackw \in
 \ap{\lang}{\canpopautbr{\control}
                        {\control_1,\ldots,\control_\numof}
                        {\outputs}}$
and the run tree witnessing this membership.
We can move down the tree, maintaining a frontier
$\configc_1, \ldots, \configc_\altnumof$
and building a tree witnessing that
$\config{\tup{\control,
              \outputs,
              \control_1, \ldots, \control_\numof}}
        {\kstack{\pdaord}{\stackw}} \in
 \prestar{\alternating{\pda}}{\auta}$.
Initially we have the frontier
$\config{\control}{\kstack{\pdaord}{\stackw}}$
and the initial configuration
$\config{\tup{\control, \outputs, \control_1, \ldots, \control_\numof}}
        {\kstack{\pdaord}{\stackw}}$.

Hence, take a configuration
$\configc = \config{\control'}{\stackw'}$
from the frontier and corresponding configuration
$\configc' =
 \config{\tup{\control', \outputs', \control'_1, \ldots, \control'_\idxi}}
        {\stackw'}$.
If the rule applied to $\configc$ is not a branching rule, we simply take the matching rule of $\alternating{\pda}$ and apply it to $\configc'$.
Note, that if the rule output $\ochb$ we remove $\ochb$ from $\outputs'$.  Hence, $\outputs'$ contains only characters that have not been output on the path from the initial configuration.

If the rule applied is branching, that is
$\pdrule{\control'}
        {\cha}
        {\ech}
        {\rew{\cha}}
        {\control''_1, \ldots, \control''_\idxj}$
then we apply the rule
\[
    \altrule{\tup{\control',
                  \outputs,
                  \control'_1, \ldots, \control'_{\idxi}}}
            {\cha}
            {\tup{\control''_1,
                  \outputs_1,
                  \control^1_1, \ldots, \control^1_{\idxi_1}},
             \ldots
             \tup{\control''_\idxj,
                  \outputs_\idxj,
                  \control^\idxj_1, \ldots, \control^\idxj_{\idxi_\idxj}}}
\]
where
$\control'_1, \ldots, \control'_\idxi$
is a permutation of
$\control^1_1, \ldots, \control^1_{\idxi_1},
 \ldots
 \control^\idxj_1, \ldots, \control^\numof_{\idxi_\idxj}$
and
$\outputs = \outputs_1 \cup \cdots \cup \outputs_\numof$.
These partitions are made in accordance with the distribution of the leaves and outputs of the run tree of $\pda$.  I.e. if a control state $\control''$ appears on the $\idxi'$th subtree, then it should appear in the $\idxi'$th target state of $\alternating{\pda}$.  Similarly, if the $\idxi'$th subtree outputs an $\ochb \in \outputs$, then $\ochb$ should be placed in $\outputs_{\idxi'}$.
Applying this alternating transition creates a matching configuration for each new branch in the frontier.

We continue in this way until we reach the leaf nodes of the frontier.
Each leaf
$\config{\control'}{\stackw}$
has a matching
$\config{\tup{\control', \emptyset, \control'}}{\stackw}$
and hence is in
$\ap{\lang}{\auta}$.
Thus, we have witnessed
$\config{\tup{\control,
              \outputs,
              \control_1, \ldots, \control_\numof}}
        {\kstack{\pdaord}{\stackw}} \in
 \prestar{\alternating{\pda}}{\auta}$
as required.

To prove the other direction, we mirror the previous argument, showing that the witnessing tree for $\alternating{\pda}$ can be used to build a run tree of $\pda$.

    \end{proof}
}

\changed[mh]{
    It is known that
    $\prestar{\pda}{\auta}$
    is computable for alternating HOPDS.

    \begin{theorem}
        \cite[Theorem~1 (specialised)]{Broadbent:2012}
        Given an alternating HOPDS $\pda$ and a top-down automaton $\auta$, we can construct an
        automaton $\auta'$ accepting
        $\prestar{\pda}{\auta}$.
    \end{theorem}

}

Hence, we can now build
$\canpopautbr{\control}
             {\control_1,\ldots,\control_\numof}
             {\outputs}$
from the control state $\control$ and top-down automaton representation of
$\prestar{\alternating{\pda}}{\auta}$
since we can effectively translate from top-down to bottom-up stack automata.

\subsection{Reduction to Lower Orders}

We generalise our reduction to \bhopda{\pdaord}{\numchs}.
Let $\auttrue$ be the automata accepting all configurations.
Note, in the following definition we allow all transitions (including branching) to be labelled by sets of output characters.
To maintain our assumed normal form we have to replace these transitions using intermediate control states to ensure all branching transitions are labelled by $\ech$ and all transitions labelled $\outputs$ are replaced by a sequence of transitions outputting a single instance of each character in $\outputs$.

The construction follows the intuition of the single character case, but with a lot more bookkeeping.
\changed[mh]{
    Given an \bhopda{\pdaord}{\numchs} $\pda$ we define an \bhopda{\pdaord-1}{\numchs} with tests $\simulator{\pda}$ such that $\pda$ satisfies the diagonal problem iff $\simulator{\pda}$ also satisfies the diagonal problem.
    The main control states of $\simulator{\pda}$ take the form
    \[
        \simcontrolbr{\control}
                     {\control_1, \ldots, \control_\numof}
                     {\outputs}
                     {\branchchs}
    \]
    where $\control, \control_1, \ldots, \control_\numof$ are control states of $\pda$ and both $\outputs$ and $\branchchs$ are sets of output characters.
    We explain the purpose of each of these components.

    We will define $\simulator{\pda}$ to generate up to $\numof$ branches of the tree decomposition of a run of $\pda$.
    In particular, for each of the characters
    $\och \in \set{\och_1, \ldots, \och_\numchs}$
    there will be a branch of the run of $\simulator{\pda}$ responsible for outputting ``enough'' of the character $\och$ to satisfy the diagonal problem.
    Note that two characters $\och$ and $\och'$ may share the same branch.
    When a control state of the above form appears on a node of the run tree, the final component $\branchchs$ makes explicit which characters the subtree rooted at that node is responsible for generating in large numbers.
    Thus, the initial control state will have
    $\branchchs = \set{\och_1, \ldots, \och_\numchs}$
    since all characters must be generated from this node.
    However, when the output tree branches -- i.e. a node has more than one child -- the contents of $\branchchs$ will be partitioned amongst the children.
    That is, the responsibility of the parent to output enough of the characters in $\branchchs$ is divided amongst its children.

    The remaining components play the role of a test
    $\canpopautbr{\control}
                 {\control_1,\ldots,\control_\numof}
                 {\outputs}$.
    That is, the current node is simulating the control state $\control$ of $\pda$, and is required to produce $\numof$ branches, where the stack is emptied on each leaf and the control states appearing on these leaves are
    $\control_1, \ldots, \control_\numof$.
    Moreover, the tree should output at least one of each character in $\outputs$.

    Note, $\simulator{\pda}$ also has (external) tests of the form
    $\canpopautbr{\control}
                {\control_1,\ldots,\control_\numof}
                {\outputs}$
    that it can use to make decisions, just like in the single character case.
    However, it also performs tests ``online'' in its control states.
    This is necessary because the tests were used to check what could have happened on branches not followed by $\simulator{\pda}$.
    In the single character case, there was only one branch, hence $\simulator{\pda}$ would uses tests to check all the branches not followed, and then continue down a single branch of the tree.
    In the multi-character case the situation is different.
    Suppose a subtree rooted at a given node was responsible for outputting enough of both $\och_1$ and $\och_2$.
    Amongst the possible children of this node we may select two children: one for outputting enough $\och_1$ characters, and one for outputting enough $\och_2$ characters.
    The alternatives not taken will be checked using tests as before.
    However, the child responsible for outputting $\och_1$ may have also wanted to run a test on the child responsible for outputting $\och_2$.
    Thus, as well as having to output enough $\och_2$ characters, this latter child will also have to run the test required by the former.
    Thus, we have to build these tests into the control state.
    As a sanity condition we enforce
    $\outputs \cap \branchchs = \emptyset$
    since a branch outputting $\och$ should never ask itself if it is able to produce at least one $\och$.
}

We explain the rules of $\simulator{\pda}$ intuitively.
It will be beneficial to refer to the formal definition (below) while reading the explanations.
The case for $\simpushrules$ is illustrated in Figure~\ref{fig:push-sim} since it covers most of the situations appearing in the other rules as well.
\begin{itemize}
\item
    The rules in $\siminitrules$ guess how many branches will be needed to output enough of each $\och$.
    (This might be less than $\numchs$ since one branch might account for several characters.)

\item
    The rules in $\simfinrules$ check whether the run can be finished (always via a $\pop{\pdaord}$ since we are aiming for the empty stack).
    This is true if we only have one branch to complete (just reach $\control'$)
    and we have no more characters that we're obliged to output.

\item
    The rules in $\simsimrules$ simulate a non-branching operation.
    They do this faithfully, simply passing along all information (updating $\outputs$ if a character is output
    \changed[mh]{
        by the simulated transition).
    }

\item
    The rules in $\simbrrules$ are the first of the complicated rules.
    This is mainly a matter of notation.
    The reasoning behind the rules is that we're at a point where the tree splits into $\altnumof$ different branches.
    These have control states
    $\control'_1, \ldots, \control'_\altnumof$
    respectively.
    We non-deterministically guess which of these branches should output which of the characters in $\branchchs$.
    Thus, we split $\branchchs$ into
    $\branchchs_1, \ldots, \branchchs_\idxi$.
    This means we are exploring $\idxi$ branches.
    Let
    $\controlx_1, \ldots, \controlx_\idxi$
    be the control states on these branches.
    The remaining branches we handle using tests on the stack.
    Let
    $\controly_1, \ldots, \controly_\idxj$
    be the control states appearing on these branches.
    We require that all of
    $\control'_1, \ldots, \control'_\altnumof$
    are accounted for, so we assert that
    $\control'_1, \ldots, \control'_\altnumof$
    is a permutation of
    $\controlx_1, \ldots, \controlx_\idxi,
     \controly_1, \ldots, \controly_\idxj$.

    Similarly, in the current subtree we are obliged to pop to leaf nodes containing the control states
    $\control_1, \ldots, \control_\numof$.
    We split these obligations between the branches we are exploring and those we are handling using tests.
    We use another permutation check to ensure the obligations have been distributed properly.

    Finally, we are required to output characters in $\outputs$.
    We may also, in choosing a particular branch for a character $\och$, need to output $\och$ to account for instances appearing on a missed branch.
    Hence we also output $\outputs'$ to account for these.
    We distribute the obligations $\outputs$ and $\outputs'$ amongst the different branches using
    $\outputsx_1, \ldots, \outputsx_\idxi$
    and
    $\outputsy_1, \ldots, \outputsy_\idxj$.

\item
    The rules in $\simpushrules$ and $\simpoprules$ follow the same intuition as in the single character case, except we have the branching to deal with.
    In particular, at a push we have one branch corresponding to exploring what happens between the push and the corresponding pops, and a branch for each of the corresponding pops.
    We choose a selection of these branches to track with the HOPDA and a selection to handle using tests.
    The difference between $\simpushrules$ and $\simpoprules$ is that the former explores the branch of the push using the HOPDA and the latter uses a test.

    In these rules, after the push we're in control state $\control'$ and we guess that we will pop to control states
    $\control'_1, \ldots, \control'_\altnumof$.
    Hence we have a branch or a test to ensure that this happens.
    The remaining branches and tests are for what happens after the pops.
    The start from the states
    $\control'_1, \ldots, \control'_\altnumof$
    and must, in total, pop to the original pop obligation
    $\control_1, \ldots, \control_\numof$.
    Hence, we distribute these tasks in the same way as the $\simbrrules$.
\end{itemize}

\begin{figure*}
\begin{center}
\scalebox{0.8}{
    \acmeasychair{
        \psset{xunit=1.9,yunit=1.5}
    }{
        \psset{xunit=2,yunit=2}
    }
    \pspicture(0,-.5)(9,5)
        \pspolygon(4,3.5)(2,2)(6,2)
        \pspolygon(2,2)(4,2)(4,0)(0,0)
        \pspolygon(4,2)(6,2)(8,0)(4,0)
        \rput(4,4.2){$\control$}
        \rput(4,3.7){$\control'$}
        \psline{->}(4,4.1)(4,3.8)
        \rput(4.5,3.9){$\push{\pdaord}$}
        \rput(4,2.71){$\simcontrolbr{\control'}
                                   {\control'_1, \ldots, \control'_\altnumof}
                                   {\outputs}
                                   {\branchchs}$}
        \rput(4,2.3){$\control'_1, \ldots, \control'_\altnumof$}
        \rput(4,2.1){$\overbrace{\hspace{45ex}}$}
        \rput(6.5,2){$\pop{\pdaord}$}
        \rput(3,1.8){$\controly_1, \ldots, \controly_\idxj$}
        \rput(5,1.8){$\controlx_1, \ldots, \controlx_\idxi$}
        \rput(2,0.2){$\controly^1_1, \ldots, \controly^1_{\idxi_1},
                      \ldots,
                      \controly^\idxj_1, \ldots, \controly^\idxj_{\idxi_\idxj}$}
        \rput(6,0.2){$\controlx^1_1, \ldots, \controlx^1_{\idxj_1},
                      \ldots,
                      \controlx^\idxi_1, \ldots, \controlx^\idxi_{\idxj_\idxi}$}
        \rput(8.5,0){$\pop{\pdaord}$}
        \rput(4,-0.1){$\underbrace{\hspace{90ex}}$}
        \rput(4,-0.3){$\control_1, \ldots, \control_\numof$}
        \rput(2.75,1){$\brac{
            \begin{array}{c} %
                \canpopautbr{\controly_1} %
                            {\controly^1_1,%
                             \ldots,%
                             \controly^1_{\idxi_1}} %
                            {\outputsy_1} \\ %
                \cap \cdots \cap \\ %
                \canpopautbr{\controly_\idxj} %
                            {\controly^\idxj_1,%
                             \ldots,%
                             \controly^\idxj_{\idxi_\idxj}} %
                            {\outputsy_\idxj} %
            \end{array} %
        }$}
        \rput(5.25, 1){$\brac{
            \begin{array}{c} %
                \simcontrolbr{\controlx_1} %
                             {\controlx^1_1, %
                              \ldots, %
                              \controlx^1_{\idxj_1}} %
                             {\outputsx_1} %
                             {\branchchs_1}, \\ %
                 \ldots, \\ %
                 \simcontrolbr{\controlx_\idxi} %
                              {\controlx^\idxi_1, %
                               \ldots, %
                               \controlx^\idxi_{\idxj_\idxi}} %
                              {\outputsx_\idxi} %
                              {\branchchs_\idxi} %
            \end{array} %
        }$}
    \endpspicture
}
    \caption{\label{fig:push-sim}Illustrating the rules in $\simpushrules$.}
\end{center}
\end{figure*}

\changed[mh]{
    Before giving the formal definition, we summarise the discussion above by recalling the meaning of the various components.
    A control state
    $\simcontrolbr{\control}
                  {\control_1, \ldots, \control_\numof}
                  {\outputs}
                  {\branchchs}$
    means we're currently simulating a node at control state $\control$ that is required to produce $\numof$ branches terminating in control states $\control_1, \ldots, \control_\numof$ respectively, that the produced tree should output at least one of each character in $\outputs$ and the entire subtree should output enough of each character in $\branchchs$ to satisfy the diagonal problem.
    In the definition below, the set $\outputs'$ is the set of new single character output obligations produced when the automaton decides which branches to follow faithfully and which to test (for the output of at least one of each character).
    The sets
    $\outputsx_1, \ldots, \outputsx_\idxi$
    and
    $\outputsy_1, \ldots, \outputsy_\idxj$
    represent the partitioning of the single character output obligations amongst the tests and new branches.

    The correctness of the reduction is stated after the definition.
    A discussion of the proof appears in Section~\ref{sec:correctness-sim}.
}

\begin{nameddefinition}{def:sim-pda-sim}{$\simulator{\pda}$}
    Given an \bhopda{\pdaord}{\numchs}
    $\pda$
    described by
    $\tup{\controls,
          \oalphabet,
          \salphabet,
          \rules,
          \set{\uniquefinalcontrol},
          \controlinit,
          \chainit,
          \csrank}$
    and automata
    $\canpopautbr{\control}
                 {\control_1,\ldots,\control_\numof}
                 {\outputs}$
    for all %
    $1 \leq \numof \leq \numchs$,
    $\control, \control_1, \ldots, \control_\numof \in \controls$,
    and
    $\outputs \subseteq \set{\och_1, \ldots, \och_\numchs}$
    we define an \bhopda{\pdaord-1}{\numchs} with tests
    \[
        \simulator{\pda} = \tup{
            \simulator{\controls},
            \oalphabet,
            \salphabet,
            \simulator{\rules},
            \simulator{\finals},
            \simcontrolinit,
            \chainit,
            \simulator{\csrank}
        }
    \]
    where $\simulator{\controls}$ is the set
    \[
        \begin{array}{l} %
            \setcomp{ %
                \simcontrolbr{\control} %
                             {\control_1, \ldots, \control_\numof} %
                             {\outputs} %
                             {\branchchs} %
            }{ %
                \begin{array}{c} %
                    1 \leq \numof \leq \numchs %
                    \ \land \\ %
                    \control, \control_1, \ldots, \control_\numof \in \controls %
                    \ \land \\ %
                    \outputs, \branchchs \subseteq %
                    \set{\och_1, \ldots, \och_\numchs} %
                    \ \land \\ %
                    \outputs \cap \branchchs = \emptyset %
                \end{array} %
            } %
            \ \uplus \\ %
            \set{\simcontrolinit, \finalcontrol} %
        \end{array} %
    \] %
    and %
    \[ %
        \begin{array}{rcl} %
            \simulator{\rules} %
            &=& %
            \siminitrules \cup %
            \simsimrules \cup %
            \simbrrules \cup %
            \simfinrules \cup %
            \simpushrules \cup %
            \simpoprules %
            \\ %
            \simulator{\finals} %
            &=& %
            \set{\finalcontrol} %
         \end{array} %
    \] %
    and
    $\ap{\simulator{\csrank}}{\simcontrolbr{\control}
                                           {\control_1,
                                            \ldots,
                                            \control_\numof}
                                           {\outputs}
                                           {\branchchs}}
     =
     \sizeof{\branchchs}$
    and is $1$ for all other control states.
    We define the sets of rules,
         where in all cases,
        $\control_1, \ldots, \control_\numof \in \controls$
        and
        $\outputs, \outputs', \branchchs
         \subseteq
         \set{\och_1, \ldots, \och_\numchs}$, to be as follows:

    \begin{itemize}
    \item
        $\siminitrules$ is the set containing all rules of the form
        \[
            \pdrule{\simcontrolinit}
                   {\chainit}
                   {\ech}
                   {\rew{\chainit}}
                   {\simcontrolbr{\controlinit}
                                 {\uniquefinalcontrol,
                                  \ldots,
                                  \uniquefinalcontrol}
                                 {\emptyset}
                                 {\set{\och_1,\ldots,\och_\numchs}}}
        \]
        where
        $\sizeof{\uniquefinalcontrol,
                 \ldots,
                 \uniquefinalcontrol} \leq \numchs$,
        and

    \item
        $\simfinrules$ is the set containing all rules of the form
        \[
            \pdrulet{\simcontrolbr{\control}
                                  {\control'}
                                  {\emptyset}
                                  {\branchchs}}
                    {\cha}
                    {\auttrue}
                    {\ech}
                    {\rew{\cha}}
                    {\finalcontrol}
        \]
        for all
        $\pdrule{\control}
                 {\cha}
                 {\ech}
                 {\pop{\pdaord}}
                 {\control'} \in \rules$
        and
        $\branchchs \subseteq \set{\och_1, \ldots, \och_\numchs}$,
        and

    \item
         $\simsimrules$ is the set containing all rules of the form
         \[
             \bigpdrulet{\simcontrolbr{\control}
                                      {\control_1, \ldots, \control_\numof}
                                      {\outputs}
                                      {\branchchs}}
                        {\cha}
                        {\auttrue}
                        {\set{\ochb} \cap \branchchs}
                        {\op}
                        {\simcontrolbr{\control'}
                                      {\control_1, \ldots, \control_\numof}
                                      {\outputs \setminus \set{\ochb}}
                                      {\branchchs}}
        \]
        for
        $\pdrule{\control}
                {\cha}
                {\ochb}
                {\op}
                {\control'} \in \rules$,
        and
        $\op \notin \set{\push{\pdaord}, \pop{\pdaord}}$,
        and

    \item
        $\simbrrules$ is the set containing all rules of the form
        \[
            \bigpdrulet{\simcontrolbr{\control}
                                     {\control_1, \ldots, \control_\numof}
                                     {\outputs}
                                     {\branchchs}}
                       {\cha}
                       {\begin{array}{c} %
                           \canpopautbr{\controly_1} %
                                       {\controly^1_1,%
                                        \ldots,%
                                        \controly^1_{\idxi_1}} %
                                       {\outputsy_1} \\ %
                           \cap \cdots \cap \\ %
                           \canpopautbr{\controly_\idxj} %
                                       {\controly^\idxj_1,%
                                        \ldots,%
                                        \controly^\idxj_{\idxi_\idxj}} %
                                       {\outputsy_\idxj} %
                        \end{array}} %
                       {\outputs' \cap \branchchs} %
                       {\rew{\cha}} %
                       {\begin{array}{c} %
                           \simcontrolbr{\controlx_1} %
                                        {\controlx^1_1, %
                                         \ldots, %
                                         \controlx^1_{\idxj_1}} %
                                        {\outputsx_1} %
                                        {\branchchs_1}, \\ %
                            \ldots, \\ %
                            \simcontrolbr{\controlx_\idxi} %
                                         {\controlx^\idxi_1, %
                                          \ldots, %
                                          \controlx^\idxi_{\idxj_\idxi}} %
                                         {\outputsx_\idxi} %
                                         {\branchchs_\idxi} %
                        \end{array}} %
        \]
        where
        \[
            \pdrule{\control}
                   {\cha}
                   {\ech}
                   {\rew{\cha}}
                   {\control'_1, \ldots, \control'_\altnumof} \in \rules
        \]
        and
        $\control'_1, \ldots \control'_\altnumof$
        is a permutation of
        \[
            \controlx_1, \ldots, \controlx_\idxi,
            \controly_1, \ldots, \controly_\idxj
        \]
        and
        $\control_1, \ldots, \control_\numof$
        is a permutation of
        \[
            \controlx^1_1, \ldots, \controlx^1_{\idxj_1},
            \ldots
            \controlx^\idxi_1, \ldots, \controlx^\idxi_{\idxj_\idxi}
            \controly^1_1, \ldots, \controly^1_{\idxi_1},
            \ldots
            \controly^\idxj_1, \ldots, \controly^\idxj_{\idxi_\idxj}
        \]
        and
        \[
            \outputs \cup \outputs' =
            \outputsx_1 \cup \cdots \cup \outputsx_\idxi
            \cup
            \outputsy_1 \cup \cdots \cup \outputsy_\idxj
        \]
        and
        $\branchchs = \branchchs_1 \cup \cdots \cup \branchchs_\idxi$.

    \item
        $\simpushrules$ is the set containing all rules of the form
        \[
            \bigpdrulet{\simcontrolbr{\control}
                                     {\control_1, \ldots, \control_\numof}
                                     {\outputs}
                                     {\branchchs}}
                       {\cha}
                       {\begin{array}{c} %
                           \canpopautbr{\controly_1} %
                                       {\controly^1_1,%
                                        \ldots,%
                                        \controly^1_{\idxi_1}} %
                                       {\outputsy_1} \\ %
                           \cap \cdots \cap \\ %
                           \canpopautbr{\controly_\idxj} %
                                       {\controly^\idxj_1,%
                                        \ldots,%
                                        \controly^\idxj_{\idxi_\idxj}} %
                                       {\outputsy_\idxj} %
                        \end{array}} %
                       {\outputs' \cap \branchchs} %
                       {\rew{\cha}} %
                       {\begin{array}{c} %
                           \simcontrolbr{\control'} %
                                        {\control'_1, \ldots, \control'_\altnumof} %
                                        {\outputsx} %
                                        {\branchchs_0}, \\ %
                           \simcontrolbr{\controlx_1} %
                                        {\controlx^1_1, %
                                         \ldots, %
                                         \controlx^1_{\idxj_1}} %
                                        {\outputsx_1} %
                                        {\branchchs_1}, \\ %
                            \ldots, \\ %
                            \simcontrolbr{\controlx_\idxi} %
                                         {\controlx^\idxi_1, %
                                          \ldots, %
                                          \controlx^\idxi_{\idxj_\idxi}} %
                                         {\outputsx_\idxi} %
                                         {\branchchs_\idxi} %
                        \end{array}} %
        \]
        where
        \[
            \pdrule{\control}
                   {\cha}
                   {\ech}
                   {\push{\pdaord}}
                   {\control'}
        \]
        and
        $\control'_1, \ldots \control'_\altnumof$
        is a permutation of
        \[
            \controlx_1, \ldots, \controlx_\idxi,
            \controly_1, \ldots, \controly_\idxj
        \]
        and
        $\control_1, \ldots, \control_\numof$
        is a permutation of
        \[
            \controlx^1_1, \ldots, \controlx^1_{\idxj_1},
            \ldots
            \controlx^\idxi_1, \ldots, \controlx^\idxi_{\idxj_\idxi}
            \controly^1_1, \ldots, \controly^1_{\idxi_1},
            \ldots
            \controly^\idxj_1, \ldots, \controly^\idxj_{\idxi_\idxj}
        \]
        and
        \[
            \outputs \cup \outputs' =
            \outputsx \cup
            \outputsx_1 \cup \cdots \cup \outputsx_\idxi
            \cup
            \outputsy_1 \cup \cdots \cup \outputsy_\idxj
        \]
        and
        $\branchchs = \branchchs_0 \cup \cdots \cup \branchchs_\idxi$.

    \item
        we have $\simpoprules$ is the set containing all rules of the form
        \[
            \bigpdrulet{\simcontrolbr{\control}
                                     {\control_1, \ldots, \control_\numof}
                                     {\outputs}
                                     {\branchchs}}
                       {\cha}
                       {\begin{array}{c} %
                           \canpopautbr{\control'} %
                                       {\control'_1,\ldots,\control'_\altnumof} %
                                       {\outputsy}\ \cap \\ %
                           \canpopautbr{\controly_1} %
                                       {\controly^1_1,%
                                        \ldots,%
                                        \controly^1_{\idxi_1}} %
                                       {\outputsy_1} \\ %
                           \cap \cdots \cap \\ %
                           \canpopautbr{\controly_\idxj} %
                                       {\controly^\idxj_1,%
                                        \ldots,%
                                        \controly^\idxj_{\idxi_\idxj}} %
                                       {\outputsy_\idxj} %
                        \end{array}} %
                       {\outputs' \cap \branchchs} %
                       {\rew{\cha}} %
                       {\begin{array}{c} %
                           \simcontrolbr{\controlx_1} %
                                        {\controlx^1_1, %
                                         \ldots, %
                                         \controlx^1_{\idxj_1}} %
                                        {\outputsx_1} %
                                        {\branchchs_1}, \\ %
                            \ldots, \\ %
                            \simcontrolbr{\controlx_\idxi} %
                                         {\controlx^\idxi_1, %
                                          \ldots, %
                                          \controlx^\idxi_{\idxj_\idxi}} %
                                         {\outputsx_\idxi} %
                                         {\branchchs_\idxi} %
                        \end{array}} %
        \]
        where
        \[
            \pdrule{\control}
                   {\cha}
                   {\ech}
                   {\push{\pdaord}}
                   {\control'}
        \]
        and
        $\control'_1, \ldots \control'_\altnumof$
        is a permutation of
        \[
            \controlx_1, \ldots, \controlx_\idxi,
            \controly_1, \ldots, \controly_\idxj
        \]
        and
        $\control_1, \ldots, \control_\numof$
        is a permutation of
        \[
            \controlx^1_1, \ldots, \controlx^1_{\idxj_1},
            \ldots
            \controlx^\idxi_1, \ldots, \controlx^\idxi_{\idxj_\idxi}
            \controly^1_1, \ldots, \controly^1_{\idxi_1},
            \ldots
            \controly^\idxj_1, \ldots, \controly^\idxj_{\idxi_\idxj}
        \]
        and
        \[
            \outputs \cup \outputs' =
            \outputsy \cup
            \outputsx_1 \cup \cdots \cup \outputsx_\idxi
            \cup
            \outputsy_1 \cup \cdots \cup \outputsy_\idxj
        \]
        and
        $\branchchs = \branchchs_1 \cup \cdots \cup \branchchs_\idxi$.
    \end{itemize}
\end{nameddefinition}

\changed[mh]{
    In Section~\ref{sec:correctness-sim} we show that the reduction is correct.

    \begin{namedlemma}{lem:correct-sim-sim}{Correctness of $\simulator{\pda}$}
        \[
            \unbounded{\och_1, \ldots, \och_\numchs}{\pda}
            \iff
            \unbounded{\och_1, \ldots, \och_\numchs}{\simulator{\pda}}
        \]
    \end{namedlemma}
}

To complete the reduction, we convert the \bhopda{\pdaord}{\numchs} with tests into a \bhopda{\pdaord}{\numchs} without tests.

\begin{namedlemma}{lem:reduction-sim}{Reduction to Lower Orders}
    For every \bhopda{\pdaord}{\numchs} $\pda$ we can build an order-$(\pdaord-1)$ $\numchs$-branch HOPDA $\pda'$ such that
    \[
        \unbounded{\och_1, \ldots, \och_\numchs}{\pda}
        \iff
        \unbounded{\och_1, \ldots, \och_\numchs}{\pda'} \ .
    \]
\end{namedlemma}
\begin{proof}
    From \refdefinition{def:sim-pda-sim} and \reflemma{lem:correct-sim-sim}, we obtain from $\pda$ an \bhopda{\pdaord-1}{\numchs} with tests $\simulator{\pda}$ satisfying the conditions of the lemma.  To complete the proof, we invoke \reftheorem{thm:no-tests-sim} to find $\pda'$ as required.
\end{proof}

We show correctness of the reduction in Section~\ref{sec:correctness-sim}.
First we show that we have decidability once we have reduced to order-$0$.

\subsection{Decidability at Order-0}

We show that the problem becomes decidable for a \hopda{0} $\pda$.
This is essentially a finite state machine and we can linearise the trees generated by saving the list of states that have been branched to in the control state.
After one branch has completed, we run the next in the list, until all branches have completed.
Hence, a tree of $\pda$ becomes a run of the linearised \hopda{0}, and vice-versa.
Since each output tree has a bounded number of branches, the list length is bounded.
Thus, we convert $\pda$ into a finite state word automaton, for which the diagonal problem is decidable.
\changed[mh]{
    Note, this result can also be obtained from the decidability of the diagonal problem for pushdown automata.
}
\acmeasychair{
    The details are given in the full version~\cite{Hague:2015}.
}{

\begin{definition}[$\linearised{\pda}$]
    Given an \bhopda{0}{\numchs}
    $\pda$
    described by the tuple
    $\tup{\controls,
          \oalphabet,
          \salphabet,
          \rules,
          \finals,
          \controlinit,
          \chainit,
          \csrank}$
    we define a \hopda{0}
    \[
        \linearised{\pda} =
        \tup{\linearised{\controls},
             \oalphabet,
             \salphabet,
             \linearised{\rules},
             \finals,
             \controlinit,
             \chainit}
    \]
    such that
    \[
        \controls =
        \setcomp{
            \lincontrol{\control}
                       {\control_1, \cha_1,
                        \ldots
                        \control_\numof, \cha_\numof}
        }{
            \begin{array}{c} %
                \control, \control_1, \ldots, \control_\numof \in \controls %
                \ \land \\ %
                \cha_1, \ldots, \cha_\numof \in \salphabet %
                \ \land \\ %
                0 \leq \numof \leq \numchs %
            \end{array} %
        }
        \cup
        \set{\finalcontrol}
    \]
    and $\linearised{\rules}$ is the set containing all rules of the form
    \[
        \bigpdrule{\lincontrol{\control}
                              {\control_1, \cha_1,
                               \ldots,
                               \control_\numof, \cha_\numof}}
                  {\cha}
                  {\ochb}
                  {\rew{\chb}}
                  {\lincontrol{\control'_1}
                              {\begin{array}{c} %
                                   \control_1, \cha_1, %
                                   \ldots, %
                                   \control_\numof, \cha_\numof, \\ %
                                   \control'_2, \chb, %
                                   \ldots, %
                                   \control'_\altnumof, \chb %
                               \end{array}}} %
    \]
    for each
    \[
        \pdrule{\control}
               {\cha}
               {\ochb}
               {\rew{\chb}}
               {\control'_1, \ldots, \control'_\altnumof}
        \in \rules
    \]
    and all rules
    \[
        \pdrule{\lincontrol{\control}
                           {\control_1, \cha_1,
                            \ldots,
                            \control_\numof, \cha_\numof}}
               {\cha}
               {\ech}
               {\rew{\cha_1}}
               {\lincontrol{\control_1}
                           {\control_2, \cha_2,
                            \ldots,
                            \control_\numof, \cha_\numof}}
    \]
    whenever
    $\control \in \finals$.
\end{definition}

\begin{namedlemma}{lem:decide-order-0-sim}{Decidability at Order-0}
    We have
    \[
        \unbounded{\och_1, \ldots, \och_\numchs}{\pda}
        \iff
        \unbounded{\och_1, \ldots, \och_\numchs}{\linearised{\pda}}
    \]
    and hence
    $\unbounded{\och_1, \ldots, \och_\numchs}{\pda}$
    is decidable.
\end{namedlemma}
\begin{proof}
    Take an accepting run tree $\pdrun$ of $\pda$.
    If this tree contains no branching, then it is straightforward to construct an accepting run of $\linearised{\pda}$.
    Hence, assume all trees with fewer than $\numchs$ branches have a corresponding run of $\linearised{\pda}$.
    At a subtree
    $\treeap{\configc}{\tree_1, \ldots, \tree_\numof}$
    we take the run trees
    $\pdrun_1, \ldots, \pdrun_\numof$
    corresponding to the subtrees.
    Let
    $\configc = \config{\control}{\cha}$
    and
    $\configc_1 = \config{\control_1}{\cha},
     \ldots,
     \configc_\numof = \config{\control_\numof}{\cha}$
    be the configurations at the roots of the subtrees.
    We build a run beginning at $\configc$ and transitioning to
    $\config{\lincontrol{\control_1}
                        {\control_2, \cha,
                         \ldots,
                         \control_\numof, \cha}}
            {\cha}$.
    The run then follows $\pdrun_1$ with the extra information in its control state.
    After $\pdrun_1$ accepts, we transition to
    $\config{\lincontrol{\control_2}
                        {\control_3, \cha,
                         \ldots,
                         \control_\numof, \cha}}
            {\cha}$
    and then replay $\pdrun_2$.
    We repeat until all subtrees have been dispatched.
    This gives an accepting run of $\linearised{\pda}$ outputting the same number of each $\och$.

    In the other direction, we replay the accepting run $\pdrun$ of $\linearised{\pda}$ until we reach a configuration
    $\config{\lincontrol{\control_1}
                        {\control_2, \cha, \ldots, \control_\numof, \cha}}
            {\cha}$
    via a rule
    \[
        \pdrule{\control}
               {\chb}
               {\ech}
               {\rew{\cha}}
               {\lincontrol{\control_1}
                           {\control_2, \cha,
                            \ldots,
                            \control_\numof, \cha}}
    \]
    At this point we apply
    \[
        \pdrule{\control}
               {\chb}
               {\ech}
               {\rew{\cha}}
               {\control_1, \ldots, \control_\numof}
    \]
    of $\pda$.
    We obtain runs for each of the new children as follows.
    We split the remainder of the run $\pdrun'$ into $\numof$ parts
    $\pdrun'_1, \ldots, \pdrun'_\numof$
    where the break points correspond to each application of a rule of the second kind.
    For each $\idxi$ we replay the transitions of $\pdrun'_1$ from
    $\config{\control_\idxi}{\cha}$
    to obtain a new run of $\linearised{\pda}$ with fewer applications of the second rule.
    Inductively, we obtain an accepting run of $\pda$ that we plug into the $\idxi$th child.
    This gives us an accepting run of $\pda$ outputting the same number of each $\och$.
\end{proof}

}

\subsection{Decidability of The Diagonal Problem}

\acmeasychair{}{
    We thus have the following theorem.
}

\begin{namedtheorem}{thm:diagonal-sim}{Decidability of the Diagonal Problem}
    For an \hopda{\pdaord} $\pda$ and output characters
    $\och_1, \ldots, \och_\numchs$,
    it is decidable whether
    $\unbounded{\och_1, \ldots, \och_\numchs}{\pda}$.
\end{namedtheorem}

\begin{proof}
    We first interpret $\pda$ as an \bhopda{\pdaord}{\numchs} and then construct via \reflemma{lem:reduction-sim} an \bhopda{\pdaord-1}{\numchs} $\pda'$ such that
    $\unbounded{\och_1, \ldots, \och_\numchs}{\pda}$
    iff
    $\unbounded{\och_1, \ldots, \och_\numchs}{\pda'}$.
    We repeat this step until we have an \bhopda{0}{\numchs}.
    Then, from
    \acmeasychair{%
        decidability at order-$0$%
    }{%
        \reflemma{lem:decide-order-0-sim}%
    }
    we obtain decidability as required.
\end{proof}


\section{Correctness for Simultaneous Unboundedness}
\label{sec:correctness-sim}

In this section we prove
\changed[mh]{
    \reflemma{lem:correct-sim-sim}.
}
The proof follows the same outline as the single character case.
To show there is a run with at least $\numof$ of each character, we take via Lemma~\ref{lem:tree-scores-sim} (Section~\ref{sec:tree-scores-sim}),
$\numof' = (\numchs+1)^\numof$,
and a run of $\pda$ outputting at least this many of each character.
Then from Lemma~\ref{lem:branch-runs-sim} (Section~\ref{sec:branch-runs-sim}) a run of $\simulator{\pda}$ outputting at least $\numof$ of each character as required.
The other direction is shown in Lemma~\ref{lem:sim-to-pda-sim} (Section~\ref{sec:sim-to-pda-sim}).

We first generalise our tree decomposition and notion of scores.
We then show that every $\numchs$-branch subtree of a tree decomposition generates a run tree of $\simulator{\pda}$ matching the scores of the tree.
Finally we prove the opposite direction.

\subsection{Tree Decomposition of Output Trees}

Given an output tree $\tree$ of $\pda$ where each $\push{\pdaord}$ operation has a matching $\pop{\pdaord}$ on all branches, we can construct a decomposed tree representation of the run inductively as follows.
We define
$\treedecomp{\treesingle{\ech}} = \treesingle{\ech}$
and, when
\[
    \tree = \treeap{\ochb}{\tree_1, \ldots, \tree_\numof}
\]
where the rule applied at the root does not contain a $\push{\pdaord}$ operation, we have
\[
    \treedecomp{\tree} =
    \treeap{\ochb}{
        \treedecomp{\tree_1},
        \ldots,
        \treedecomp{\tree_\numof}
    } \ .
\]
In the final case, let
\[
    \tree = \treeap{\ech}{\tree'}
\]
where the rule applied at the root contains a $\push{\pdaord}$ operation and the corresponding $\pop{\pdaord}$ operations occur at nodes
$\node_1, \ldots, \node_\numof$.

Note, if the output trees had an arbitrary number of branches, $\numof$ may be unbounded.
In our case, $\numof \leq \numchs$, without which our reduction would fail: $\simulator{\pda}$ would be unable to accurately count the number of $\pop{\pdaord}$ nodes.
In fact, our trees would have unbounded out degree and \reflemma{lem:tree-scores} would not generalise.

Let $\tree_1, \ldots, \tree_\numof$ be the output trees rooted at $\node_1, \ldots, \node_\numof$ respectively and let $\tree'$ be $\tree$ with these subtrees removed.
Observe all branches of $\tree$ are cut by this operation since the $\push{\pdaord}$ must be matched on all branches.
We define
\[
    \treedecomp{\tree} =
    \treeap{\ech}{
        \treedecomp{\tree'},
        \treedecomp{\tree_1},
        \ldots,
        \treedecomp{\tree_\numof}
    } \ .
\]

An accepting run of $\pda$ has an extra $\pop{\pdaord}$ operation at the end of each branch leading to the empty stack.
Let $\tree'$ be the tree obtained by removing the final $\pop{\pdaord}$-induced edge leading to the leaves of each branch.
The tree decomposition of an accepting run is
\[
    \treedecomp{\tree} =
    \treeap{\ech}{
        \treedecomp{\tree'},
        \treesingle{\ech},
        \ldots,
        \treesingle{\ech}
    }
\]
where there are as many $\treesingle{\ech}$ as there are leaves of $\tree$.

Notice that our trees have out-degree at most $(\numchs + 1)$.

\subsection{Scoring Trees}
\label{sec:tree-scores-sim}

We score branches in the same way as the single character case.
We simply define
$\treescorech{\och}{\pdrun}$
to be
$\treescore{\pdrun}$
when $\och$ is considered as the only output character (all others are replaced with $\ech$).

We have to slightly modify our minimum score lemma to accommodate the increased out-degree of the nodes in the trees.

\begin{namedlemma}{lem:tree-scores-sim}{Minimum Scores}
    Given a tree $\tree$ with maximum out-degree $(\numchs + 1)$, containing, for each
    $\och \in \set{\och_1, \ldots, \och_\numchs}$,
    at least $\numof$ nodes labelled $\och$, for each
    $\och \in \set{\och_1, \ldots, \och_\numchs}$ we have
    \[
        \treescorech{\och}{\tree} \geq \ap{\log_{(\numchs+1)}}{\numof}
    \]
\end{namedlemma}
\begin{proof}
    This is a simple extension of the proof of \reflemma{lem:tree-scores}.
    We simply replace the two-child case with a tree with up to $(\numchs+1)$ children.
    In this case, we have to use $\log_{(\numchs+1)}$ rather than $\log$ to maintain the lemma.
\end{proof}

\acmeasychair{
    \subsection{Completing the proof}

    As in the single character case, we complete the proof with the following two lemmas, shown in the full version~\cite{Hague:2015}.
        As before, these lemmas simply formalise the fact that $\simulator{\pda}$ runs along branches of a tree decomposition of a run of $\pda$.
}{
    \subsection{From Branches to Runs}
}

\label{sec:branch-runs-sim}
\begin{namedlemma}{lem:branch-runs-sim}{Scores to Runs}
    Given an accepting output tree $\pdrun$ of $\pda$, if for all
    $\och \in \set{\och_1, \ldots, \och_\numchs}$ we have
    $\treescorech{\och}{\treedecomp{\pdrun}} \geq \numof$,
    then
    $\exists \tree \in \ap{\lang}{\simulator{\pda}}$
    with
    $\chcount{\och}{\tree} \geq \numof$
    for all
    $\och \in \set{\och_1, \ldots, \och_\numchs}$.
\end{namedlemma}
\acmeasychair{}{
    \begin{proof}
        We will construct a tree $\simulator{\pdrun}$ in
$\ap{\lang}{\simulator{\pda}}$
top down.
At each step we will maintain a ``frontier'' of $\simulator{\pdrun}$ and extend one leaf of this frontier until the whole tree is constructed.
The frontier is of the form
\[
    \tup{\configc_1, \node_1, \outputs_1, \branchchs_1,
         \ldots,
         \configc_\altnumof,
         \node_\altnumof,
         \outputs_\altnumof
         \branchchs_\altnumof}
\]
which means that there are $\altnumof$ nodes in the frontier.
We have
$\branchchs_1 \uplus \cdots \uplus \branchchs_\altnumof =
 \set{\och_1, \ldots, \och_\numchs}$
and each $\branchchs_\idxi$ indicates that the $\idxi$th branch, ending in configuration $\configc_\idxi$, is responsible for outputting enough of each of the characters in $\branchchs_\idxi$.
Each $\node_\idxi$ is the corresponding node in $\treedecomp{\pdrun}$ that is being tracked by the $\idxi$th branch of the output of $\simulator{\pda}$.

Let
$\uniquefinalcontrol$
be the final (accepting) control state of $\pda$
and let
$\tree = \treedecomp{\pdrun}$.
We begin at the root node of $\tree$, which corresponds to the initial configuration of $\pdrun$.
Let
$\config{\control}{\kstack{\pdaord}{\stackw}}$
be this initial configuration and let
$\configc =
 \config{\simcontrolbr{\control}
                      {\uniquefinalcontrol, \ldots, \uniquefinalcontrol}
                      {\emptyset}
                      {\set{\och_1, \ldots, \och_\numchs}}}
        {\stackw}$
be the configuration of $\simulator{\pda}$ after an application of a rule from
$\siminitrules$.
The initial frontier is
$\tup{\configc, \rootnode, \set{\och_1, \ldots, \och_\numchs}}$.

Thus, assume we have a frontier
\[
    \tup{\configc_1, \node_1, \outputs_1, \branchchs_1,
      \ldots,
      \configc_\fronlen,
      \node_\fronlen,
      \outputs_\fronlen,
      \branchchs_\fronlen}
\]
and for each of the sequences
$\simulator{\configc}, \node, \outputs, \branchchs$
of the frontier we have
\begin{enumerate}
\item
    $\tree'$ is the subtree of $\tree$ rooted at $\node$, and

\item
    $\configc = \config{\control}{\stackw}$ labelling $\node$, and

\item
    $\simulator{\configc} =
     \config{\simcontrolbr{\control}
                          {\control_1,
                           \ldots,
                           \control_\numof}
                          {\outputs}
                          {\branchchs}}
            {\ap{\topop{\pdaord}}{\stackw}}$, and

\item
    the node of $\pdrun$ corresponding to $\node$ has $\numof$ locations where the $\topop{\pdaord}$ stack is first popped via rules reaching
    $\control_1, \ldots, \control_\numof$,
    moreover, these leaves have corresponding leaves in $\tree'$, and

\item
    the branch from the root of the constructed run to the node labelled $\simulator{\configc}$ in the frontier outputs, for each
    $\och \in \branchchs$,
    at least
    $\brac{\numof - \treescorech{\och}{\tree'}}$
    occurrences of $\och$, and

\item
    $\outputs \cap \branchchs = \emptyset$
    and for each
    $\och \in \outputs$
    there is at least one node labelled by $\och$ in $\tree'$.
\end{enumerate}

Pick such a sequence
$\simulator{\configc}, \node, \outputs, \branchchs$.
We replace this sequence using a transition of $\simulator{\pda}$ in a way that produces a new frontier with the above properties and moves us a step closer to reaching leaves of $\tree$.
There are three cases when we are dealing with internal nodes.
\begin{itemize}
\item
    $\tree' = \treeap{\ochb}{\tree_1}$.

    In this case there is a transition
    $\configc \pdatran{\ochb} \configc'$
    via a rule
    $\pdrule{\control}{\cha}{\ochb}{\op}{\control'}$
    where
    $\op \notin \set{\push{\pdaord}, \pop{\pdaord}}$.
    Hence, we have
    \[
        \bigpdrulet{\simcontrolbr{\control}
                                 {\control_1,
                                  \ldots
                                  \control_\numof}
                                 {\outputs}
                                 {\branchchs}}
                   {\cha}
                   {\auttrue}
                   {\set{\ochb} \cap \branchchs}
                   {\op}
                   {\simcontrolbr{\control'}
                                 {\control_1,
                                  \ldots
                                  \control_\numof}
                                 {\outputs \setminus \set{\ochb}}
                                 {\branchchs}}
    \]
    in $\simulator{\pda}$ and thus we can extend $\simulator{\pdrun}$ with a transition
    $\simulator{\configc} \pdatran{\ochb} \simulator{\configc'}$
    via this rule.
    The new frontier is obtained by replacing
    $\simulator{\configc}, \node, \outputs, \branchchs$
    with
    $\simulator{\configc'}, \node', \outputs \setminus \set{\ochb}, \branchchs$
    where $\node'$ is the child of $\node$.
    The properties on the frontier are easily seen to be retained.

\item
    $\tree' = \treeap{\ech}{\tree_1, \ldots, \tree_\altnumof}$
    from a rule
    $\pdrule{\control}
            {\cha}
            {\ech}
            {\rew{\cha}}
            {\control'_1, \ldots, \control'_\altnumof}$
    of $\pda$.

    We separate
    $\branchchs = \branchchs'_1 \uplus \cdots \uplus \branchchs'_\idxi$
    such that
    $\branchchs'_\idxj$
    is the set of characters $\och$ that have their score derived from $\tree_\idxj$
    (i.e. the subtree with the higher score for $\och$ characters).
    Let $\outputs'$ be the set of all $\och$ who had a $+1$ in their score derived from another subtree.
    Let
    $\config{\controlx_1}{\stackw},
     \ldots
     \config{\controlx_\idxi}{\stackw}$
    be the configurations labelling the root nodes
    $\node_0, \node_1, \ldots, \node_\idxi$
    of these subtrees.
    Let
    $\config{\controly_1}{\stackw},
     \ldots,
     \config{\controly_\idxj}{\stackw}$
    be the configurations labelling the root nodes of the remaining subtrees.
    Since $\tree'$ includes $\numof$ leaves that are followed in $\pdrun$ by pops to
    $\control_1, \ldots, \control_\numof$
    we can distribute these control states amongst the branches, obtaining
    \[
        \controlx^1_1, \ldots, \controlx^1_{\idxj_1},
        \ldots
        \controlx^\idxi_1, \ldots, \controlx^\idxi_{\idxj_\idxi}
        \controly^1_1, \ldots, \controly^1_{\idxi_1},
        \ldots
        \controly^\idxj_1, \ldots, \controly^\idxj_{\idxi_\idxj} \ .
    \]
    Finally, we can distribute
    \[
        \outputs \cup \outputs' =
        \outputsx_1 \cup \cdots \cup \outputsx_\idxi
        \cup
        \outputsy_1 \cup \cdots \cup \outputsy_\idxj
    \]
    amongst the subtrees
    $\tree_1, \ldots, \tree_\altnumof$
    since $\outputs$ can be distributed by assumption and we chose $\outputs'$ such that this can be done.

    From the runs corresponding to
    $\tree_1, \ldots, \tree_\altnumof$
    and our choices above we know that the tests will pass.
    That is,
    $\simulator{\configc} \in
     \ap{\lang}{\canpopautbr{\control'_1}
                            {\controly^1_1,%
                             \ldots,%
                             \controly^1_{\idxi_1}}
                            {\outputsy_1}}$,
    \ldots,
    $\simulator{\configc} \in
     \ap{\lang}{\canpopautbr{\control'_\idxj}
                            {\controly^\idxj_1,%
                             \ldots,%
                             \controly^\idxj_{\idxi_\idxj}}
                            {\outputsy_\idxj}}$.

    Hence, we apply to $\simulator{\configc}$ the rule
    \[
        \bigpdrulet{\simcontrolbr{\control}
                                 {\control_1, \ldots, \control_\numof}
                                 {\outputs}
                                 {\branchchs}}
                   {\cha}
                   {\begin{array}{c} %
                       \canpopautbr{\controly_1} %
                                   {\controly^1_1, %
                                    \ldots, %
                                    \controly^1_{\idxi_1}} %
                                   {\outputsy_1} \\ %
                       \cap \cdots \cap \\ %
                       \canpopautbr{\controly_\idxj} %
                                   {\controly^\idxj_1, %
                                    \ldots, %
                                    \controly^\idxj_{\idxi_\idxj}} %
                                   {\outputsy_\idxj} %
                    \end{array}} %
                   {\outputs' \cap \branchchs} %
                   {\rew{\cha}} %
                   {\begin{array}{c} %
                       \simcontrolbr{\controlx_1} %
                                    {\controlx^1_1, %
                                     \ldots, %
                                     \controlx^1_{\idxj_1}} %
                                    {\outputsx_1} %
                                    {\branchchs_1}, \\ %
                        \ldots, \\ %
                        \simcontrolbr{\controlx_\idxi} %
                                     {\controlx^\idxi_1, %
                                      \ldots, %
                                      \controlx^\idxi_{\idxj_\idxi}} %
                                     {\outputsx_\idxi} %
                                     {\branchchs_\idxi} %
                    \end{array}} %
    \]
    and obtain configurations
    $\simulator{\configc^1},
     \ldots,
     \simulator{\configc^\idxi}$
    and a new frontier satisfying the required properties by replacing
    $\simulator{\configc}, \node, \outputs, \branchchs$
    with the sequence
    \[
        \simulator{\configc^1}, \node_1, \outputsx_1, \branchchs'_1,
        \ldots
        \simulator{\configc^\idxi},
        \node_\idxi,
        \outputsx_\idxi,
        \branchchs'_\idxi \ .
    \]

\item
    $\tree' = \treeap{\ech}{\tree_1, \ldots, \tree_\altnumof}$
    not from a rule
    $\pdrule{\control}
            {\cha}
            {\ech}
            {\rew{\cha}}
            {\control'_1, \ldots, \control'_\altnumof}$
    of $\pda$.

    In this case we have that $\tree'$ (subtree of the decomposition $\tree$) corresponds to a run tree $\pdrun_{\tree'}$ that can be decomposed into
    \begin{itemize}
    \item
        $\treeap{\configc}{\pdrun'}$ with
        $\configc' = \config{\control'}{\ap{\push{\pdaord}}{\stackw}}$
        at the root of $\pdrun'$ via a rule
        $\pdrule{\control}{\cha}{\ech}{\push{\pdaord}}{\control'}$
        and $\altnumof$ leaf nodes labelled
        $\configc_1, \ldots, \configc_\altnumof$
        respectively, and

    \item
        runs
        $\pdrun_1, \ldots, \pdrun_\altnumof$
        with the roots labelled
        $\configc'_1 = \config{\control'_1}{\stackw},
         \ldots,
         \configc'_\altnumof = \config{\control'_\altnumof}{\stackw}$
        where, for each $\idxi$, we have
        $\configc_\idxi \pdatran{\ech} \configc'_\idxi$
        via a $\pop{\pdaord}$ rule, and these are the first points $\stackw$ is seen along each branch, and

    \item
        the leaves of
        $\pdrun_1, \ldots, \pdrun_\altnumof$
        are the leaves of $\pdrun_{\tree'}$.
    \end{itemize}

    There are two cases depending on whether we send the HOPDA down the branch corresponding to the push.
    \begin{itemize}
    \item
        We separate
        $\branchchs = \branchchs'_0 \uplus
                      \branchchs'_1 \uplus
                      \cdots \uplus
                      \branchchs'_\idxi$
        such that
        $\branchchs'_\idxj$
        is the set of characters $\och$ that have their score derived from $\tree_\idxj$
        (i.e. the subtree with the higher score for $\och$ characters).
        Assume $\tree_1$ is amongst these subtrees (and will get $\branchchs'_0$).
        Let $\outputs'$ be the set of all $\och$ who had a $+1$ in their score derived from another subtree.
        Let
        $\config{\control'}{\ap{\push{\pdaord}}{\stackw}},
         \config{\controlx_1}{\stackw},
         \ldots
         \config{\controlx_\idxi}{\stackw}$
        be the configurations labelling the root nodes
        $\node_1, \ldots, \node_\idxi$
        of these subtrees, with the first belonging to $\tree_1$.
        Let
        $\config{\controly_1}{\stackw},
         \ldots,
         \config{\controly_\idxj}{\stackw}$
        be the configurations labelling the root nodes of the remaining subtrees.
        Since $\tree'$ has $\numof$ leaves that are followed in $\pdrun$ by pops to
        $\control_1, \ldots, \control_\numof$
        we can distribute these control states amongst the branches, obtaining
        \[
            \controlx^1_1, \ldots, \controlx^1_{\idxj_1},
            \ldots
            \controlx^\idxi_1, \ldots, \controlx^\idxi_{\idxj_\idxi}
            \controly^1_1, \ldots, \controly^1_{\idxi_1},
            \ldots
            \controly^\idxj_1, \ldots, \controly^\idxj_{\idxi_\idxj} \ .
        \]
        We can also distribute
        \[
            \outputs \cup \outputs' =
            \outputsx \cup
            \outputsx_1 \cup \cdots \cup \outputsx_\idxi
            \cup
            \outputsy_1 \cup \cdots \cup \outputsy_\idxj
        \]
        amongst the subtrees
        $\tree_1, \ldots, \tree_\altnumof$
        with $\outputsx$ belonging to $\tree_1$ since $\outputs$ can be distributed by assumption and we chose $\outputs'$ such that this can be done.

        From the existence of the runs
        $\pdrun_1, \ldots, \pdrun_\altnumof$
        we know
        $\simulator{\configc} \in
         \ap{\lang}{\canpopautbr{\control'_1}
                                {\controly^1_1,%
                                 \ldots,%
                                 \controly^1_{\idxi_1}}
                                {\outputsy_1}}$,
        \ldots,
        $\simulator{\configc} \in
         \ap{\lang}{\canpopautbr{\control'_\idxj}
                                {\controly^\idxj_1,%
                                 \ldots,%
                                 \controly^\idxj_{\idxi_\idxj}}
                                {\outputsy_\idxj}}$.

        Hence, we apply to $\simulator{\configc}$ the rule
        \[
            \bigpdrulet{\simcontrolbr{\control}
                                     {\control_1, \ldots, \control_\numof}
                                     {\outputs}
                                     {\branchchs}}
                       {\cha}
                       {\begin{array}{c} %
                           \canpopautbr{\controly_1} %
                                       {\controly^1_1, %
                                        \ldots, %
                                        \controly^1_{\idxi_1}} %
                                       {\outputsy_1} \\ %
                           \cap \cdots \cap \\ %
                           \canpopautbr{\controly_\idxj} %
                                       {\controly^\idxj_1, %
                                        \ldots, %
                                        \controly^\idxj_{\idxi_\idxj}} %
                                       {\outputsy_\idxj} %
                        \end{array}} %
                       {\outputs' \cap \branchchs} %
                       {\rew{\cha}} %
                       {\begin{array}{c} %
                           \simcontrolbr{\control'} %
                                        {\control'_1, \ldots, \control'_\altnumof} %
                                        {\outputsx} %
                                        {\branchchs_0}, \\ %
                           \simcontrolbr{\controlx_1} %
                                        {\controlx^1_1, %
                                         \ldots, %
                                         \controlx^1_{\idxj_1}} %
                                        {\outputsx_1} %
                                        {\branchchs_1}, \\ %
                            \ldots, \\ %
                            \simcontrolbr{\controlx_\idxi} %
                                         {\controlx^\idxi_1, %
                                          \ldots, %
                                          \controlx^\idxi_{\idxj_\idxi}} %
                                         {\outputsx_\idxi} %
                                         {\branchchs_\idxi} %
                        \end{array}} %
        \]
        and obtain configurations
        $\simulator{\configc^0},
         \simulator{\configc^1},
         \ldots,
         \simulator{\configc^\idxi}$
        and a new frontier satisfying the required properties by replacing
        $\simulator{\configc}, \node, \outputs, \branchchs$
        with the sequence
        \[
            \simulator{\configc^0}, \node_0, \outputsx, \branchchs'_0,
            \simulator{\configc^1}, \node_1, \outputsx_1, \branchchs'_1,
            \ldots
            \simulator{\configc^\idxi},
            \node_\idxi,
            \outputsx_\idxi,
            \branchchs'_\idxi \ .
        \]

    \item
        We separate
        $\branchchs = \branchchs'_1 \uplus \cdots \uplus \branchchs'_\idxi$
        such that
        $\branchchs'_\idxj$
        is the set of characters $\och$ that have their score derived from $\tree_\idxj$
        (i.e. the subtree with the higher score for $\och$ characters).
        Assume $\tree_1$ is not amongst these subtrees.
        Let $\outputs'$ be the set of all $\och$ who had a $+1$ in their score derived from another subtree.
        Let
        $\config{\controlx_1}{\stackw},
         \ldots
         \config{\controlx_\idxi}{\stackw}$
        be the configurations labelling the root nodes
        $\node_1, \ldots, \node_\idxi$
        of these subtrees.
        Let
        $\config{\control'}{\ap{\push{\pdaord}}{\stackw}},
         \config{\controly_1}{\stackw},
         \ldots,
         \config{\controly_\idxj}{\stackw}$
        be the configurations labelling the root nodes of the remaining subtrees, with the first belonging to $\tree_1$.
        Since $\tree'$ has $\numof$ leaves that are followed in $\pdrun$ by pops to
        $\control_1, \ldots, \control_\numof$
        we can distribute these control states amongst the branches, obtaining
        \[
            \controlx^1_1, \ldots, \controlx^1_{\idxj_1},
            \ldots
            \controlx^\idxi_1, \ldots, \controlx^\idxi_{\idxj_\idxi}
            \controly^1_1, \ldots, \controly^1_{\idxi_1},
            \ldots
            \controly^\idxj_1, \ldots, \controly^\idxj_{\idxi_\idxj} \ .
        \]
        We can also distribute
        \[
            \outputs \cup \outputs' =
            \outputsx_1 \cup \cdots \cup \outputsx_\idxi
            \cup
            \outputsy
            \cup
            \outputsy_1 \cup \cdots \cup \outputsy_\idxj
        \]
        amongst the subtrees
        $\tree_1, \ldots, \tree_\altnumof$
        with $\outputsy$ belonging to $\tree_1$ since $\outputs$ can be distributed by assumption and we chose $\outputs'$ such that this can be done.

        From the existence of $\pdrun'$ we know that
        $\simulator{\configc} \in
         \ap{\lang}{\canpopautbr{\control'}
                                {\control'_1,\ldots,\control'_\altnumof}
                                {\outputsy}}$
        and from the existence of
        $\pdrun_1, \ldots, \pdrun_\altnumof$
        we also know
        $\simulator{\configc} \in
         \ap{\lang}{\canpopautbr{\control'_1}
                                {\controly^1_1,%
                                 \ldots,%
                                 \controly^1_{\idxi_1}}
                                {\outputsy_1}}$,
        \ldots,
        $\simulator{\configc} \in
         \ap{\lang}{\canpopautbr{\control'_\idxj}
                                {\controly^\idxj_1,%
                                 \ldots,%
                                 \controly^\idxj_{\idxi_\idxj}}
                                {\outputsy_\idxj}}$.

        Hence, we apply to $\simulator{\configc}$ the rule
        \[
            \bigpdrulet{\simcontrolbr{\control}
                                     {\control_1, \ldots, \control_\numof}
                                     {\outputs}
                                     {\branchchs}}
                       {\cha}
                       {\begin{array}{c} %
                           \canpopautbr{\control'} %
                                       {\control'_1,\ldots,\control'_\altnumof} %
                                       {\outputsy} \ \cap \\ %
                           \canpopautbr{\controly_1} %
                                       {\controly^1_1,%
                                        \ldots,%
                                        \controly^1_{\idxi_1}} %
                                       {\outputsy_1} \\ %
                           \cap \cdots \cap \\ %
                           \canpopautbr{\controly_\idxj}%
                                       {\controly^\idxj_1,%
                                        \ldots,%
                                        \controly^\idxj_{\idxi_\idxj}} %
                                       {\outputsy_\idxj} %
                        \end{array}} %
                       {\outputs' \cap \branchchs} %
                       {\rew{\cha}} %
                       {\begin{array}{c} %
                           \simcontrolbr{\controlx_1} %
                                        {\controlx^1_1, %
                                         \ldots, %
                                         \controlx^1_{\idxj_1}} %
                                        {\outputsx_1} %
                                        {\branchchs_1}, \\ %
                            \ldots, \\ %
                            \simcontrolbr{\controlx_\idxi} %
                                         {\controlx^\idxi_1, %
                                          \ldots, %
                                          \controlx^\idxi_{\idxj_\idxi}} %
                                         {\outputsx_\idxi} %
                                         {\branchchs_\idxi} %
                        \end{array}} %
        \]
        and obtain configurations
        $\simulator{\configc^1},
         \ldots,
         \simulator{\configc^\idxi}$
        and a new frontier satisfying the required properties by replacing
        $\simulator{\configc}, \node, \outputs, \branchchs$
        with the sequence
        \[
            \simulator{\configc^1}, \node_1, \outputsx_1, \branchchs'_1,
            \ldots
            \simulator{\configc^\idxi},
            \node_\idxi,
            \outputsx_\idxi,
            \branchchs'_\idxi \ .
        \]
    \end{itemize}
\end{itemize}

Finally, we reach a leaf node $\node$ with a run outputting the required number of $\och$s.
We need to show that the run constructed is accepting.
From the tree decomposition, we know that the corresponding node of $\pdrun$ is immediately followed by a $\pop{\pdaord}$.
Thus, from our conditions on the frontier, we must have
$\numof = 1$
and
$\outputs = \emptyset$.
We also have a rule
$\pdrule{\control}{\cha}{\ech}{\pop{\pdaord}}{\control_1}$
and therefore
$\pdrulet{\simcontrolbr{\control}
                       {\control_1}
                       {\emptyset}
                       {\branchchs}}
         {\cha}
         {\auttrue}
         {\ech}
         {\rew{\cha}}
         {\finalcontrol}$
with which we can complete the run of $\simulator{\pda}$ as required.

    \end{proof}

    \subsection{The Other Direction}

    Finally, we need to show that each accepting run tree of $\simulator{\pda}$ gives rise to an accepting run tree of $\pda$ containing at least as many of each output character $\och$.
}

\label{sec:sim-to-pda-sim}
\begin{lemma}[$\simulator{\pda}$ to $\pda$]
\label{lem:sim-to-pda-sim}
    We have
    \acmeasychair{
        \[
            \unbounded{\och_1, \ldots, \och_\numchs}{\simulator{\pda}}
            \Rightarrow
            \unbounded{\och_1, \ldots, \och_\numchs}{\pda} \ .
        \]
    }{
        $\unbounded{\och_1, \ldots, \och_\numchs}{\simulator{\pda}}$
        implies
        $\unbounded{\och_1, \ldots, \och_\numchs}{\pda}$.
    }
\end{lemma}
\acmeasychair{}{
    \begin{proof}
        Take an accepting run tree $\simulator{\pdrun}$ of $\simulator{\pda}$.
We show that there exists a corresponding run tree $\pdrun$ of $\pda$ outputting at least as many $\och$s.

We maintain a frontier
\[
    \configc_1, \ldots, \configc_\fronlen
\]
of $\simulator{\pdrun}$ and a run $\pdrun$ of $\pda$ ``with holes'' such that
\begin{itemize}
\item
    there are $\fronlen$ nodes of $\pdrun$ labelled by
    $\configc_1, \ldots, \configc_\fronlen$
    respectively (these are the holes), and

\item
    each of these holes labelled $\configc$ is the only child of a parent node labelled $\configc'$ of $\pda$, and

\item
    for each corresponding pair $\configc$ and $\configc'$ we have
    \begin{itemize}
    \item
        $\configc' = \config{\control}{\stackw}$,
        and

    \item
        $\configc =
         \config{\simcontrolbr{\control}
                              {\control_1, \ldots, \control_\numof}
                              {\outputs}
                              {\branchchs}}
                {\ap{\topop{\pdaord}}{\stackw}}$,
        and

    \item
        the node labelled by $\configc$ has $\numof$ children with the $\idxi$th child being labelled
        $\config{\control_\idxi}{\ap{\pop{\pdaord}}{\stackw}}$,
        and

    \item
        all leaf nodes of $\pdrun$ are accepting, and

    \item
        for each
        $\och \in \set{\och_1, \ldots, \och_\numchs}$
        the number of $\och$ output by run tree of $\pda$ is at least as many as on the branch of $\simulator{\pda}$ to the configuration with
        $\och \in \branchchs$
        less 1 if
        $\och \in \outputs$.
    \end{itemize}
\end{itemize}

Initially after a rule from $\siminitrules$ we have the frontier
$\configc =
 \config{\simcontrolbr{\control}
                      {\uniquefinalcontrol, \ldots, \uniquefinalcontrol}
                      {\emptyset}
                      {\set{\och_1, \ldots, \och_\numchs}}}
        {\stackw}$
with corresponding run $\pdrun$ of $\pda$ being
\[
    \treeap{\config{\control}{\kstack{\pdaord}{\stackw}}}{
        \treeap{\configc}{
            \config{\uniquefinalcontrol}{\kstack{\pdaord}{}},
            \ldots,
            \config{\uniquefinalcontrol}{\kstack{\pdaord}{}}
        }
    } \ .
\]

Pick a configuration
$\simulator{\configc} =
 \config{\simcontrolbr{\control}
                      {\control_1, \ldots, \control_\numof}
                      {\outputs}
                      {\branchchs}}
        {\ap{\topop{\pdaord}}{\stackw}}$
of the frontier that is not a leaf of $\simulator{\pdrun}$ and its corresponding node in $\pdrun$ with parent labelled
$\configc = \config{\control}{\stackw}$.
Let $\simulator{\pdrun'}$ be the subtree of $\simulator{\pda}$ rooted at this configuration.

We show how to extend the frontier closer to the leaves of $\simulator{\pdrun}$.
There are several cases depending on the transition of $\simulator{\pda}$ used to exit our chosen node.
\begin{itemize}
\item
    $\simulator{\pdrun'} =
     \treeap{\simulator{\configc}}{\simulator{\pdrun^1}}$
    and the rule applied is of the form
    \[
         \bigpdrulet{\simcontrolbr{\control}
                                  {\control_1, \ldots, \control_\numof}
                                  {\outputs}
                                  {\branchchs}}
                    {\cha}
                    {\auttrue}
                    {\set{\ochb} \cap \branchchs}
                    {\op}
                    {\simcontrolbr{\control'}
                                  {\control_1, \ldots, \control_\numof}
                                  {\outputs \setminus \set{\ochb}}
                                  {\branchchs}} \ .
    \]
    Let $\simulator{\configc'}$ be the configuration labelling the root of
    $\simulator{\pdrun^1}$.
    We have
    $\pdrule{\control}
            {\cha}
            {\ochb}
            {\op}
            {\control'} \in \rules$
    and
    $\op \notin \set{\push{\pdaord}, \pop{\pdaord}}$.
    We can apply
    $\configc \pdatran{\ochb} \configc'$.
    Let $\node$ be the node labelled $\simulator{\configc}$.
    We insert above $\node$ a node labelled $\configc'$.
    Then we change the label of $\node$  to $\simulator{\configc'}$.
    We keep the same children of $\node$.
    This extended run maintains all properties as required.

\item
    $\simulator{\pdrun'} =
     \treeap{\simulator{\configc}}
            {\simulator{\pdrun^1},
             \ldots,
             \simulator{\pdrun^\idxi}}$
    via a rule
    \[
        \bigpdrulet{\simcontrolbr{\control}
                                 {\control_1, \ldots, \control_\numof}
                                 {\outputs}
                                 {\branchchs}}
                   {\cha}
                   {\begin{array}{c} %
                       \canpopautbr{\controly_1} %
                                   {\controly^1_1,
                                    \ldots,
                                    \controly^1_{\idxi_1}} %
                                   {\outputsy_1} \\ %
                       \cap \cdots \cap \\ %
                       \canpopautbr{\controly_\idxj} %
                                   {\controly^\idxj_1,
                                    \ldots,
                                    \controly^\idxj_{\idxi_\idxj}} %
                                   {\outputsy_\idxj} %
                    \end{array}} %
                   {\outputs' \cap \branchchs} %
                   {\rew{\cha}} %
                   {\begin{array}{c} %
                       \simcontrolbr{\controlx_1} %
                                    {\controlx^1_1, %
                                     \ldots, %
                                     \controlx^1_{\idxj_1}} %
                                    {\outputsx_1} %
                                    {\branchchs_1}, \\ %
                        \ldots, \\ %
                        \simcontrolbr{\controlx_\idxi} %
                                     {\controlx^\idxi_1, %
                                      \ldots, %
                                      \controlx^\idxi_{\idxj_\idxi}} %
                                     {\outputsx_\idxi} %
                                     {\branchchs_\idxi} %
                    \end{array}} %
    \]
    derived from some rule
    \[
        \pdrule{\control}
               {\cha}
               {\ech}
               {\rew{\cha}}
               {\control'_1, \ldots, \control'_\altnumof} \in \rules \ .
    \]
    In this case, we apply the above rule to $\pdrun$ which means taking the node $\node$ labelled $\configc$ and replacing its ``hole'' child with $\altnumof$ new children.
    We need to rebuild the rest of the tree the from these nodes.
    These nodes have configurations
    $\config{\control'_1}{\stackw},
     \ldots,
     \config{\control'_\altnumof}{\stackw}$.
    These control states are distributed between
    $\controlx_1, \ldots, \controlx_\idxi$
    and
    $\controly_1, \ldots, \controly_\idxj$.
    Consider $\controly_1$ (the other
    $\controly_2, \ldots, \controly_\idxj$
    are identical).
    We have from the respective passed test that
    $\config{\controly_1}{\stackw}$
    has a run where the first popping of the $\topop{\pdaord}$ stack
    leads to configurations
    $\config{\controly^1_1}{\stackw},
     \ldots,
     \config{\controly^1_{\idxi_1}}{\stackw}$.
    We insert this run underneath the node corresponding to the $\controly_1$.
    Since
    $\controly^1_1, \ldots, \controly^1_{\idxi_1}$
    appear amongst
    $\control_1, \ldots, \control_\numof$
    we append the subtrees that appeared as the relevant children of the node labelled $\simulator{\configc}$
    to complete these branches.
    The remaining subtrees corresponding to
    $\control_1, \ldots, \control_\numof$
    are distributed amongst
    $\controlx^1_1, \ldots, \controlx^1_{\idxj_1},
     \ldots,
     \controlx^\idxi_1, \ldots, \controlx^\idxi_{\idxj_\idxi}$.
    Consider $\controlx_1$ (the others are identical arguments), we have a new child labelled by
    $\config{\simcontrolbr{\controlx_1}
                          {\controlx^1_1, \ldots, \controlx^1_{\idxj_1}}
                          {\outputsx_1}
                          {\branchchs_1}}
            {\ap{\topop{\pdaord}}{\stackw}}$.
    We take the subrees distributed to
    $\controlx^1_1, \ldots, \controlx^1_{\idxj_1}$
    as children of this new child to satisfy the requirements.

    The new frontier replaces
    $\simulator{\configc}$
    with
    \[
        \begin{array}{c} %
            \config{\simcontrolbr{\controlx_1} %
                                 {\controlx^1_1, \ldots, \controlx^1_{\idxj_1}} %
                                 {\outputsx_1} %
                                 {\branchchs_1}} %
                   {\ap{\topop{\pdaord}}{\stackw}}, %
            \\ %
            \ldots, %
            \\ %
            \config{\simcontrolbr{\controlx_\idxj} %
                                 {\controlx^\idxi_1, %
                                  \ldots, %
                                  \controlx^\idxi_{\idxj_\idxi}} %
                                 {\outputsx_\idxj} %
                                 {\branchchs_\idxj}} %
                   {\ap{\topop{\pdaord}}{\stackw}} %
        \end{array} %
    \]
    which satisfies all properties as needed.

\item
    $\simulator{\pdrun'} =
     \treeap{\simulator{\configc}}
            {\simulator{\pdrun^1},
             \ldots,
             \simulator{\pdrun^\idxi}}$
    via a rule
    \[
        \bigpdrulet{\simcontrolbr{\control}
                                 {\control_1, \ldots, \control_\numof}
                                 {\outputs}
                                 {\branchchs}}
                   {\cha}
                   {\begin{array}{c} %
                       \canpopautbr{\controly_1} %
                                   {\controly^1_1,
                                    \ldots,
                                    \controly^1_{\idxi_1}} %
                                   {\outputsy_1} \\ %
                       \cap \cdots \cap \\ %
                       \canpopautbr{\controly_\idxj} %
                                   {\controly^\idxj_1,
                                    \ldots,
                                    \controly^\idxj_{\idxi_\idxj}} %
                                   {\outputsy_\idxj} %
                    \end{array}} %
                   {\outputs' \cap \branchchs} %
                   {\rew{\cha}} %
                   {\begin{array}{c} %
                       \simcontrolbr{\control'} %
                                    {\control'_1, \ldots, \control'_\altnumof} %
                                    {\outputsx} %
                                    {\branchchs_0}, \\ %
                       \simcontrolbr{\controlx_1} %
                                    {\controlx^1_1, %
                                     \ldots, %
                                     \controlx^1_{\idxj_1}} %
                                    {\outputsx_1} %
                                    {\branchchs_1}, \\ %
                        \ldots, \\ %
                        \simcontrolbr{\controlx_\idxi} %
                                     {\controlx^\idxi_1, %
                                      \ldots, %
                                      \controlx^\idxi_{\idxj_\idxi}} %
                                     {\outputsx_\idxi} %
                                     {\branchchs_\idxi} %
                    \end{array}} %
    \]
    derived from some rule
    \[
        \pdrule{\control}
               {\cha}
               {\ech}
               {\push{\pdaord}}
               {\control'}
    \]
    In this case, we apply the above rule to $\pdrun$.
    This means replacing the node labelled $\simulator{\configc}$ with one labelled
    $\config{\control'}{\ap{\push{\pdaord}}{\stackw}}$.
    This new node has a new child node with the label
    \[
        \config{\simcontrolbr{\control'}
                             {\control'_1, \ldots, \control'_\altnumof}
                             {\outputsx}
                             {\branchchs_0}}
               {\ap{\topop{\pdaord}}{\stackw}} \ .
    \]
    We need to add $\altnumof$ children to this new ``hole'' node.

    These nodes have configurations
    $\config{\control'_1}{\stackw},
     \ldots,
     \config{\control'_\altnumof}{\stackw}$
    (since $\stackw = \ap{\pop{\pdaord}}{\ap{\push{\pdaord}}{\stackw}}$).
    These control states are distributed between
    $\controlx_1, \ldots, \controlx_\idxi$
    and
    $\controly_1, \ldots, \controly_\idxj$.
    Consider $\controly_1$ (the other
    $\controly_2, \ldots, \controly_\idxj$
    are identical).
    We have from the passed test that
    $\config{\controly_1}{\stackw}$
    has a run where the first popping of the $\topop{\pdaord}$ stack
    leads to configurations
    $\config{\controly^1_1}{\ap{\pop{\pdaord}}{\stackw}},
     \ldots,
     \config{\controly^1_{\idxi_1}}{\ap{\pop{\pdaord}}{\stackw}}$.
    We append this run tree as a child of the node corresponding to $\controly_1$.
    Since
    $\controly^1_1, \ldots, \controly^1_{\idxi_1}$
    appear amongst
    $\control_1, \ldots, \control_\numof$
    we append the relevant subtrees we had already constructed for these nodes to complete these branches with the required properties.

    Now consider $\controlx_1$ (the other cases are symmetric).  In this case we append a node labelled
    $\config{\simcontrolbr{\controlx_1}
                          {\controlx^1_1, \ldots, \controlx^1_{\idxj_1}}
                          {\outputsx_1}
                          {\branchchs_1}}
            {\ap{\topop{\pdaord}}{\stackw}}$
    as a child of the node corresponding to $\controlx_1$.
    Since
    $\controlx^1_1, \ldots, \controlx^1_{\idxj_1}$
    appear amongst
    $\control_1, \ldots, \control_\numof$
    we append the relevant subtrees we had already constructed for these nodes to complete these branches with the required properties.

    The new frontier replaces
    $\simulator{\configc}$
    with
    \[
        \config{\simcontrolbr{\control'}
                             {\control'_1, \ldots, \control'_\altnumof}
                             {\outputsx}
                             {\branchchs_0}}
               {\ap{\topop{\pdaord}}{\stackw}}
    \]
    and
    \[
        \begin{array}{c} %
            \config{\simcontrolbr{\controlx_1} %
                                 {\controlx^1_1, \ldots, \controlx^1_{\idxj_1}} %
                                 {\outputsx_1} %
                                 {\branchchs_1}} %
                   {\ap{\topop{\pdaord}}{\stackw}}, %
            \\ %
            \ldots, %
            \\ %
            \config{\simcontrolbr{\controlx_\idxj} %
                                 {\controlx^\idxi_1, %
                                  \ldots, %
                                  \controlx^\idxi_{\idxj_\idxi}} %
                                 {\outputsx_\idxi} %
                                 {\branchchs_\idxi}} %
                   {\ap{\topop{\pdaord}}{\stackw}} %
        \end{array} %
    \]
    which satisfies all the required properties.

\item
    $\simulator{\pdrun'} =
     \treeap{\simulator{\configc}}
            {\simulator{\pdrun^1},
             \ldots,
             \simulator{\pdrun^\idxi}}$
    via a rule
    \[
        \bigpdrulet{\simcontrolbr{\control}
                                 {\control_1, \ldots, \control_\numof}
                                 {\outputs}
                                 {\branchchs}}
                   {\cha}
                   {\begin{array}{c} %
                       \canpopautbr{\control'} %
                                   {\control'_1,\ldots,\control'_\altnumof} %
                                   {\outputsy}\ \cap \\ %
                       \canpopautbr{\controly_1} %
                                   {\controly^1_1,
                                    \ldots,
                                    \controly^1_{\idxi_1}} %
                                   {\outputsy_1} \\ %
                       \cap \cdots \cap \\ %
                       \canpopautbr{\controly_\idxj} %
                                   {\controly^\idxj_1,
                                    \ldots,
                                    \controly^\idxj_{\idxi_\idxj}} %
                                   {\outputsy_\idxj} %
                    \end{array}} %
                   {\outputs' \cap \branchchs} %
                   {\rew{\cha}} %
                   {\begin{array}{c} %
                       \simcontrolbr{\controlx_1} %
                                    {\controlx^1_1, %
                                     \ldots, %
                                     \controlx^1_{\idxj_1}} %
                                    {\outputsx_1} %
                                    {\branchchs_1}, \\ %
                        \ldots, \\ %
                        \simcontrolbr{\controlx_\idxi} %
                                     {\controlx^\idxi_1, %
                                      \ldots, %
                                      \controlx^\idxi_{\idxj_\idxi}} %
                                     {\outputsx_\idxi} %
                                     {\branchchs_\idxi} %
                    \end{array}} %
    \]
    derived from some rule
    \[
        \pdrule{\control}
               {\cha}
               {\ech}
               {\push{\pdaord}}
               {\control'}
    \]
    In this case, we again apply the above rule to $\pdrun$.
    This means replacing the node labelled $\simulator{\configc}$ with one labelled
    $\config{\control'}{\ap{\push{\pdaord}}{\stackw}}$.
    Since we know the test
    $\canpopautbr{\control'}
                 {\control'_1,\ldots,\control'_\altnumof}
                 {\controly}$
    passed we have a run popping the newly pushed stack to controls
    $\control'_1, \ldots, \control'_\altnumof$.
    We set this run tree as the only child of the node whose label we replaced.
    This new tree has $\altnumof$ leaves which we need to complete.

    These leaf nodes are completed using the same argument as the previous case.
    That is, they are labelled with configurations
    $\config{\control'_1}{\stackw},
     \ldots,
     \config{\control'_\altnumof}{\stackw}$.
    These control states are distributed between
    $\controlx_1, \ldots, \controlx_\idxi$
    and
    $\controly_1, \ldots, \controly_\idxj$.
    Consider $\controly_1$ (the other
    $\controly_2, \ldots, \controly_\idxj$
    are identical).
    We have from the passed test that
    $\config{\controly_1}{\stackw}$
    has a run where the first popping of the $\topop{\pdaord}$ stack
    leads to configurations
    $\config{\controly^1_1}{\ap{\pop{\pdaord}}{\stackw}},
     \ldots,
     \config{\controly^1_{\idxi_1}}{\ap{\pop{\pdaord}}{\stackw}}$.
    We append this run tree as a child of the node corresponding to $\controly_1$.
    Since
    $\controly^1_1, \ldots, \controly^1_{\idxi_1}$
    appear amongst
    $\control_1, \ldots, \control_\numof$
    we append the relevant subtrees we had already constructed for these nodes to complete these branches with the required properties.

    Now consider $\controlx_1$ (the other cases are symmetric).  In this case we append a node labelled
    $\config{\simcontrolbr{\controlx_1}
                          {\controlx^1_1, \ldots, \controlx^1_{\idxj_1}}
                          {\outputsx_1}
                          {\branchchs_1}}
            {\ap{\topop{\pdaord}}{\stackw}}$
    as a child of the node corresponding to $\controlx_1$.
    Since
    $\controlx^1_1, \ldots, \controlx^1_{\idxj_1}$
    appear amongst
    $\control_1, \ldots, \control_\numof$
    we append the relevant subtrees we had already constructed for these nodes to complete these branches with the required properties.

    The new frontier replaces
    $\simulator{\configc}$
    with
    \[
        \config{\simcontrolbr{\control'}
                             {\control'_1, \ldots, \control'_\altnumof}
                             {\outputsx}
                             {\branchchs_0}}
               {\ap{\topop{\pdaord}}{\stackw}}
    \]
    and
    \[
        \begin{array}{c} %
            \config{\simcontrolbr{\controlx_1} %
                                 {\controlx^1_1, \ldots, \controlx^1_{\idxj_1}} %
                                 {\outputsx_1} %
                                 {\branchchs_1}} %
                   {\ap{\topop{\pdaord}}{\stackw}}, %
            \\ %
            \ldots, %
            \\ %
            \config{\simcontrolbr{\controlx_\idxj} %
                                 {\controlx^\idxi_1, %
                                  \ldots, %
                                  \controlx^\idxi_{\idxj_\idxi}} %
                                 {\outputsx_\idxi} %
                                 {\branchchs_\idxi}} %
                   {\ap{\topop{\pdaord}}{\stackw}} %
        \end{array} %
    \]
    which satisfies all the required properties.

\item
    $\simulator{\pdrun'} =
     \treeap{\simulator{\configc}}
            {\config{\finalcontrol}{\stackw}}$.

    In this case $\configc$ has the form
    \[
        \config{\simcontrolbr{\control}
                             {\control'}
                             {\emptyset}
                             {\branchchs}}
               {\ap{\topop{\pdaord}}{\stackw}}
    \]
    and there is a rule
    \[
        \pdrule{\control}{\cha}{\ech}{\pop{\pdaord}}{\control'} \ .
    \]
    We can remove the hole from $\pdrun$ by applying this rule.
    That is, we remove the hole node, setting its parent to have its (only) child as its child.
    This is possible since by our conditions the child has the label
    $\config{\control'}{\ap{\pop{\pdaord}}{\stackw}}$.
    We remove $\simulator{\configc}$ from the frontier.
\end{itemize}
Thus, the frontier moves towards the leaves of the tree and finally is empty.
At this point we have an accepting run of $\pda$ as required.
To see that the run outputs enough of each character, one needs to observe that at each stage the tests and $\outputs$ component of the control state ensured at least one character output for each that appeared in some $\outputs'$ labelling a transition.
Then, for characters output along branches followed were reproduced faithfully.

    \end{proof}
}


\section{Conclusions}

We have shown, using a recent result by Zetzsche, that the downward closures of languages defined by HOPDA are computable.
We believe this to be a useful foundational result upon which new analyses may be based.
Our result already has several immediate consequences, including
\changed[mh]{
    separation by%
}
piecewise testability and asynchronous parameterised systems.

\changed[mh]{
    Regarding the complexity of the approach.
    We are unaware of any complexity bounds implied by Zetzsche's techniques.
    Due to the complexity of the reachability problem for HOPDA, the test automata may be a tower of exponentials of height $\pdaord$ for HOPDA of order $\pdaord$.
    These test automata are built into the system before proceeding to reduce to order $(\pdaord-1)$.
    Thus, we may reach a tower of exponentials of height $O(\pdaord^2)$.
}

A natural next step is to consider collapsible pushdown systems, which are equivalent to
recursion schemes (without the safety constraint).
However, it is not currently clear how to generalise our techniques due to the non-local behaviour introduced by collapse.
We may also try to adapt our techniques to a higher-order version of BS-automata~\cite{Bojanczyk:2010}, which may be used, e.g., to check boundedness of resource usage for higher-order programs.

    \paragraph{Acknowledgements}

    We thank Georg Zetzsche for keeping us up to date with his work, Jason Crampton for knowing about logarithms when they were most required, and Chris Broadbent for discussions.
This work was supported by the Engineering and Physical Sciences Research
Council [EP/K009907/1 and EP/M023974/1].

\end{document}